\pdfoutput=1
\documentclass[11pt,fleqn]{article}

\usepackage[utf8]{inputenc}
\usepackage[T1]{fontenc}

\usepackage[dvipsnames]{xcolor}
\usepackage{comment}
\usepackage{fullpage}
\usepackage[indent,parfill=0pt]{parskip}  
\usepackage{indentfirst}

\usepackage{amsmath,amsthm}
\usepackage{amssymb,amsfonts}
\usepackage{mathtools,thm-restate}

\usepackage{mathptmx}
\usepackage{stmaryrd}  
\usepackage[cal=pxtx,scr=boondox]{mathalpha}

\usepackage[style=alphabetic,natbib=true,sortcites=true,backref=true,maxnames=99,useprefix=true,maxalphanames=99]{biblatex}

\DefineBibliographyStrings{english}{%
  backrefpage = {$\Lsh$~p.},
  backrefpages = {$\Lsh$~pp.},
}
\usepackage[colorlinks=true,citecolor=Blue,linkcolor=Red,urlcolor=Magenta,linktoc=page,bookmarksnumbered=true]{hyperref}
\usepackage[nameinlink]{cleveref}

\usepackage[basic,sanserif,disableredefinitions]{complexity}
\newclass{\PCRP}{PCRP}

\usepackage{booktabs,multirow}
\usepackage{enumitem}
\usepackage{framed,ascmac}
\usepackage{subcaption}
\usepackage{threeparttable}

\usepackage{soul}
\sethlcolor{Yellow}

\usepackage{comment}
\usepackage{url}

\usepackage[colorinlistoftodos,prependcaption]{todonotes}

\usepackage{xspace}
\usepackage{bm}
\usepackage{mleftright}
\mleftright

\crefname{theorem}{Theorem}{Theorems}
\crefname{proposition}{Proposition}{Propositions}
\crefname{lemma}{Lemma}{Lemmas}
\crefname{claim}{Claim}{Claims}
\crefname{corollary}{Corollary}{Corollaries}
\crefname{remark}{Remark}{Remarks}
\crefname{observation}{Observation}{Observations}
\crefname{hypothesis}{Hypothesis}{Hypotheses}
\crefname{fact}{Fact}{Facts}
\crefname{definition}{Definition}{Definitions}
\crefname{problem}{Problem}{Problems}
\crefname{example}{Example}{Examples}

\crefname{appendix}{Appendix}{Appendices}
\crefname{section}{Section}{Sections}
\crefname{equation}{Eq.}{Eqs.}
\crefname{figure}{Figure}{Figures}
\crefname{table}{Table}{Tables}
\crefname{algorithm}{Algorithm}{Algorithms}

\usepackage[ruled]{algorithm}
\usepackage[commentColor=Blue]{algpseudocodex}

\renewcommand{\Return}{\textbf{return}\xspace}
\algrenewcommand\textproc{\textsl}

\renewcommand{\geq}{\geqslant}
\renewcommand{\leq}{\leqslant}
\renewcommand{\epsilon}{\varepsilon}
\renewcommand{\phi}{\varphi}
\renewcommand{\bar}{\overline}

\renewcommand{\tilde}{\widetilde}

\newcommand{\defi}[1]{\textbf{\emph{#1}}}

\newcommand{\prb}[1]{\textsc{#1}\xspace}
\newcommand{\nth}[1]{#1\textsuperscript{th}\xspace}

\newcommand{\sq}[1]{\vec{#1}}
\newcommand{\rv}[1]{\bm{#1}}
\renewcommand{\ul}[1]{\underline{#1}}

\newcommand{\reco}{\leftrightsquigarrow}
\newcommand{\defeq}{\coloneq}
\newcommand{\rme}{\mathrm{e}}
\newcommand{\sym}{\frakS}

\let\Pr\relax\DeclareMathOperator*{\Pr}{\mathbf{Pr}}
\let\E\relax\DeclareMathOperator*{\E}{\mathbf{E}}  
\DeclareMathOperator{\bigO}{\mathrm{O}}
\DeclareMathOperator{\dist}{dist}

\DeclareMathOperator{\val}{\mathsf{val}}
\DeclareMathOperator{\opt}{\mathsf{opt}}

\DeclareMathOperator{\mix}{\mathsf{mix}}
\let\poly\relax\DeclareMathOperator*{\poly}{\mathrm{poly}}
\let\polylog\relax\DeclareMathOperator*{\polylog}{\mathrm{polylog}}

\newcommand{\close}[3]{{#1 \underset{#3}{\approx} #2}}
\newcommand{\far}[3]{{#1 \underset{#3}{\not\approx} #2}}

\newcommand{\sss}{\mathrm{start}}
\newcommand{\ttt}{\mathrm{end}}
\newcommand{\zo}{\{0,1\}}
\newcommand{\Yes}{\textsc{Yes}\xspace}
\newcommand{\No}{\textsc{No}\xspace}
\newcommand{\evt}{\calE}
\newcommand{\RW}{\mathrm{RW}}

\newcommand{\V}{\calV}
\renewcommand{\W}{\calW}
\newcommand{\T}{\calT}
\newcommand{\Tstd}[1]{\T^{\mathrm{std}}_{#1}}
\newcommand{\Ttol}[1]{\T^{\mathrm{tol}}_{#1}}

\newcommand{\Ttuple}{\T_{\mathrm{tuple}}}
\newcommand{\Fset}{F_{\mathrm{set}}}
\newcommand{\Ftuple}{F_{\mathrm{tuple}}}

\newcommand{\e}{\eta}
\newcommand{\rnd}{r}
\newcommand{\Rnd}{\scrR}

\newcommand{\Cons}{{\mathsf{Cons}}}
\newcommand{\evtA}{\evt_{\mathrm{\ref{eq:IKW12:lem310:A}}}}
\newcommand{\evtB}{\evt_{\mathrm{\ref{eq:IKW12:lem310:B}}}}
\newcommand{\Sample}{\textsc{Sample}}
\DeclareMathOperator*{\PLR}{\mathsf{plurality}}

\DeclareMathOperator{\supp}{\mathsf{supp}}
\newcommand{\smol}{\epsilon}
\newcommand{\step}{\gamma}
\newcommand{\radi}{\delta}
\newcommand{\Cell}{\calA}
\newcommand{\cell}{a}

\newcommand{\twoCSPReconf}{\prb{2-CSP Reconfiguration}}
\newcommand{\MMtwoCSPReconf}{\prb{Maxmin 2-CSP Reconfiguration}}
\newcommand{\qCSPReconf}{\prb{$q$-CSP Reconfiguration}}
\newcommand{\MMqCSPReconf}{\prb{Maxmin $q$-CSP Reconfiguration}}

\newcommand{\MMCSPReconf}[1]{\prb{Maxmin #1-CSP Reconfiguration}}

\newcommand{\calA}{\mathcal{A}}

\newcommand{\calD}{\mathcal{D}}
\newcommand{\calE}{\mathcal{E}}

\newcommand{\calH}{\mathcal{H}}
\newcommand{\calI}{\mathcal{I}}

\newcommand{\calK}{\mathcal{K}}

\newcommand{\calM}{\mathcal{M}}

\newcommand{\calP}{\mathcal{P}}

\newcommand{\calS}{\mathcal{S}}
\newcommand{\calT}{\mathcal{T}}

\newcommand{\calV}{\mathcal{V}}
\newcommand{\calW}{\mathcal{W}}
\newcommand{\calX}{\mathcal{X}}
\newcommand{\calY}{\mathcal{Y}}

\newcommand{\bbN}{\mathbb{N}}

\newcommand{\bbR}{\mathbb{R}}

\newcommand{\scrI}{\mathscr{I}}

\newcommand{\scrR}{\mathscr{R}}

\newcommand{\frakS}{\mathfrak{S}}

\newtheorem{theorem}{Theorem}[section]
\newtheorem{proposition}[theorem]{Proposition}
\newtheorem{lemma}[theorem]{Lemma}
\newtheorem{claim}[theorem]{Claim}
\newtheorem{corollary}[theorem]{Corollary}
\newtheorem{observation}[theorem]{Observation}

\newtheorem{remark}[theorem]{Remark}
\theoremstyle{definition}
\newtheorem{definition}[theorem]{Definition}
\newtheorem{problem}[theorem]{Problem}

\numberwithin{equation}{section}
\hypersetup{bookmarksdepth=3}

\title{Optimal PSPACE-hardness of Approximating \\ $q$-CSP Reconfiguration}
\author{
Shuichi Hirahara \\
\small{National Institute of Informatics, Japan} \\
\small{\href{mailto:s_hirahara@nii.ac.jp}{\texttt{s\_hirahara@nii.ac.jp}}}
\and
Naoto Ohsaka \\
\small{CyberAgent, Inc., Japan} \\
\small{\href{mailto:ohsaka_naoto@cyberagent.co.jp}{\texttt{ohsaka\_naoto@cyberagent.co.jp}}}
}
\date{}

\begin{document}
\maketitle

\thispagestyle{empty}
\begin{abstract}In the \prb{Maxmin $q$-CSP Reconfiguration} problem, given a satisfiable $q$-CSP instance and a pair of its satisfying assignments, we are asked to transform one assignment into the other by repeatedly changing the value assigned to a single variable.
The objective is to find such a transformation that maximizes the minimum fraction of satisfied constraints along the transformation.
In this paper, we prove that for any $q \geq 2$ and $\varepsilon > 0$, \prb{Maxmin $q$-CSP Reconfiguration} is $\mathsf{PSPACE}$-hard to approximate within a factor of $\frac{1}{2^{q-1}}+\varepsilon$.
To complement this hardness result, we prove that a $\bigl(\frac{1}{2^{q-1}}-\varepsilon\bigr)$-factor approximation for \prb{Maxmin $q$-CSP Reconfiguration} is in $\mathsf{NP}$ in the perfect completeness case.
These results establish the optimal $\mathsf{PSPACE}$-hardness of approximating \prb{Maxmin $q$-CSP Reconfiguration} for every $q \geq 2$ under $\mathsf{NP} \neq \mathsf{PSPACE}$.
\end{abstract}
{\small\tableofcontents}

\clearpage

\section{Introduction}

\qCSPReconf is a canonical reconfiguration problem defined as follows:
Let $G$ be a satisfiable instance of the $q$-ary constraint satisfaction problem ($q$-CSP) 
with $n$ variables and alphabet size $\sigma$.
A sequence of assignments for $G$,
denoted by $\sq{f} = (f^{(1)}, \ldots, f^{(T)})$,
is called a \defi{reconfiguration sequence} if
every adjacent pair $f^{(t)}$ and $f^{(t+1)}$ differ in a single variable.
In the \qCSPReconf problem,
for a pair of satisfying assignments $f_\sss$ and $f_\ttt$ for $G$,
we are asked to decide if
there exists a reconfiguration sequence $\sq{f}$ from $f_\sss$ to $f_\ttt$
consisting only of satisfying assignments for $G$.
In other words, \qCSPReconf asks the $st$-connectivity question over the \defi{solution space} of $G$,
which is defined as the subgraph of the Hamming graph $H(n,\sigma)$ induced by all satisfying assignments for $G$.
Special cases of \qCSPReconf include the following reconfiguration problems:

\begin{itemize}
\item
In the \prb{$k$-SAT Reconfiguration} problem \cite{gopalan2009connectivity},
for a satisfiable $k$-CNF formula $\phi$ and a pair of its satisfying assignments,
we seek a path from one assignment to the other in the Boolean hypercube that
passes only through satisfying assignments for $\phi$.
This problem is in $\cP$ if $k \leq 2$ and $\PSPACE$-complete if $k \geq 3$ \cite{gopalan2009connectivity}.
Studying \prb{$k$-SAT Reconfiguration} and its relatives was originally motivated by applications to analyzing the structure of the solution space of Boolean formulas \cite{gopalan2009connectivity}.

\item
In the \prb{$k$-Coloring Reconfiguration} problem \cite{cereceda2007mixing,cereceda2008connectedness},
for a $k$-colorable graph $G$ and a pair of its proper $k$-colorings,
we wish to transform one $k$-coloring into the other by recoloring a single vertex at a time.
This problem is in $\cP$ if $k \leq 3$ \cite{cereceda2011finding} and $\PSPACE$-complete if $k \geq 4$ \cite{bonsma2009finding}.
The solution space of \prb{$k$-Coloring Reconfiguration} is related to the Glauber dynamics \cite{dyer2006randomly,jerrum1995very,molloy2004glauber};
see also \cite[\S5]{heuvel2013complexity}.
\end{itemize}

\noindent
In general, \qCSPReconf with alphabet size $\sigma$
is in $\cP$ if $q \leq 2$ and $\sigma \leq 2$ \cite{gopalan2009connectivity,hatanaka2018complexity} while
it is $\PSPACE$-complete if
($q=2$ and $\sigma=3$) \cite{hatanaka2018complexity} or
($q=3$ and $\sigma=2$) \cite{gopalan2009connectivity,hearn2005pspace,hearn2009games}.
See \cref{sec:related} for related work on other reconfiguration problems.

In this paper, we study the \emph{approximability} of \qCSPReconf.
Recently, approximability of reconfiguration problems has been studied from both hardness and algorithmic perspectives 
\cite{guruswami2025inapproximability,gur20263,hirahara2024probabilistically,hirahara2024optimal,hirahara2025asymptotically,hirahara2025asymptoticallya,hoang2026inapproximability,ohsaka2022reconfiguration,ohsaka2023gap,ohsaka2024gap,ohsaka2024alphabet,ohsaka2024tight,ohsaka2025approximate,ohsaka2025yet} (see also \cref{sec:related}).
The approximate version of \qCSPReconf is called \MMqCSPReconf \cite{ohsaka2023gap,ito2011complexity}.
Given a satisfiable $q$-CSP instance $G$ and a pair of its satisfying assignments $f_\sss$ and $f_\ttt$,
this problem asks for a reconfiguration sequence $\sq{f}$ from $f_\sss$ to $f_\ttt$ consisting of any (not necessarily satisfying) assignments for $G$.
The objective is to maximize the minimum fraction of satisfied constraints of $G$,
where the minimum is taken over all assignments appearing in $\sq{f}$.

\begin{itembox}[l]{\textbf{\MMqCSPReconf}}
\begin{tabular}{ll}
    \textbf{Input:}
    & a satisfiable $q$-CSP instance $G$ and a pair of its satisfying assignments $f_\sss$ and $f_\ttt$.
    \\
    \textbf{Output:}
    & a reconfiguration sequence $\sq{f}$ from $f_\sss$ to $f_\ttt$.
    \\
    \textbf{Goal:}
    & maximize the minimum fraction of satisfied constraints of $G$ over all assignments in $\sq{f}$.
\end{tabular}
\end{itembox}

\noindent
Solving this problem approximately,
we may obtain a reasonable reconfiguration sequence consisting of almost-satisfying assignments,
which may help us handle \No-instances of \qCSPReconf.
Note that an ``$\NP$ analogue'' of \MMqCSPReconf is the \prb{Max $q$-CSP} problem.

We review known results on the complexity of \MMtwoCSPReconf, which has been studied intensively.
\MMtwoCSPReconf is $\PSPACE$-hard to solve exactly,
which follows from the $\PSPACE$-hardness of \twoCSPReconf, e.g., \cite{bonsma2009finding,hatanaka2018complexity}.
The Probabilistically Checkable Reconfiguration Proof (PCRP) theorem \cite{hirahara2024probabilistically,guruswami2025inapproximability} implies that 
\MMtwoCSPReconf is $\PSPACE$-hard to approximate within some constant factor.
Since the PCRP theorem does not provide an explicit hardness factor,
gap amplifiability for \MMtwoCSPReconf has been investigated.
Note that, unlike the parallel repetition theorem for \prb{Max 2-CSP} \cite{raz1998parallel},
parallel repetition does not reduce the soundness error of \MMtwoCSPReconf \cite{ohsaka2025approximate}.
\citet{ohsaka2024gap} proved that a $0.9942$-factor approximation is $\PSPACE$-hard, and
a $(0.75+\epsilon)$-factor approximation is $\NP$-hard.
\citet{ohsaka2025approximate} developed a $(0.25-\epsilon)$-factor approximation algorithm for sparse 2-CSP instances.
\citet{guruswami2025inapproximability} proved
the $\NP$-hardness of a $(0.5+\epsilon)$-factor approximation, improving upon \cite{ohsaka2024gap}, and
developed a $(0.5-\epsilon)$-factor approximation algorithm, improving upon \cite{ohsaka2025approximate}.
These results are tight with respect to $\NP$-hardness.
Very recently, \citet{gur20263} proved
the $\PSPACE$-hardness of a $(0.9+\epsilon)$-factor approximation.
These hardness results leave open the intriguing possibility that
a $0.51$-factor approximation could be in $\NP$.
Indeed, \citet{guruswami2025inapproximability} posed determining the optimal $\PSPACE$-hardness of approximating \MMtwoCSPReconf as an open problem.\footnote{
Specifically, \cite{guruswami2025inapproximability} noted the following:
``the tight $\PSPACE$-hardness threshold for GapMaxMin-2-$\mathsf{CSP}_q$ remains open.
[\,\dots]
a negative answer---by placing the gap version in a complexity class believed to be a strict subset of $\PSPACE$ (e.g., $\Sigma_2^P$)---would also be very interesting (indeed, arguably more so than a $\PSPACE$-hardness result).''
}

\subsection{Our Results}

In this paper, we prove that for each arity $q \geq 2$,
\MMqCSPReconf is $\PSPACE$-hard to approximate within a factor of $\frac{1}{2^{q-1}}+\epsilon$.

\begin{theorem}[informal; see \cref{thm:PSPACE}]
\label{thm:intro:PSPACE}
For any integer $q \geq 2$ and any real $\epsilon > 0$,
there exists a positive integer $\sigma$ such that it is $\PSPACE$-hard,
given a satisfiable $q$-CSP instance $G$ with alphabet size $\sigma$ and a pair of its satisfying assignments $f_\sss$ and $f_\ttt$,
to distinguish between the following two cases:
\begin{description}
    \item[(Completeness)] 
        There exists a reconfiguration sequence from $f_\sss$ to $f_\ttt$
        consisting of satisfying assignments for $G$.\footnote{
        This is a \Yes-instance of \qCSPReconf.
        }
    \item[(Soundness)] 
        Every reconfiguration sequence from $f_\sss$ to $f_\ttt$
        contains an assignment that
        violates more than a $\left(1-\frac{1}{2^{q-1}}-\epsilon\right)$-fraction of the constraints of $G$.
\end{description}
In particular, for any integer $q \geq 2$ and any small real $\epsilon > 0$,
\MMqCSPReconf is $\PSPACE$-hard to approximate within a factor of $\frac{1}{2^{q-1}}+\epsilon$.
Moreover, the same hardness result holds even if $G$ is regular
(i.e., each variable appears in the same number of constraints).
\end{theorem}

\noindent
\cref{thm:intro:PSPACE} identifies an \emph{approximation threshold} of $0.5$ for \MMtwoCSPReconf; i.e.,
achieving an approximation factor above $0.5$ is $\PSPACE$-hard,
whereas below $0.5$ is in $\cP$ \cite{guruswami2025inapproximability}.
This resolves the open problem posed by \cite{guruswami2025inapproximability} and
rules out the possibility that a $0.51$-factor approximation could be in $\NP$ (unless $\NP = \PSPACE$).
\cref{thm:intro:PSPACE} also improves
the $\PSPACE$-hardness of approximation for \MMqCSPReconf for every $q \geq 2$
    due to \cite{gur20263} and
the $\NP$-hardness of a $(0.5+\epsilon)$-factor approximation for \MMCSPReconf{3}
    due to \cite{guruswami2025inapproximability}.
See \cref{tab:summary} for a comparison of \cref{thm:intro:PSPACE} with the existing results.

\begin{table}[t]
\centering
\begin{threeparttable}
    \caption{
        Summary of approximability results for \MMqCSPReconf for $2 \leq q \leq 5$, where $\epsilon > 0$ is any small real, and ``regular'' means that each variable appears in the same number of constraints.
        The best results for $\PSPACE$-hardness, $\NP$-membership, $\NP$-hardness, and $\cP$-membership are \hl{highlighted}.
    }
    \label{tab:summary}
    \setlength{\tabcolsep}{5pt} 
    \footnotesize
    \begin{tabular}{c|ccc|ccc}
    \toprule
    & \multicolumn{3}{c|}{(this paper)} & \multicolumn{3}{c}{(existing results)} \\
    arity $q$ & $\PSPACE$-hard & $\NP$ & $\cP$ (regular) &
    $\PSPACE$-hard & $\NP$-hard & $\cP$ \\
    \midrule
    2 &
        \hl{$0.5+\epsilon$}\phantom{000} &
        $0.5-\epsilon$\phantom{000} &
        $0.5-\epsilon$\phantom{000} &
        $0.9+\epsilon$\phantom{0} \cite{gur20263}\phantom{ \tnote{$\dagger$}} &
        $0.5+\epsilon$ \cite{guruswami2025inapproximability} &
        \hl{$0.5-\epsilon$} \cite{guruswami2025inapproximability} \\
    3 &
        \hl{$0.25+\epsilon$}\phantom{00} &
        \hl{$0.25-\epsilon$}\phantom{00} &
        $0.25-\epsilon$\phantom{00} &
        $0.9+\epsilon$\phantom{0} \cite{gur20263}\phantom{ \tnote{$\dagger$}} &
        $0.5+\epsilon$ \cite{guruswami2025inapproximability} &
        --- \\
    4 &
        \hl{$0.125+\epsilon$}\phantom{0} &
        \hl{$0.125-\epsilon$}\phantom{0} &
        $0.125-\epsilon$\phantom{0} &
        $0.81+\epsilon$ \cite{gur20263} \tnote{$\dagger$} &
        \hl{$\epsilon$} \cite{ohsaka2024gap} &
        --- \\
    5 &
        \hl{$0.0625+\epsilon$} &
        \hl{$0.0625-\epsilon$} &
        $0.0625-\epsilon$ &
        $0.5+\epsilon$\phantom{0} \cite{gur20263}\phantom{ \tnote{$\dagger$}} &
        \hl{$\epsilon$} \cite{ohsaka2024gap} & 
        --- \\
    \bottomrule
    \end{tabular}
    \footnotesize
    \begin{tablenotes}
        \item[$\dagger$] The $\PSPACE$-hardness of a $(0.81+\epsilon)$-factor approximation for $q=4$ follows by applying standard sequential repetition to the $\PSPACE$-hardness of a $(0.9+\epsilon)$-factor approximation for the case of $q=2$.
    \end{tablenotes}
\end{threeparttable}
\end{table}

To complement \cref{thm:intro:PSPACE}, we prove that for each arity $q \geq 2$,
a $\left(\frac{1}{2^{q-1}}-\epsilon\right)$-factor approximation for \MMqCSPReconf is in $\NP$
in the perfect completeness case.

\begin{theorem}[informal; see \cref{thm:NP}]
\label{thm:intro:NP}
    For a satisfiable $q$-CSP instance $G$ with $n$ variables and alphabet size $\sigma$ and
    a pair of its satisfying assignments $f_\sss$ and $f_\ttt$,
    suppose that there exists a reconfiguration sequence from $f_\sss$ to $f_\ttt$ consisting of satisfying assignments for $G$.
    Then, for any real $\epsilon > 0$,
    there exists an $n^{\bigO_{q,\sigma.\epsilon}(1)}$-length reconfiguration sequence from $f_\sss$ to $f_\ttt$ such that
    every assignment satisfies at least a $\left(\frac{1}{2^{q-1}}-\epsilon\right)$-fraction of the constraints of $G$.
    In particular, for any positive integers $q \geq 2$ and $\sigma$, and any small real $\epsilon > 0$,
    a $\left(\frac{1}{2^{q-1}}-\epsilon\right)$-factor approximation for \MMqCSPReconf
    is in $\NP$ in the perfect completeness case.
\end{theorem}

\noindent
\cref{thm:intro:PSPACE,thm:intro:NP} establish the optimal $\PSPACE$-hardness of approximating \MMqCSPReconf for every arity $q \geq 2$ assuming $\NP \neq \PSPACE$.

Combining \cref{thm:intro:PSPACE,thm:intro:NP} and \cite{ohsaka2024gap,guruswami2025inapproximability},
we obtain the following threshold result for the approximability of \MMqCSPReconf.
For any reals $c,s$ with $0 < s \leq c \leq 1$,
\prb{Gap$_{c,s}$ \qCSPReconf} is defined as a promise problem that asks whether the optimal value of \MMqCSPReconf is at least $c$ or less than $s$.

\begin{corollary}[from \cref{thm:intro:PSPACE,thm:intro:NP} and \cite{ohsaka2024gap,guruswami2025inapproximability}]
\label{cor:intro}
Let $q \geq 4$ be any integer and $s \in (0,1]$ be any real.
For any sufficiently large alphabet size $\sigma$, the following hold:
\begin{itemize}
    \item If $s > \frac{1}{2}$, then
        \prb{Gap$_{1,s}$ 2-CSP Reconfiguration} is $\PSPACE$-complete.
    \item If $s < \frac{1}{2}$, then
        \prb{Gap$_{1,s}$ 2-CSP Reconfiguration} is in $\cP$ \cite{guruswami2025inapproximability}.
    \item If $s > \frac{1}{4}$, then
        \prb{Gap$_{1,s}$ 3-CSP Reconfiguration} is $\PSPACE$-complete.
    \item If $s < \frac{1}{4}$, then
        \prb{Gap$_{1,s}$ 3-CSP Reconfiguration} is in $\NP$.
    \item If $s > \frac{1}{2^{q-1}}$, then
        \prb{Gap$_{1,s}$ $q$-CSP Reconfiguration} is $\PSPACE$-complete.
    \item If $s < \frac{1}{2^{q-1}}$, then
        \prb{Gap$_{1,s}$ $q$-CSP Reconfiguration} is $\NP$-complete,
        where $\NP$-hardness is shown in \cite{ohsaka2024gap}.
\end{itemize}
\end{corollary}

\noindent
\cref{cor:intro} reveals an intriguing complexity-theoretic contrast between the regimes of $q=2$ and $q \geq 4$.
For \MMtwoCSPReconf,
the approximation threshold $\frac{1}{2}$ separates $\cP$-membership from $\PSPACE$-completeness.
By contrast, for \MMCSPReconf{4},
the approximation threshold $\frac{1}{8}$ separates $\NP$-completeness from $\PSPACE$-completeness.
To the best of our knowledge,
this is the first $\PSPACE$-complete reconfiguration problem
whose approximate version is $\NP$-complete
(rather than lying in $\cP$, as in the case of $q=2$).
See also \cref{sec:overview:perspective} for the discussion on the approximability of \MMCSPReconf{3}.

We finally develop a deterministic $\left(\frac{1}{2^{q-1}}-\epsilon\right)$-factor approximation algorithm for \MMqCSPReconf on \defi{regular} instances;
i.e., each variable appears in the same number of constraints.

\begin{theorem}[informal; see \cref{thm:regular}]
\label{thm:intro:regular}
    For a satisfiable regular $q$-CSP instance $G$,
    a pair of its satisfying assignments $f_\sss$ and $f_\ttt$, and
    any real $\epsilon > 0$,
    there exists a polynomial-length reconfiguration sequence $\sq{f}$ from $f_\sss$ to $f_\ttt$ such that
    every assignment satisfies at least a $\left(\frac{1}{2^{q-1}}-\epsilon\right)$-fraction of the constraints of $G$.
    Moreover, such $\sq{f}$ can be found in deterministic polynomial time.
    In particular,
    for any integer $q \geq 2$ and any small real $\epsilon > 0$,
    there exists a deterministic $\left(\frac{1}{2^{q-1}}-\epsilon\right)$-factor approximation algorithm
    for \MMqCSPReconf on regular instances.
\end{theorem}

\subsection{Organization}
The rest of this paper is organized as follows.
In \cref{sec:overview},
    we present an overview of the proofs of \cref{thm:intro:PSPACE,thm:intro:NP}.
In \cref{sec:related},
    we review related work on reconfiguration problems and direct product testing.
In \cref{sec:pre},
    we formally define \MMqCSPReconf and its gap version.
In \cref{sec:TDP},
    we introduce and analyze the tolerant $q$-query direct product tester.
In \cref{sec:PSPACE},
    we prove the $\PSPACE$-hardness of a $\left(\frac{1}{2^{q-1}}+\epsilon\right)$-factor approximation.
In \cref{sec:NP},
    we prove the $\NP$-membership of a $\left(\frac{1}{2^{q-1}}-\epsilon\right)$-factor approximation.
In \cref{sec:regular},
    we develop a deterministic $\left(\frac{1}{2^{q-1}}-\epsilon\right)$-factor approximation algorithm for regular instances.
Some technical proofs are deferred to \cref{app:IKW12,app:PSPACE}.

\section{Proof Overview}
\label{sec:overview}

\subsection{\texorpdfstring{%
$\PSPACE$-hardness of $\left(\frac{1}{2^{q-1}}+\epsilon\right)$-factor Approximation (\cref{sec:TDP,sec:PSPACE})
}{%
PSPACE-hardness of (1/2\textasciicircum(q-1)+ε)-factor Approximation (Sections \ref{sec:TDP} and \ref{sec:PSPACE})
}}
\label{sec:overview:PSPACE}

First, we outline the proof of \cref{thm:intro:PSPACE},
i.e., $\PSPACE$-hardness of a $\left(\frac{1}{2^{q-1}}+\epsilon\right)$-factor approximation for \MMqCSPReconf.
Hereafter, a \defi{$q$-CSP instance} is specified by a quadruple $G = (V,E,\Sigma,\Psi)$ such that
$(V,E)$ is a $q$-uniform hypergraph called the \defi{underlying hypergraph},
$\Sigma$ is a finite set called the \defi{alphabet}, and
$\Psi = (\psi_e)_{e \in E}$
is a collection of \defi{$q$-ary constraints} $\psi_e \colon \Sigma^e \to \zo$ associated with hyperedges $e$ in $E$.
For an assignment $f \colon V \to \Sigma$,
the \defi{value} $\val_G(f)$ is defined as the fraction of constraints satisfied by $f$.
For a reconfiguration sequence $\sq{f} = (f^{(1)}, \ldots, f^{(T)})$,
the \defi{value} $\val_G(\sq{f})$ is defined as the minimum value of $\val_G(f^{(t)})$ over all assignments $f^{(t)}$ in $\sq{f}$.
For a pair of satisfying assignments $f_\sss, f_\ttt \colon V \to \Sigma$,
the \defi{optimal value} $\opt_G(f_\sss \reco f_\ttt)$
is defined as the maximum value of $\val_G(\sq{f})$ over all possible reconfiguration sequences $\sq{f}$ from $f_\sss$ to $f_\ttt$.
For reals $c,s$ with $0 < s \leq c \leq 1$,
\prb{Gap$_{c,s}$ \qCSPReconf} is defined as a promise problem that asks whether
$\opt_G(f_\sss \reco f_\ttt) \geq c$ or $\opt_G(f_\sss \reco f_\ttt) < s$.
See \cref{sec:pre} for the formal definition.

The proof of \cref{thm:intro:PSPACE} is based on
the following gap-amplifying reduction from \MMtwoCSPReconf to \MMqCSPReconf.

\begin{lemma}[informal; see \cref{lem:PSPACE}]
\label{lem:intro:PSPACE}
For any reals $s \in (0,1)$ and $\epsilon > 0$, and
any positive integers $q \geq 2$ and $\sigma$,
there exists a positive integer $k$ and
a polynomial-time reduction from
\prb{Gap$_{1,s}$ \twoCSPReconf} with alphabet size $\sigma$ to 
\prb{Gap$_{1,\frac{1}{2^{q-1}}+\epsilon}$ \qCSPReconf} with alphabet size $\sigma^{2k}$.
\end{lemma}

\noindent
By the PCRP theorem \cite{hirahara2024probabilistically,guruswami2025inapproximability},
\prb{Gap$_{1,s}$ \twoCSPReconf} with alphabet size $\bigO(1)$ is $\PSPACE$-hard for some real $s \in (0,1)$.
Combining this with \cref{lem:intro:PSPACE},
we obtain that \prb{Gap$_{1,\frac{1}{2^{q-1}}+\epsilon}$ \qCSPReconf} is $\PSPACE$-hard for any real $\epsilon > 0$,
which implies \cref{thm:intro:PSPACE}.
In the remainder of this subsection, we outline the proof of \cref{lem:intro:PSPACE}.

For the sake of simplicity, we first show the imperfect-completeness version in the case of $q=2$.
\begin{lemma}
\label{lem:intro:2CSP}
For any reals $s \in (0,1)$ and $\epsilon > 0$, and
any positive integer $\sigma$,
there exists a positive integer $k$ and
a polynomial-time reduction from
\prb{Gap$_{1,s}$ \twoCSPReconf} with alphabet size $\sigma$ to 
\prb{Gap$_{1-o(1),\frac{1}{2}+\epsilon}$ \twoCSPReconf} with alphabet size $\sigma^{2k}$.
\end{lemma}

\begin{remark}
We note that the $\PSPACE$-hardness result of \cite{gur20263} is based on \cite[Corollary~4]{guruswami2025inapproximability},
which does not inherently derive \cref{thm:intro:PSPACE} in the case of $q=2$.
Specifically, \cite[Corollary~4]{guruswami2025inapproximability} states that
if there exists a $q$-query Probabilistically Checkable Proof of Proximity (PCPP) \cite{ben-sasson2006robust,dinur2006assignment} with
randomness $\bigO(\log n)$,
soundness $s < 1$, and
alphabet size $\bigO_s(1)$,
then \MMtwoCSPReconf is $\PSPACE$-hard to approximate within a factor of $1 - \frac{1-s}{q+1}$.
Since any such PCPP must have query complexity $q \geq 2$ (under $\cP \neq \NP$),\footnote{
If such a $1$-query PCPP existed,
then there exists a $1$-query PCP for $\NP$ with randomness $\bigO(\log n)$ and alphabet size $\bigO_s(1)$.
Calculating its maximum acceptance probability in polynomial time, one would obtain $\cP = \NP$.
}
even the best possible PCPP would yield only the $\frac{2}{3}$-factor hardness.
To obtain the stronger $\left(\frac{1}{2}+\epsilon\right)$-factor hardness,
we develop a gap-amplifying reduction different from the one in \cite{guruswami2025inapproximability}.
\end{remark}

\subsubsection{\texorpdfstring{%
Proof Overview of \cref{lem:intro:2CSP}
}{%
Proof Overview of Lemma \ref{lem:intro:2CSP}
}}
\label{sec:overview:PSPACE:2CSP}

Our general strategy for proving \cref{lem:intro:2CSP}
is to encode assignments for \MMtwoCSPReconf using the direct product function.
For a function $f \colon V \to \Sigma$, its \defi{$k$-wise direct product} is defined as
a $k$-set function $f^k \colon \binom{V}{k} \to \Sigma^k$ such that
\begin{align}
    f^k\bigl( \{x_1, \ldots, x_k\} \bigr)
    \defeq \bigl( f(x_1), \ldots, f(x_k) \bigr)
    \text{ for each set } \{x_1, \ldots, x_k\} \in \tbinom{V}{k}.
\end{align}
\citet{goldreich2000combinatorial} introduced the problem of direct product testing,
which has found several applications to hardness of approximation \cite{dinur2006assignment,dinur2007pcp,dinur2011derandomized,impagliazzo2012new,bafna2025quasi} and to hardness amplification.
We first recapitulate the canonical 2-query direct product tester \cite{dinur2006assignment,dinur2008locally,impagliazzo2012new}.
Let $F \colon \binom{V}{k} \to \Sigma^k$ be a $k$-set function that we wish to test for closeness to the direct product function.
The \defi{2-query direct product tester} $\Tstd{2}$
first samples $\rv{I}_1 \in \binom{V}{\ell}$ and $\rv{X}_1, \rv{X}_2 \in \binom{V}{k}$ conditioned on
$\rv{X}_1 \cap \rv{X}_2 \supset \rv{I}_1$,\footnote{
Throughout this paper, we write random variables in boldface.
}
where $\ell \leq k$ is the \defi{intersection parameter}.
Then, $\Tstd{2}$ accepts if
$F(\rv{X}_1)$ and $F(\rv{X}_2)$ agree on the intersection $\rv{I}_1$; namely,
$F(\rv{X}_1)|_{\rv{I}_1} = F(\rv{X}_2)|_{\rv{I}_1}$.
\citet[Theorem~1.3]{dinur2008locally} proved that
if $\Tstd{2}$ with $\ell = \Theta(\sqrt{k})$ accepts $F$ with probability $p \geq k^{-\Omega(1)}$,
then there exists a function $f \colon V \to \Sigma$ such that
$F(\rv{X})$ agrees with $f^k(\rv{X})$ on all but a $k^{-\Omega(1)}$-fraction of coordinates with probability $p(1-o(1))$; namely,
\begin{align}
\label{eq:intro:2CSP:DG08}
    \Pr_{\rv{X} \in \binom{V}{k}}\left[
        \close{F(\rv{X})}{f^k(\rv{X})}{k^{-\Omega(1)}}
    \right] \geq p(1-o(1)),
\end{align}
where ``$\close{\alpha}{\beta}{\delta}$'' means that $\alpha$ and $\beta$ are $\delta$-close.

We now describe a gap-amplifying reduction from \MMtwoCSPReconf to itself,
which is based on the reduction from \prb{Max 2-CSP} to itself due to \cite[Theorem~1.4]{impagliazzo2012new}.
Let $s \in (0,1)$ and $\epsilon > 0$ be any reals, and $\sigma$ be any positive integer.
Let $k$ be a sufficiently large integer.
Let $(G,f_\sss,f_\ttt)$ be an instance of \prb{Gap$_{1,s}$ \twoCSPReconf}, where
$G = (V,E,\Sigma,\Psi)$ is a satisfiable 2-CSP instance with $|\Sigma| = \sigma$, and
$f_\sss, f_\ttt \colon V \to \Sigma$ are a pair of its satisfying assignments.
Our reduction is described in the language of PCP verifiers.
Specifically, we design a 2-query verifier along with a pair of its accepting proofs,
which define a new instance of \MMtwoCSPReconf.
The \defi{proof} is represented by 
a $k$-set function $F \colon \binom{E}{k} \to (\Sigma^2)^k$; namely,
for an edge set $Z \in \binom{E}{k}$,
a pair of symbols in $\Sigma^2$ are assigned to the endpoints of each edge in $Z$.
Since the domains of the assignment $f \colon V \to \Sigma$ and the proof $F \colon \binom{E}{k} \to (\Sigma^2)^k$ are different (i.e., $V$ and $E$),
we introduce the \defi{induced edge function} of $f$,
which is a function $f_E \colon E \to \Sigma^2$ defined as
\begin{align}
    f_E(xy) \defeq \bigl( f(x), f(y) \bigr) \text{ for each edge } xy \in E.
\end{align}
Intuitively, $F$ is supposed to be
the $k$-wise direct product of the induced edge function of some assignment $f \colon V \to \Sigma$.
Specifically, we consider the function $f_E^k \colon \binom{E}{k} \to (\Sigma^2)^k$ such that
\begin{align}
    f_E^k\bigl(\{e_1, \ldots, e_k\}\bigr)
    \defeq \bigl( f_E(e_1), \ldots, f_E(e_k) \bigr)
    \text{ for each set } \{e_1, \ldots, e_k\} \in \tbinom{E}{k}.
\end{align}

Consider the 2-query verifier $\V$ for \MMtwoCSPReconf,
which has oracle access to $F \colon \binom{E}{k} \to (\Sigma^2)^k$ and performs the following two tests:

\begin{description}
    \item[(Direct product test)]
    This test aims to determine whether
    $F$ is close to $f_E^k$ for some assignment $f \colon V \to \Sigma$.
    For this purpose,
    we run $\Tstd{2}$ on the following \defi{randomized oracle} $C^F \colon \binom{V}{k} \to \Sigma^k$
    \cite{impagliazzo2012new}:
    For a vertex set $X \in \binom{V}{k}$,
    the oracle samples $k$ random edges $\rv{Z} \in \binom{E}{k}$ incident to the vertices in $X$ and returns $F(\rv{Z})|_X$.
    If $\Tstd{2}$ rejects, then we immediately reject;
    otherwise, we advance to the satisfiability test.
    This test makes two queries $\rv{Z}_1, \rv{Z}_2 \in \binom{E}{k}$ to $F$.
    
    \item[(Satisfiability test)]
    This test aims to simulate $\Omega(k)$ runs of the 2-CSP verifier for $G$.
    Since $F$ passed the direct product test,
    the query results obtained in the direct product test (i.e., $F(\rv{Z}_1)$ and $F(\rv{Z}_2)$)
    are expected to contain consistent information about assignments to $\Omega(k)$ edges of $\rv{Z}_1 \cup \rv{Z}_2$.
    One can thus verify multiple constraints without any additional queries.
    Specifically, for each $i \in [2]$,
    we check whether every edge $\rv{e}$ in $\rv{Z}_i$
    is satisfied by the corresponding assignment $F(\rv{Z}_i)|_{\rv{e}}$.
\end{description}

\noindent
The starting and ending proofs $F_\sss, F_\ttt \colon \binom{E}{k} \to (\Sigma^2)^k$
are defined as the direct product of the induced edge functions of $f_\sss$ and $f_\ttt$, respectively; namely,
$F_\sss \defeq (f_\sss)_E^k$ and $F_\ttt \defeq (f_\ttt)_E^k$.
This defines a new instance $(\V,F_\sss,F_\ttt)$ of \MMtwoCSPReconf, completing the description of the reduction.
For a function $F$,
the \defi{value} $\val_{\V}(F)$ is defined as the probability that $\V$ accepts $F$.
For a reconfiguration sequence $\sq{F} = (F^{(1)}, \ldots, F^{(T)})$,
the \defi{value} $\val_{\V}(\sq{F})$ is defined as the minimum value of $\val_{\V}(F^{(t)})$ over all functions $F^{(t)}$ in $\sq{F}$.
The \defi{optimal value} $\opt_{\V}(F_\sss \reco F_\ttt)$
is defined as the maximum value of $\val_{\V}(\sq{F})$ over all possible reconfiguration sequences $\sq{F}$ from $F_\sss$ to $F_\ttt$.

We shall prove the following properties.
\begin{description}
    \item[(Completeness)] 
        If $\opt_G(f_\sss \reco f_\ttt) = 1$,
        then $\opt_{\V}(F_\sss \reco F_\ttt) \geq 1-o(1)$.
    \item[(Soundness)] 
        If $\opt_G(f_\sss \reco f_\ttt) < s$,
        then $\opt_{\V}(F_\sss \reco F_\ttt) < \frac{1}{2} + \epsilon$.
\end{description}

The completeness is almost immediate from the definition of $\V$,
and thus we focus on proving the soundness.
Assume that $\opt_G(f_\sss \reco f_\ttt) < s$.
Let $\sq{F} = (F^{(1)}, \ldots, F^{(T)})$ be any reconfiguration sequence from $F_\sss$ to $F_\ttt$.
We would like to show that $\val_{\V}(\sq{F}) < \frac{1}{2} + \epsilon$.
Suppose first that $\Tstd{2}$ accepts the randomized oracle $C^{F^{(t)}}$
with probability at most $\frac{1}{2} + \epsilon$ for some $t \in [T]$.
By definition, $\V$ also accepts $F^{(t)}$ with probability at most $\frac{1}{2} + \epsilon$; i.e.,
$\val_{\V}(\sq{F}) < \frac{1}{2} + \epsilon$, as desired.
Hereafter, we can assume that
$\Tstd{2}$ accepts $C^{F^{(t)}}$
with probability at least $\frac{1}{2} + \epsilon$ for every $t \in [T]$.
By applying \cite[Theorem~1.3]{dinur2008locally} (see also \cref{eq:intro:2CSP:DG08}) to each function $F^{(t)}$,
there exists an assignment $f^{(t)} \colon V \to \Sigma$ such that
\begin{align}
    \Pr_{\rv{X} \in \binom{V}{k}}\left[
        \close{C^{F^{(t)}}(\rv{X})}{(f^{(t)})^k(\rv{X})}{k^{-\Omega(1)}}
    \right] \geq \frac{1}{2} + \Omega(\epsilon).
\end{align}
By \cite[Proof of Theorem~6.1]{impagliazzo2012new} (see also \cref{clm:PSPACE:W:fE}), we further obtain
\begin{align}
\label{eq:intro:2CSP:soundness:agree}
    \Pr_{\rv{Z} \in \binom{E}{k}}\left[
        \close{F^{(t)}(\rv{Z})}{(f^{(t)})_E^k(\rv{Z})}{k^{-\Omega(1)}}
    \right] \geq \frac{1}{2} + \Omega(\epsilon).
\end{align}

Consider now the assignment sequence $\sq{f} \defeq (f^{(1)}, \ldots, f^{(T)})$.\footnote{
We can safely assume that $f^{(1)} = f_\sss$ and $f^{(T)} = f_\ttt$.
} 
If $\sq{f}$ were a valid reconfiguration sequence,
the soundness assumption would imply that
$\val_G(f^{(t)}) < s$ for some assignment $f^{(t)}$.
However, $\sq{f}$ is not necessarily a reconfiguration sequence
since $f^{(t)}$ and $f^{(t+1)}$ may differ in two or more vertices.
To address this issue, we first use the following claim.

\begin{claim}[informal; see \cref{clm:PSPACE:soundness:F1F2f1f2}]
\label{clm:intro:2CSP:soundness:F1F2f1f2}
    The distance between $f^{(t)}$ and $f^{(t+1)}$ is $k^{-\Omega(1)}$.
\end{claim}

\noindent
The crux of the proof of \cref{clm:intro:2CSP:soundness:F1F2f1f2} is that
we can perform the \emph{unique decoding} on $F^{(t)}$ and $F^{(t+1)}$ to obtain $f^{(t)}$ and $f^{(t+1)}$.
Since $F^{(t)}$ and $F^{(t+1)}$ differ only in a single coordinate and \cref{eq:intro:2CSP:soundness:agree} holds,
$(f^{(t)})_E^k(\rv{Z})$ and $(f^{(t+1)})_E^k(\rv{Z})$ approximately agree 
with non-negligible probability $\Omega(\epsilon)$,
which implies that $f^{(t)}$ and $f^{(t+1)}$ must become closer as $k$ increases.
We then show that 
if every adjacent pair of assignments in $\sq{f}$ are sufficiently close,
then some assignment $f^{(t)}$ violates an $\Omega(1)$-fraction of the constraints.

\begin{claim}[informal; see \cref{clm:PSPACE:soundness:interpolate}]
\label{clm:intro:2CSP:soundness:interpolate}
    If $f^{(t)}$ and $f^{(t+1)}$ are $\frac{1-s}{4}$-close for every $t \in [T-1]$,
    then there exists some assignment $f^{(t)}$ such that
    \begin{align}
    \label{eq:intro:2CSP:soundness:interpolate}
        \val_G(f^{(t)}) < \frac{1+s}{2} = 1 - \Omega(1).
    \end{align}
\end{claim}

\noindent
In the proof of \cref{clm:intro:2CSP:soundness:interpolate},
we reconstruct a valid reconfiguration sequence from $\sq{f}$
by ``interpolating'' between each adjacent pair $f^{(t)}$ and $f^{(t+1)}$.
Since we can assume that the underlying graph $(V,E)$ is regular \cite{ohsaka2023gap},
all intermediate assignments connecting $f^{(t)}$ and $f^{(t+1)}$ have nearly the same value.

By applying \cref{clm:intro:2CSP:soundness:F1F2f1f2,clm:intro:2CSP:soundness:interpolate} for sufficiently large $k$,
we obtain $F^{(t)}$ and $f^{(t)}$ satisfying 
both \cref{eq:intro:2CSP:soundness:agree,eq:intro:2CSP:soundness:interpolate}.
It remains to bound the acceptance probability of $\V$ from above.

\begin{claim}[informal; see \cref{clm:PSPACE:soundness:last}]
\label{clm:intro:2CSP:soundness:last}
    $\V$ accepts $F^{(t)}$ with probability at most $\frac{1}{2}+\epsilon$.
\end{claim}

\noindent
To provide intuition for the proof of \cref{clm:intro:2CSP:soundness:last},
let us ignore the correlation between the two queries $\rv{Z}_1, \rv{Z}_2 \in \binom{E}{k}$
and consider a thought experiment in which they are independent.
By \cref{eq:intro:2CSP:soundness:agree},
$\close{F^{(t)}(\rv{Z}_i)}{(f^{(t)})_E^k(\rv{Z}_i)}{k^{-\Omega(1)}}$ holds for some $i \in [2]$ 
with probability $1-\frac{1}{4}$.
Moreover, by \cref{eq:intro:2CSP:soundness:interpolate},
$f^{(t)}$ violates an $\Omega(1)$-fraction of the edges in such $\rv{Z}_i$ with probability $1-2^{-\Omega(k)}$.
Combining these events,
with probability $1-\frac{1}{4}-2^{-\Omega(k)}$,
at least one edge of $\rv{Z}_i$ is violated by $F^{(t)}(\rv{Z}_i)$; i.e.,
$F^{(t)}$ would fail the satisfiability test.
Consequently,
$\val_{\V}(\sq{F}) \leq \val_{\V}(F^{(t)}) \leq \frac{1}{4} + 2^{-\Omega(k)} < \frac{1}{2} + \epsilon$,
which completes the proof of \cref{lem:intro:2CSP}.
\qed

\begin{figure}[t]
    \centering
    \null\hfill
    \begin{minipage}[t]{0.64\linewidth}
        \centering\includegraphics[height=3.2cm]{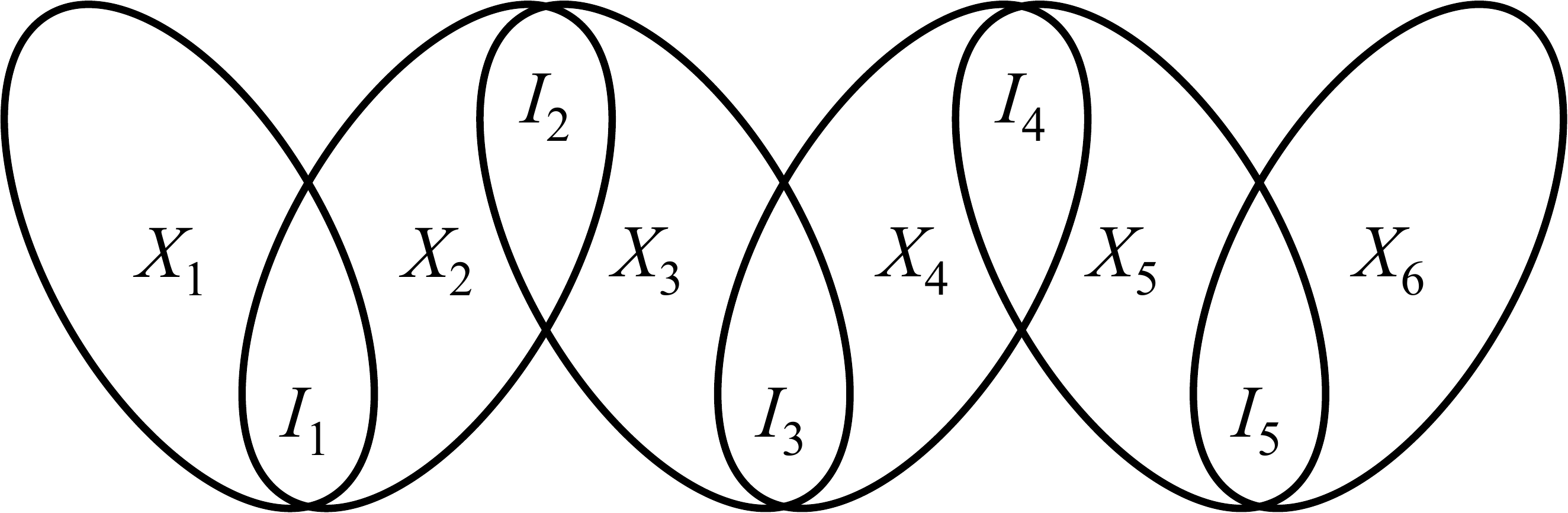}
        \caption{
        Illustration of the query sequence of the $q$-query direct product tester.
        For each $i \in [q-1]$, the consecutive query sets $X_i$ and $X_{i+1}$ both contain $I_i$, yielding the ``zig-zag'' pattern.
        }
        \label{fig:q-query}
    \end{minipage}
    \hfill
    \begin{minipage}[t]{0.32\linewidth}
        \centering\includegraphics[height=3.2cm]{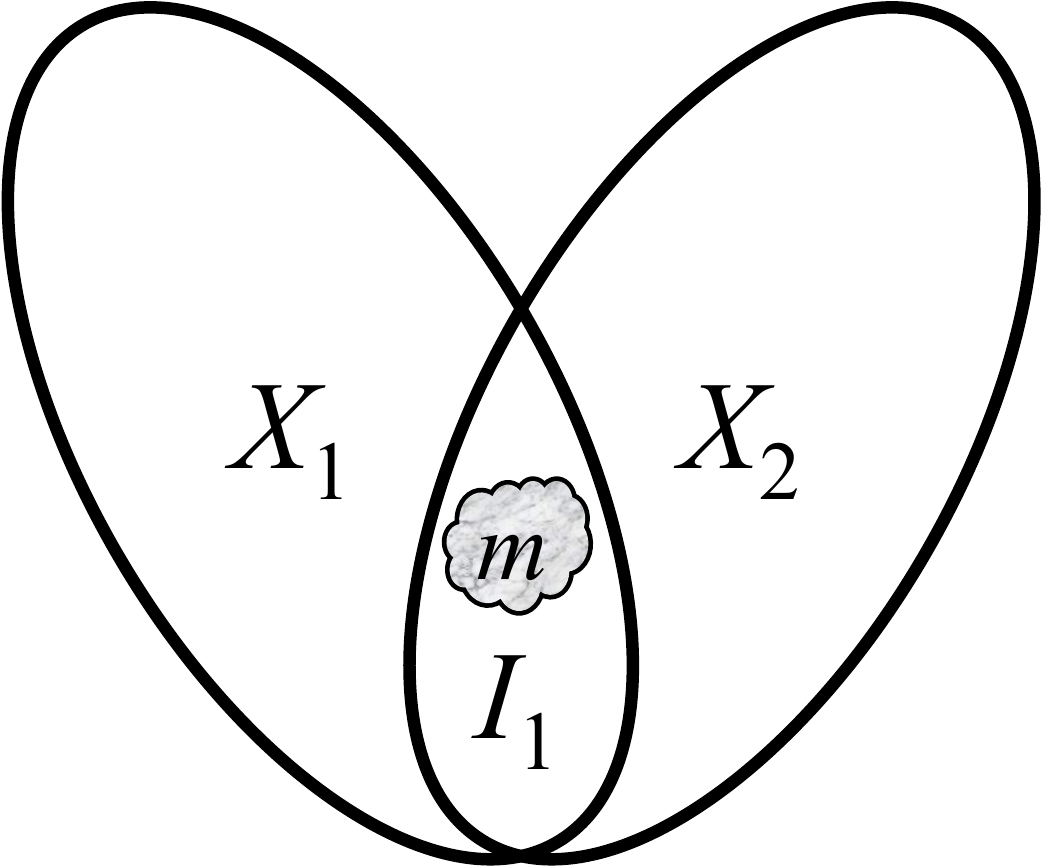}
        \caption{
        Illustration of the acceptance condition of the tolerant direct product tester,
        which allows up to $m$ disagreements between $F(\rv{X}_1)|_{\rv{I}_1}$ and $F(\rv{X}_2)|_{\rv{I}_1}$.
        }
        \label{fig:tolerant}
    \end{minipage}
    \hfill\null
\end{figure}

\subsubsection{\texorpdfstring{%
Highlight of the Proof of \cref{lem:intro:PSPACE}
}{%
Highlight of the Proof of Lemma \ref{lem:intro:PSPACE}
}}
\label{sec:overview:PSPACE:reduce}

Consider now proving \cref{lem:intro:PSPACE}.
To this end, we extend the 2-query direct product tester $\Tstd{2}$
used in \cref{sec:overview:PSPACE:2CSP} in two directions.
Let $F \colon \binom{V}{k} \to \Sigma^k$ be a $k$-set function.

\begin{description}
\item[($q$-query tester)]
To reduce \MMtwoCSPReconf to \MMqCSPReconf,
we modify $\Tstd{2}$ to make $q$ queries for an arbitrary integer $q \geq 2$.
Specifically, we first sample
$\rv{I}_1, \ldots, \rv{I}_{q-1} \in \binom{V}{\ell}$ and
$\rv{X}_1, \ldots, \rv{X}_q \in \binom{V}{k}$
conditioned on
$\rv{X}_i \cap \rv{X}_{i+1} \supset \rv{I}_i$ for every $i \in [q-1]$.
Then, we accept if
$F(\rv{X}_i)$ and $F(\rv{X}_{i+1})$ agree on $\rv{I}_i$ for every $i \in [q-1]$; namely,
$F(\rv{X}_i)|_{\rv{I}_i} = F(\rv{X}_{i+1})|_{\rv{I}_i}$.
See \cref{fig:q-query} for an illustration of the ``zig-zag'' query pattern of this test.
Note that the cases of $q=2$ and $q=3$ coincide with
the 2-query direct product tester (called the ``V-test'') 
\cite{dinur2006assignment,dinur2008locally,impagliazzo2012new} and
the 3-query direct product tester (called the ``Z-test'') \cite{impagliazzo2012new,dinur2023exponentially},\footnote{
Strictly speaking, the case of $q=3$ is identical to \cite[Test~1]{dinur2023exponentially}.
}
respectively.
Hence, this $q$-query tester can be thought of as
a natural generalization of these canonical direct product testers.

\item[(Tolerant tester)]
To ensure perfect completeness,
we relax the acceptance condition of $\Tstd{2}$ so as to accept ``approximate'' direct product functions.
Similarly to $\Tstd{2}$,
we first sample $\rv{I}_1 \in \binom{V}{\ell}$ and $\rv{X}_1, \rv{X}_2 \in \binom{V}{k}$
conditioned on $\rv{X}_1 \cap \rv{X}_2 \supset \rv{I}_1$.
Then, we accept if
$F(\rv{X}_1)$ and $F(\rv{X}_2)$ agree on at least $\ell-m$ coordinates of $\rv{I}_1$; namely,
$\close{F(\rv{X}_1)|_{\rv{I}_1}}{F(\rv{X}_2)|_{\rv{I}_1}}{m/\ell}$.
Here, $m$ is the \defi{tolerance parameter} such that $1 \leq m \leq \ell$.
See \cref{fig:tolerant} for an illustration of the acceptance condition of this test.
Observe that
if there exists a function $f \colon V \to \Sigma$ such that
$\close{F(X)}{f^k(X)}{m/(2k)}$ for each set $X \in \binom{V}{k}$,
then $F$ is accepted with probability $1$.
In this sense, we always accept approximate direct product functions.
\end{description}

Combining these extensions,
we obtain the \defi{tolerant $q$-query direct product tester $\Ttol{q}$} described below.

\begin{itembox}[l]{\textbf{Tolerant $q$-query direct product tester $\Ttol{q}$}}
\begin{algorithmic}[1]
    \item[\textbf{Input:}]
        an intersection parameter $\ell$ and a tolerance parameter $m$ with $1 \leq m \leq \ell \leq k$.
    \item[\textbf{Oracle access:}]
        a $k$-set function $F \colon \binom{V}{k} \to \Sigma^k$.
    \State sample $\rv{I}_1, \ldots, \rv{I}_{q-1}$ from $\binom{V}{\ell}$ uniformly.
    \State sample $\rv{X}_1, \ldots, \rv{X}_q$ from $\binom{V}{k}$ uniformly
    conditioned on $\rv{X}_i \cap \rv{X}_{i+1} \supset \rv{I}_i$ for every $i \in [q-1]$.
    \If{$\close{
        F(\rv{X}_i)|_{\rv{I}_i}}{
        F(\rv{X}_{i+1})|_{\rv{I}_i}}{
        m/\ell}$ for every $i \in [q-1]$}
        \State \Return $1$.
    \Else
        \State \Return $0$.
    \EndIf
\end{algorithmic}
\end{itembox}

\noindent
The following theorem provides the completeness and soundness guarantees of $\Ttol{q}$,
which may be of independent interest.

\begin{theorem}[informal; see \cref{thm:TDP}]
\label{thm:intro:TDP}
    For a $k$-set function $F \colon \binom{V}{k} \to \Sigma^k$, the following hold:
    \begin{description}
        \item[(Completeness)] 
            Suppose that there exists a function $f \colon V \to \Sigma$ such that
            \begin{align}
                \Pr_{\rv{X} \in \binom{V}{k}}\left[
                    \close{F(\rv{X})}{f^k(\rv{X})}{m/(2k)}
                \right] \geq p.
            \end{align}
            Then, $\Ttol{q}$ accepts $F$ with probability at least $p^q(1-o(1))$.
            Moreover, if $p=1$, then $\Ttol{q}$ accepts $F$ with probability $1$.
        \item[(Soundness)]
            Suppose that $\Ttol{q}$ accepts $F$ with probability $p \geq k^{-\Omega(1)}$.
            Then, there exists a function $f \colon V \to \Sigma$ such that
            \begin{align}
                \Pr_{\rv{X} \in \binom{V}{k}}\left[
                    \close{F(\rv{X})}{f^k(\rv{X})}{\delta}
                \right]
                \geq p^{\frac{1}{q-1}}(1-o(1)),
                \text{ where }
                \delta = \max\left\{
                    \frac{\ell}{k},
                    \frac{m}{\ell}
                \right\}^{\Omega(1)}.
            \end{align}
    \end{description}
\end{theorem}

Some remarks are in order.
First, approximate agreement in the soundness guarantee becomes more likely as $q$ increases.
Suppose, for example, that 
$\Ttol{q}$ accepts $F$ with probability $\frac{1}{2^{q-1}} + \epsilon$;
then, there exists a function $f \colon V \to \Sigma$ such that
$\close{F(\rv{X})}{f^k(\rv{X})}{\delta}$ with probability $\frac{1}{2}+\Omega_q(\epsilon)$.
Second, we have a good approximate agreement parameter $\delta = k^{-\Omega(1)}$
as long as $\frac{\ell}{k} = k^{-\Omega(1)}$ and $\frac{m}{\ell} = k^{-\Omega(1)}$.
For example, setting $\ell \defeq \Theta(\sqrt{k})$ and $m \defeq \Theta(\sqrt{\ell})$ is sufficient.
\cref{thm:intro:TDP} can be proved by extending the analysis of $\Tstd{2}$ due to \cite{impagliazzo2012new,dinur2008locally}.
Our insight behind the proof
is that the query sequence $(\rv{X}_1, \ldots, \rv{X}_q)$ of $\Ttol{q}$ can be approximated by
the length-$q$ random walk on the Johnson graph $J(n,k,\ell)$.\footnote{
    The vertex set of $J(n,k,\ell)$ is $\binom{[n]}{k}$ and
    the edge set contains an edge $X_1X_2$ if $|X_1 \cap X_2| = \ell$.
}
Since $J(n,k,\ell)$ is an $\bigO\left(\frac{\ell}{k}\right)$-expander \cite{dinur2008locally},
we can apply the expander hitting property \cite{ajtai1987deterministic,alon1995derandomized}
to extend the argument of \cite{dinur2008locally} to the $q$-query regime.
See \cref{sec:TDP} for details.

Our gap-amplifying reduction from \MMtwoCSPReconf to \MMqCSPReconf is obtained basically
from the reduction in \cref{sec:overview:PSPACE:2CSP}
by replacing the standard 2-query direct product tester $\Tstd{2}$ with
the tolerant $q$-query direct product tester $\Ttol{q}$.
See \cref{sec:PSPACE} for details.

\subsection{\texorpdfstring{%
$\NP$-membership of $\left(\frac{1}{2^{q-1}}-\epsilon\right)$-factor Approximation (\cref{sec:NP})
}{%
NP-membership of (1/2\textasciicircum(q-1)-ε)-factor Approximation (Section \ref{sec:NP})
}}
\label{sec:overview:NP}

Second, we outline the proof of \cref{thm:intro:NP}, i.e.,
$\NP$-membership of a $\left(\frac{1}{2^{q-1}}-\epsilon\right)$-factor approximation in the perfect completeness case.
Let $G = (V,E,\Sigma,\Psi)$ be a satisfiable $q$-CSP instance, and
$f_\sss,f_\ttt \colon V \to \Sigma$ be a pair of its satisfying assignments.
Define $n \defeq |V|$, $m \defeq |E|$, and $\sigma \defeq |\Sigma|$.
Let $\epsilon > 0$ be any real.
Assume that $\opt_G(f_\sss \reco f_\ttt) = 1$.
To prove \cref{thm:intro:NP},
it is sufficient to show that there exists a polynomial-length reconfiguration sequence
whose value is at least $\frac{1}{2^{q-1}}-\epsilon$.

Our main idea is to partition the vertex set $V$ into low-degree and high-degree vertices.
A similar strategy was used to approximate
\MMtwoCSPReconf \cite{guruswami2025inapproximability} and
\prb{Maxmin $k$-Cut Reconfiguration} \cite{hirahara2025asymptotically}.
Define 
\begin{align}
    \Delta \defeq \Theta\left(\frac{\epsilon^2 m}{q \log n}\right).
\end{align}
For a vertex $x$ of $G$, let $d(x)$ denote the \defi{degree} of $v$.
We say that a vertex is \defi{low degree} if its degree is at most $\Delta$ and \defi{high degree} otherwise.
Let $L$ and $H$ denote the sets of low-degree and high-degree vertices, respectively; namely,
$L \defeq \{ x \in V \mid d(x) \leq \Delta \}$ and
$H \defeq \{ x \in V \mid d(x) > \Delta \}$.
The number of high-degree vertices is bounded by
$|H| \leq \frac{qm}{\Delta} = \bigO_{q,\epsilon}(\log n)$.

We first consider the restricted case where
$f_\sss$ and $f_\ttt$ disagree only on low-degree vertices.
Then, there always exists a polynomial-length reconfiguration sequence
from $f_\sss$ to $f_\ttt$ with value at least $\frac{1}{2^{q-1}}-\epsilon$.

\begin{lemma}[informal; see \cref{lem:NP:low-degree}]
\label{lem:intro:NP:low-degree}
    Let $D$ denote the set of vertices on which $f_\sss$ and $f_\ttt$ disagree; namely,
    $D \defeq \{ x \in V \mid f_\sss(x) \neq f_\ttt(x) \}$.
    If every vertex $x \in D$ is low degree,
    then there exists a reconfiguration sequence $\sq{f}$ from $f_\sss$ to $f_\ttt$ of length $|D|+1$
    such that $\val_G(\sq{f}) \geq \frac{1}{2^{q-1}}-\epsilon$.
\end{lemma}
\begin{proof}[Proof Sketch]
Consider a random reconfiguration sequence
$\sq{\rv{f}} = (\rv{f}^{(1)}, \ldots, \rv{f}^{(|D|+1)})$ from $f_\sss$ to $f_\ttt$
obtained by changing the assignments to the vertices of $D$ in a random order.
Simple calculation yields that each intermediate assignment $\rv{f}^{(t)}$ satisfies
each hyperedge $e \in E$ with probability at least $\frac{1}{2^{q-1}}$.
Since the degree of every vertex in $D$ is at most $\Delta$,
McDiarmid's inequality \cite{mcdiarmid1989method} derives that
with probability $1 - n^{-\Omega(1)}$,
all functions in $\sq{\rv{f}}$ simultaneously satisfy
at least a $\left(\frac{1}{2^{q-1}}-\epsilon\right)$-fraction of the constraints.
Therefore, $\val_G(\sq{\rv{f}}) \geq \frac{1}{2^{q-1}}-\epsilon$ with high probability, as desired.
\end{proof}

\noindent
\cref{lem:intro:NP:low-degree} implies that
we can ignore assignments to the low-degree vertices,
which exponentially reduces the number of assignments to consider.

We now consider the general case where $f_\sss$ and $f_\ttt$ may disagree on high-degree vertices.
Let $\sq{f} = (f^{(1)}, \ldots, f^{(T)})$ be a reconfiguration sequence from $f_\sss$ to $f_\ttt$
consisting only of satisfying assignments for $G$.
Since $\sq{f}$ may be exponentially long,
we shall sparsify it to reduce the number of assignments.
Specifically, we extract a subsequence of $\sq{f}$ by the following procedure.

\begin{itembox}[l]{\textbf{Sparsifying $\sq{f} = (f^{(1)}, \ldots, f^{(T)})$}}
\begin{algorithmic}[1]
    \State let $I \defeq [T]$.
    \For{\textbf{each} index $t_\ell \in I$ in ascending order}
        \State find the largest index $t_r \in I$ such that $f^{(t_\ell)}|_H = f^{(t_r)}|_H$.
        \If{$t_\ell+1 \leq t_r-1$}
            \State delete $t_\ell+1, \ldots, t_r-1$ from $I$.
        \EndIf
    \EndFor
    \State \Return the subsequence $(f^{(t)})_{t \in I}$.
\end{algorithmic}
\end{itembox}

\noindent
Intuitively, we consider the Hamming graph
whose vertices are partial assignments $\alpha \colon H \to \Sigma$ to the high-degree vertices, and
compress subpaths connecting the same vertex of this graph.
Let $\sq{g} = (g^{(1)}, \ldots, g^{(T')})$ denote the subsequence of $\sq{f}$ obtained by the sparsification procedure.
Note that $g^{(1)} = f^{(1)}$ and $g^{(T')} = f^{(T)}$.
Consider first bounding the length $T'$ of $\sq{g}$.
Let $\alpha \colon H \to \Sigma$ be a partial assignment for the high-degree vertices.
By the sparsification procedure, $\alpha$ appears in $g^{(1)}|_H, \ldots, g^{(T')}|_H$ at most twice.\footnote{
Suppose that there exist $t_1, t_2, t_3 \in I$ such that
$t_1 < t_2 < t_3$ and $f^{(t_1)}|_H = f^{(t_2)}|_H = f^{(t_3)}|_H$.
Then, the sparsification procedure would have removed $t_2$.
}
The total number of such partial assignments is thus $|\Sigma^H|$.
Therefore, $T' \leq 2 \sigma^{|H|} \leq n^{\bigO_{q,\sigma,\epsilon}(1)}$.

We are now ready to construct
a polynomial-length reconfiguration sequence with value at least $\frac{1}{2^{q-1}}-\epsilon$.
Observe that for each $t \in [T'-1]$, either
$g^{(t)}|_H = g^{(t+1)}|_H$, or 
$g^{(t)}$ and $g^{(t+1)}$ differ in a single vertex.
In the former case, by \cref{lem:intro:NP:low-degree},
there exists a reconfiguration sequence $\sq{h}_t$ from $g^{(t)}$ to $g^{(t+1)}$ of length at most $n+1$ 
with value at least $\frac{1}{2^{q-1}}-\epsilon$.
In the latter case,
$\sq{h}_t \defeq (g^{(t)}, g^{(t+1)})$ is already a valid reconfiguration sequence with value $1$.
By concatenating $\sq{h}_t$ for every $t \in [T'-1]$,
we obtain a reconfiguration sequence $\sq{h}$ from $g^{(1)}$ to $g^{(T')}$ with value at least $\frac{1}{2^{q-1}}-\epsilon$.
Moreover, the length of $\sq{h}$ is at most $T'(n+1) \leq n^{\bigO_{q,\sigma,\epsilon}(1)}$.
Consequently,
there exists an $n^{\bigO_{q,\sigma,\epsilon}(1)}$-length reconfiguration sequence from $f_\sss$ to $f_\ttt$
whose value is at least $\frac{1}{2^{q-1}}-\epsilon$,
which completes the proof of \cref{thm:intro:NP}.

\subsection{Discussion and Open Problem}
\label{sec:overview:perspective}

In this paper, we establish the optimal $\PSPACE$-hardness of approximating \MMqCSPReconf; namely,
a $\left(\frac{1}{2^{q-1}}+\epsilon\right)$-factor approximation is $\PSPACE$-hard, whereas
a $\left(\frac{1}{2^{q-1}}-\epsilon\right)$-factor approximation is in $\NP$.
In particular, we discover that
the approximability of \MMCSPReconf{4} can be $\NP$-complete.
Prior to this work,
researchers had implicitly assumed that the approximability of reconfiguration problems falls into either $\cP$ or $\PSPACE$-complete.
Our results exemplify the first reconfiguration problem whose approximation is both $\NP$-complete and $\PSPACE$-complete,
suggesting a more intricate complexity landscape for the approximability of reconfiguration problems.
Specifically, to fully understand the complexity of approximating a reconfiguration problem,
one may need to resolve not only
(1)~polynomial-time approximability and (2)~$\PSPACE$-hardness of approximation, but also
(3)~$\NP$-membership of approximation and (4)~$\NP$-hardness of approximation.

An immediate open problem arising from \cref{cor:intro} is to determine
the complexity of a $\left(\frac{1}{4}-\epsilon\right)$-factor approximation for \MMCSPReconf{3},
which is currently only known to lie in $\NP$.
There are two possibilities.
One is that a $\left(\frac{1}{4}-\epsilon\right)$-factor approximation is in $\cP$,
analogous to the case of \MMCSPReconf{2},
which admits a $\left(\frac{1}{2}-\epsilon\right)$-factor approximation \cite{guruswami2025inapproximability}.
This would require finding a reconfiguration sequence achieving a $\left(\frac{1}{4}-\epsilon\right)$-factor approximation in polynomial time.
The other is that a $\left(\frac{1}{4}-\epsilon\right)$-factor approximation is $\NP$-hard.
By \cref{thm:intro:regular}, \MMCSPReconf{3} on regular instances admits a $\left(\frac{1}{4}-\epsilon\right)$-factor approximation.
Therefore, proving such $\NP$-hardness would require a reduction that produces highly non-regular instances.
This is a different trend from \prb{Max $q$-CSP},
which exhibits almost the same inapproximability behavior on regular and non-regular instances,
see, e.g., \cite[\S7]{minzer2024near}, \cite[\S6]{minzer2025near}, and \cite{stankovic2022regularity}.
Indeed, to show that \MMCSPReconf{4} is $\NP$-hard to approximate within any constant factor,
\citet{ohsaka2024gap} constructed a reduction from \prb{Max 2-CSP} to \MMCSPReconf{4}
such that a few variables appear in all constraints.
It is unclear how to adapt this reduction to \MMCSPReconf{3}.

\section{Related Work}
\label{sec:related}

\subsection{Reconfiguration Problems}

Combinatorial reconfiguration aims to study algorithmic problems and structural properties in the space of feasible solutions.
In the unified framework \cite{ito2011complexity},
a \defi{reconfiguration problem} is defined with respect to 
a combinatorial problem $\Pi$ and a transformation rule $R$ on the feasible solutions of $\Pi$.
For an instance $\calI$ of $\Pi$ and 
a pair of its feasible solutions $S_\sss$ and $S_\ttt$,
the reconfiguration problem asks whether $S_\sss$ can be transformed into $S_\ttt$
by repeatedly applying the transformation rule $R$ while
always preserving the feasibility of any intermediate solution.
Speaking differently, the reconfiguration problem concerns the $st$-connectivity in the \defi{configuration graph},
which is a graph $G_{\calI,R}$ where
each node corresponds to a feasible solution of $\calI$ and
each link represents that its endpoints can be transformed into each other by applying $R$.
The instance $(\calI,S_\sss,S_\ttt)$ is a \Yes-instance of the reconfiguration problem
if and only if there is a path from $S_\sss$ to $S_\ttt$ in $G_{\calI,R}$.
Such a sequence of feasible solutions that forms a path in the configuration graph
is called a \defi{reconfiguration sequence}.
Over the past two decades, many combinatorial problems have given rise to reconfiguration problems.
For example, reconfiguration problems of
\prb{3-SAT} \cite{gopalan2009connectivity},
\prb{4-Coloring} \cite{bonsma2009finding},
\prb{Independent Set} \cite{hearn2005pspace,hearn2009games,kaminski2012complexity}, and
\prb{Shortest Path} \cite{bonsma2013complexity}
are $\PSPACE$-complete, whereas
those of
\prb{2-SAT} \cite{gopalan2009connectivity},
\prb{3-Coloring} \cite{cereceda2011finding},
\prb{Matching} \cite{ito2011complexity}, and
\prb{Spanning Tree} \cite{ito2011complexity}
belong to $\cP$.
We refer the reader to the surveys \cite{nishimura2018introduction,heuvel2013complexity,mynhardt2019reconfiguration,bousquet2024survey}
as well as the Combinatorial Reconfiguration wiki \cite{hoang2024combinatorial}.

\subsection{Approximability of Reconfiguration Problems}

For a reconfiguration problem,
its approximate version allows infeasible feasible solutions,
but requires optimizing the ``worst'' feasibility along the reconfiguration sequence.
\citet{ito2011complexity}
showed that several reconfiguration problems are $\NP$-hard to approximate.
Since many reconfiguration problems are $\PSPACE$-complete,
$\NP$-hardness results are not optimal.
\cite{ito2011complexity} posed
the $\PSPACE$-hardness of approximation for reconfiguration problems as an open problem.
\citet{ohsaka2023gap} postulated a reconfiguration analogue of the PCP theorem \cite{arora1998probabilistic,arora1998proof},
called the Reconfiguration Inapproximability Hypothesis (RIH), and
proved that assuming RIH, several reconfiguration problems are $\PSPACE$-hard to approximate,
including those of 
\prb{3-SAT}, \prb{Independent Set}, \prb{Vertex Cover}, \prb{Clique}, and \prb{Set Cover}.
\citet{hirahara2024probabilistically,guruswami2025inapproximability}
independently proved RIH by establishing the Probabilistically Checkable Reconfiguration Proof (PCRP) theorem,
which provides a PCP-type characterization of $\PSPACE$.
The PCRP theorem, along with a series of gap-preserving reductions,
implies the $\PSPACE$-hardness of approximating the reconfiguration problems listed above,
thereby resolving the open problem of \cite{ito2011complexity} affirmatively.

Since the PCRP theorem itself only implies $\PSPACE$-hardness of approximation within some constant factor, subsequent work has investigated explicit hardness factors.
In the $\NP$ regime,
the parallel repetition theorem of \citet{raz1998parallel} can be used to derive
many strong inapproximability results, e.g.,
\cite{hastad1999clique,hastad2001some,feige1998threshold,bellare1998free}.
In particular, \prb{Max 2-CSP} is $\NP$-hard to approximate within any constant factor.
However, for \MMtwoCSPReconf,
parallel repetition does not reduce the soundness error and
a $(0.25-\epsilon)$-factor approximation algorithm exists for sparse instances \cite{ohsaka2025approximate}.
Using Dinur's gap amplification \cite{dinur2007pcp},
\citet{ohsaka2024gap} showed that
\prb{Maxmin 2-CSP Reconfiguration} and \prb{Minmax Set Cover Reconfiguration} are 
$\PSPACE$-hard to approximate within factors of $0.9942$ and $1.0029$, respectively.
\citet{guruswami2025inapproximability} proved the $\NP$-hardness
of a $(0.5+\epsilon)$-factor approximation for \prb{Maxmin 2-CSP Reconfiguration} and
of a $(2-\epsilon)$-factor approximation for \prb{Minmax Set Cover Reconfiguration}.
These results are tight with respect to $\NP$-hardness because
\prb{Maxmin 2-CSP Reconfiguration} admits
    a $(0.5 - \epsilon)$-factor approximation algorithm \cite{guruswami2025inapproximability} and
\prb{Minmax Set Cover Reconfiguration} admits
    a $2$-factor approximation algorithm \cite{ito2011complexity}.
\citet{hirahara2024optimal} showed that \prb{Minmax Set Cover Reconfiguration} is
$\PSPACE$-hard to approximate within a factor of $2-o(1)$,
improving upon \cite{guruswami2025inapproximability,ohsaka2024gap}.
This is the first optimal $\PSPACE$-hardness result for approximability of reconfiguration problems.
\citet{hirahara2025asymptotically} showed that
the approximation threshold of \prb{Maxmin $k$-Cut Reconfiguration} is $1-\Theta\left(\frac{1}{k}\right)$.
\citet{hirahara2025asymptoticallya} showed that
the approximation threshold of \prb{Maxmin E$k$-SAT Reconfiguration} is $1-\Theta\left(\frac{1}{k}\right)$.
Since the approximation threshold of its $\NP$ analogue, i.e., \prb{Max E$k$-SAT}, is $1-\frac{1}{2^k}$ \cite{hastad2001some},
this is the first reconfiguration problem whose approximation threshold is (asymptotically) worse than that of its $\NP$ analogue.
Very recently, \citet{gur20263} proved
the $\PSPACE$-hardness of a $(0.9+\epsilon)$-factor approximation for \prb{Maxmin 2-CSP Reconfiguration} and
of a $(0.5+\epsilon)$-factor approximation for \prb{Maxmin 5-CSP Reconfiguration},
improving upon \cite{ohsaka2024gap,guruswami2025inapproximability}.

\subsection{Direct Product Testing}

\citet{goldreich2000combinatorial} introduced direct product testing
as a combinatorial analogue of low-degree testing \cite{arora2003improved,raz1997sub}.
For a function $f \colon V \to \Sigma$,
its $k$-wise direct product is a function $f^k \colon \binom{V}{k} \to \Sigma^k$ such that
$f^k(X) = f|_X$ for each $X \in \binom{V}{k}$.
In direct product testing,
we aim to design an efficient test that,
given oracle access to a function $F \colon \binom{V}{k} \to \Sigma^k$,
determines whether $F$ is close to some direct product function.
There are two regimes of interest with respect to the acceptance probability $p$ of the test.
The high-soundness (or ``99\%'') regime considers the case where
$p = 1-\epsilon$ for a small real $\epsilon \approx 0$.
One wishes to show that
there exists a function $f \colon V \to \Sigma$ such that
$F(X) \approx f|_X$ for a $(1-\bigO(\epsilon))$-fraction of $X \in \binom{V}{k}$.
The low-soundness (or ``1\%'') regime in which $p$ may be very small is more challenging.
One is expected to show that
$F(X) \approx f|_X$ for a $\poly(p)$-fraction of $X \in \binom{V}{k}$.
This regime is particularly important for applications to constructing PCPs with small soundness.

In the 99\% regime,
\citet{goldreich2000combinatorial} gave the first direct product tester with a constant number of queries.
\citet{dinur2006assignment} later reduced the query complexity to $2$.
In the 1\% regime, \citet{dinur2008locally} first proved that
if a 2-query direct product tester (called the ``V-test'') accepts $F$ with probability $p \geq k^{-\Omega(1)}$,
then $F(X) \approx f|_X$ for an $\Omega(p^6)$-fraction of $X \in \binom{V}{k}$.
\cite{dinur2008locally} also showed that the V-test does not work when $p \ll k^{-1}$.
\citet{impagliazzo2012new} introduced a 3-query direct product test (called the ``Z-test''),
which works even when $p \leq 2^{-\poly(k)}$.
Moreover, \cite{impagliazzo2012new} showed
how to apply the V-test to obtain a 2-query PCP whose soundness error decreases exponentially in $k$,
giving an alternative proof of the parallel repetition theorem \cite{raz1998parallel}.
Subsequent work further investigated both the V-test and the Z-test
for various parameter settings \cite{dinur2014direct,dinur2023exponentially}.

One drawback of the direct product encoding is its size:
encoding a function $f \colon V \to \Sigma$ by its $k$-wise direct product $f^k \colon \binom{V}{k} \to \Sigma^k$ incurs the exponential size in $k$.
Such a blow-up is problematic in some applications, e.g., when
constructing length-efficient PCPs with small soundness error.
This motivates the study of derandomized direct product testing,
which aims to identify a small subfamily $\calK \subset \binom{V}{k}$ such that
direct product testing remains possible when $F$ is defined only on $\calK$.
Early derandomization results were obtained in \cite{goldreich2000combinatorial,impagliazzo2012new}.
Recently, \citet{dinur2017high} proposed derandomized direct product testing on high-dimensional expanders,
which led to significant progress
\cite{dikstein2019agreement,bafna2024characterizing,bafna2024constant,dikstein2024agreement,dikstein2024low,dikstein2024swap,kaufman2025coboundary}.

Our focus differs from the above studies in that for each $q \geq 2$,
we need a $q$-query direct product tester that 
accepts approximate direct product functions and
captures the fine-grained relation between the acceptance probability and the approximate agreement with direct product functions.

\section{Preliminaries}
\label{sec:pre}

\paragraph{Notations.}

For a nonnegative integer $n$, let $[n] \defeq \{1,2,3,\ldots,n\}$.
Unless otherwise specified, the base of logarithms is $2$.
The symmetric group on $[n]$ (i.e., the set of all permutations of $[n]$) is denoted by $\sym_n$.
For nonnegative integers $n$ and $k$ with $k \leq n$,
let $n^{\ul{k}}$ denote the falling factorial; namely,
$n^{\ul{k}} \defeq \binom{n}{k}k! = \prod_{0 \leq i \leq k-1}(n-i)$.
For a finite set $V$ and a positive integer $k$,
we write $\binom{V}{k}$ for the family of all size-$k$ subsets of $V$, and
write $V^{\ul{k}}$ for the set of all $k$-tuples with distinct elements in $V$; namely,
\begin{align}
    \tbinom{V}{k} & \defeq \bigl\{
        \{x_1, \ldots, x_k\} \subseteq V
        \bigm|
        \forall i \neq j, \; x_i \neq x_j
    \bigr\}, \\
    V^{\ul{k}} & \defeq \bigl\{
        (x_1, \ldots, x_k) \in V^k
        \bigm|
        \forall i \neq j, \; x_i \neq x_j
    \bigr\}.
\end{align}
Note that
$\bigl|\binom{V}{k}\bigr| = \binom{|V|}{k}$,
$|V^{\ul{k}}| = |V|^{\ul{k}}$, and
$V^{\ul{k}} \subset V^k$.
A sequence of a finite number of elements $a^{(1)}, \ldots, a^{(T)}$
is denoted by $\sq{a} = (a^{(1)}, \ldots, a^{(T)})$.
We use the Iverson bracket $\llbracket \cdot \rrbracket$; i.e.,
for a statement $P$, we define $\llbracket P \rrbracket$ as $1$ if $P$ is true and $0$ otherwise.
We write random variables in boldface (e.g., $\rv{x}$ and $\rv{X}$).
For a set $V$, we write $\rv{x} \in V$ to mean that 
$\rv{x}$ is a random variable uniformly sampled from $V$.

Let $f,g \colon \calX \to \calY$ be two functions.
We shall use the function notation $f \colon \calX \to \calY$ and
the string notation $f \in \calY^{\calX}$ interchangeably.
For a subset $\scrI \subset \calX$,
we use $f|_{\scrI} \colon \scrI \to \calY$ to denote the restriction of $f$ to $\scrI$.
The \defi{relative Hamming distance} between $f$ and $g$,
denoted by $\dist(f,g)$,
is defined as the fraction of coordinates on which $f$ and $g$ are different; namely,
\begin{align}
    \dist(f,g)
    \defeq \frac{1}{|\calX|}\left|\bigl\{ x \in \calX \bigm| f(x) \neq g(x) \bigr\}\right|
    = \Pr_{\rv{x} \in \calX}\bigl[ f(\rv{x}) \neq g(\rv{x}) \bigr].
\end{align}
We say that
$f$ is \defi{$\delta$-close} to $g$ if $\dist(f,g) \leq \delta$, and
$f$ is \defi{$\delta$-far} from $g$ if $\dist(f,g) > \delta$.
We denote
$\close{f}{g}{\delta}$ to mean that $f$ is $\delta$-close to $g$, and
$\far{f}{g}{\delta}$ to mean that $f$ is $\delta$-far from $g$.
The \defi{mixture} of $f$ and $g$, denoted by $\mix(f,g)$,
is defined as the set of all functions $h \colon \calX \to \calY$ such that
$h(x)$ is equal to $f(x)$ or $g(x)$ for each $x \in \calX$; namely,
\begin{align}
    \mix(f,g) \defeq \Bigl\{
        h \colon \calX \to \calY
        \Bigm|
        \forall x \in \calX, \; h(x) \in \bigl\{f(x), g(x)\bigr\}
    \Bigr\}.
\end{align}

Consider a $k$-set function $F \colon \binom{V}{k} \to \Sigma^k$
for two finite sets $V$ and $\Sigma$ and a positive integer $k$.
For notational convenience,
we view $F(S)$ for each set $S \in \binom{V}{k}$ as a $k$-tuple indexed by the elements of $S$, i.e., $F(S) \in \Sigma^S$.
Specifically, if $F(S) = (\alpha_1, \ldots, \alpha_k)$,
then we write $F(S)|_{x_i} \defeq \alpha_i$
for a canonical ordering $S = \{x_1, \ldots, x_k\}$.
For a positive integer $k$ and a function $f \colon V \to \Sigma$,
the \defi{$k$-wise direct product} of $f$ is defined as
a $k$-set function $f^k \colon \binom{V}{k} \to \Sigma^k$ such that
\begin{align}
    f^k\bigl(\{x_1, \ldots, x_k\}\bigr)
    \defeq \bigl( f(x_1), \ldots, f(x_k) \bigr)
    \text{ for each } k\text{-set } \{x_1, \ldots, x_k\} \in \tbinom{V}{k}.
\end{align}
By abuse of notation,
we also think of the $k$-wise direct product as a $k$-tuple function $f^k \colon V^k \to \Sigma^k$
such that
\begin{align}
    f^k(x_1, \ldots, x_k) \defeq \bigl(f(x_1), \ldots, f(x_k)\bigr)
    \text{ for each } k\text{-tuple } (x_1, \ldots, x_k) \in V^k.
\end{align}

For a graph $G = (V,E)$, let $V(G)$ and $E(G)$ denote the vertex set and edge set of $G$, respectively.
We write $xy$ for an edge between a pair of vertices $x$ and $y$.
For a vertex set $S \subseteq V(G)$, we write $G[S]$ for the subgraph of $G$ induced by $S$.
The \defi{size} of $G$, denoted by $|G|$, is defined as the number of edges in $G$; namely,
$|G| \defeq |E(G)|$.
The \defi{degree} of $v$, denoted by $d(v)$, is defined as the number of edges that are incident to $v$.
We use the same notations for hypergraphs.

\paragraph{Definition of \MMqCSPReconf.}
We formulate \MMqCSPReconf and its gap version.
The notion of $q$-CSP instance is introduced below.

\begin{definition}
A \defi{$q$-CSP instance} is defined as a quadruple $G = (V,E,\Sigma,\Psi)$ such that
\begin{itemize}
    \item $(V,E)$ is a $q$-uniform hypergraph called the \defi{underlying hypergraph},
    \item $\Sigma$ is a finite set called the \defi{alphabet}, and
    \item $\Psi = (\psi_e)_{e \in E}$
    is a collection of \defi{$q$-ary constraints},
    where each constraint $\psi_e \colon \Sigma^e \to \zo$ is a circuit.
\end{itemize}
\end{definition}

An \defi{assignment} for a $q$-CSP instance $G=(V,E,\Sigma,\Psi = (\psi_e)_{e \in E})$ is
a function $f \colon V \to \Sigma$ that assigns a symbol of $\Sigma$ to each vertex of $V$.
We say that
$f$ \defi{satisfies} a hyperedge $e = \{x_1, \ldots, x_q\} \in E$ (or a constraint $\psi_e$) 
if $\psi_e(f(x_1), \ldots, f(x_q)) = 1$, and
$f$ \defi{satisfies} $G$ if it satisfies all hyperedges of $G$.
For a pair of assignments $f_\sss,f_\ttt \colon V \to \Sigma$ for $G$,
a \defi{reconfiguration sequence} from $f_\sss$ to $f_\ttt$
is defined as a sequence $(f^{(1)}, \ldots, f^{(T)})$ over $\Sigma^V$ such that
$f^{(1)} = f_\sss$, $f^{(T)} = f_\ttt$, and 
every adjacent pair $f^{(t)}$ and $f^{(t+1)}$ differ in at most a single vertex.
We sometimes call $f_\sss$ and $f_\ttt$ the \defi{starting} and \defi{ending} assignments, respectively.
In the \prb{$q$-CSP Reconfiguration} problem,
for a satisfiable $q$-CSP instance $G$ and a pair of its satisfying assignments $f_\sss$ and $f_\ttt$,
we are asked to decide if
there exists a reconfiguration sequence from $f_\sss$ to $f_\ttt$ consisting of satisfying assignments for $G$.
Hereafter, the suffix ``$q$-CSP$_\sigma$'' indicates that the alphabet size is $\sigma$.

We now formulate an approximate version of \prb{$q$-CSP Reconfiguration}.
Let $G=(V,E,\Sigma,\Psi)$ be a $q$-CSP instance.
The \defi{value} of an assignment $f \colon V \to \Sigma$ for $G$,
denoted by $\val_G(f)$,
is defined as the fraction of hyperedges of $G$ satisfied by $f$; namely,
\begin{align}
    \val_G(f)
    \defeq \frac{1}{|E|} \left|\bigl\{ e \in E \bigm| f \text{ satisfies } e \bigr\}\right|
    = \Pr_{\rv{e} \in E}\bigl[ f \text{ satisfies } \rv{e} \bigr].
\end{align}
The \defi{value} of a reconfiguration sequence $\sq{f} = (f^{(1)}, \ldots, f^{(T)})$,
denoted by $\val_G(\sq{f})$,
is defined as the minimum fraction of satisfied hyperedges of $G$, where
the minimum is taken over all assignments $f^{(t)}$ in $\sq{f}$; namely,
\begin{align}
    \val_G(\sq{f}) \defeq \min_{1 \leq t \leq T} \bigl\{ \val_G(f^{(t)}) \bigr\}.
\end{align}

\noindent
The \MMqCSPReconf problem \cite{ito2011complexity,ohsaka2023gap} is defined as follows.

\begin{problem}
Given a satisfiable $q$-CSP instance $G$ and
a pair of its satisfying assignments $f_\sss$ and $f_\ttt$,
\MMqCSPReconf asks for
a reconfiguration sequence $\sq{f}$ from $f_\sss$ to $f_\ttt$ such that $\val_G(\sq{f})$ is maximized.
\end{problem}

Let $\opt_G(f_\sss \reco f_\ttt)$ denote the \defi{optimal value} of \MMqCSPReconf,
which is defined as the maximum value of $\val_G(\sq{f})$
over all possible reconfiguration sequences $\sq{f}$ from $f_\sss$ to $f_\ttt$; namely,
\begin{align}
    \opt_G(f_\sss \reco f_\ttt)
    \defeq \max_{\sq{f} = (f_\sss, \ldots, f_\ttt)} \bigl\{ \val_G(\sq{f}) \bigr\}.
\end{align}

\noindent
The gap version of \MMqCSPReconf is defined as follows.

\begin{problem}
For any reals $c,s$ with $0 < s \leq c \leq 1$,
\prb{Gap$_{c,s}$ \qCSPReconf} is a promise problem that asks,
given a satisfiable $q$-CSP instance $G$ and a pair of its satisfying assignments $f_\sss$ and $f_\ttt$,
to distinguish whether $\opt_G(f_\sss \reco f_\ttt) \geq c$ or $\opt_G(f_\sss \reco f_\ttt) < s$.
\end{problem}

\noindent
Note that the case of $c=s=1$ is equivalent to \qCSPReconf.
By the PCRP theorem \cite{guruswami2025inapproximability,hirahara2024probabilistically},
\prb{Gap$_{1,s}$ 2-CSP$_\sigma$ Reconfiguration} is $\PSPACE$-complete for
some real $s \in (0,1)$ and some positive integer $\sigma$.

\paragraph{Expander Graphs and Random Walk.}
Let $G$ be a $d$-regular graph and $A$ be its normalized adjacency matrix; i.e.,
$A_{x,y} \defeq \frac{1}{d}\llbracket xy \in E(G) \rrbracket$ for each $x,y \in V(G)$.
The \defi{second largest eigenvalue} of $G$, denoted by $\lambda(G)$, is defined as
\begin{align}
    \lambda(G)
    \defeq \max_{x \in \bbR^{V(G)}: x \perp \vec{1}} \frac{\|Ax\|_2}{\|x\|_2},
\end{align}
where $\vec{1}$ is the all-ones vector.
For a real $\lambda \in (0,1)$,
a \defi{$\lambda$-expander graph} is defined as
a regular graph such that $\lambda(G) \leq \lambda$.

The \defi{length-$q$ random walk} on a graph $G$ is defined as
a sequence of random $q$ vertices $(\rv{x}_1, \ldots, \rv{x}_q)$ such that
$\rv{x}_1$ is selected from $V(G)$ uniformly and 
$\rv{x}_{i+1}$ is selected from the neighbors of $\rv{x}_i$ uniformly.
We write $(\rv{x}_1, \ldots, \rv{x}_q) \sim \RW_q(G)$
for the length-$q$ random walk on $G$.
The hitting property of expander walks \cite{ajtai1987deterministic,alon1995derandomized} says that 
for a $\lambda$-expander graph $G$ and a vertex set $S$ of size $p|V(G)|$,
the length-$q$ random walk on $G$ is entirely contained in $S$ with probability $(p\pm\bigO(\lambda))^q$
(see also \cite[\S3]{hoory2006expander}).

\begin{theorem}[Expander hitting property \cite{ajtai1987deterministic,alon1995derandomized}]
\label{thm:hitting}
    Let $G$ be a $\lambda$-expander graph,
    $S \subseteq V(G)$ be a vertex set, and
    $p \defeq \frac{|S|}{|V(G)|}$.
    Then,
    \begin{align}
        \Pr_{(\rv{x}_1, \ldots, \rv{x}_q) \sim \RW_q(G)}\Biggl[
            \bigwedge_{1 \leq i \leq q} \rv{x}_i \in S
        \Biggr] \leq p(p+\lambda)^{q-1}.
    \end{align}
    Moreover, if $\lambda < \frac{p}{6}$, then
    \begin{align}
        \Pr_{(\rv{x}_1, \ldots, \rv{x}_q) \sim \RW_q(G)}\Biggl[
            \bigwedge_{1 \leq i \leq q} \rv{x}_i \in S
        \Biggr] \geq p(p-2\lambda)^{q-1}.
    \end{align}
\end{theorem}

\paragraph{Concentration Inequalities.}

We introduce Hoeffding's inequality \cite{hoeffding1963probability}, the Chernoff bound \cite{chernoff1952measure}, and McDiarmid's inequality \cite{mcdiarmid1989method}.

\begin{lemma}[Hoeffding's inequality \cite{hoeffding1963probability}]
\label{lem:Hoeffding}
    Let $\rv{X}_1, \ldots, \rv{X}_n$ be $n$ independent random variables such that
    $a_i \leq \rv{X}_i \leq b_i$ for every $i \in [n]$.
    Let $\rv{X} \defeq \sum_{1 \leq i \leq n} \rv{X}_i$.
    Then, for any real $t > 0$,
    \begin{align}
        \Pr\Bigl[
            \bigl|\rv{X} - \E[\rv{X}]\bigr| \geq t
        \Bigr]
        \leq 2\exp\left(
            -\frac{2t^2}{\sum_{1 \leq i \leq n} (b_i-a_i)^2}
        \right).
    \end{align}
\end{lemma}

\noindent
Hoeffding's inequality holds when
$\rv{X}_i$'s are sampled without replacement, described as follows.

\begin{lemma}[Hoeffding's inequality for sampling without replacement]
\label{lem:Hoeffding:without}
    Let $\calX = (x_1,\ldots,x_N) \in [0,1]^N$ be a finite population of $N$ reals, and
    $\rv{X}_1,\ldots,\rv{X}_n$ be $n$ random variables
    sampled without replacement from $\calX$ uniformly.
    Let $\rv{X} \defeq \sum_{1 \leq i \leq n} \rv{X}_i$.
    Then, for any real $t > 0$,
    \begin{align}
        \Pr\Bigl[
            \bigl|\rv{X}-\E[\rv{X}]\bigr| \geq t
        \Bigr]
        \leq
        2\exp\left(
            -\frac{2t^2}{n}
        \right).
    \end{align}
\end{lemma}

\begin{lemma}[Chernoff bound \cite{chernoff1952measure}]
\label{lem:Chernoff}
    Let $\rv{X}_1, \ldots, \rv{X}_n$ be $n$ independent Bernoulli random variables,
    $\rv{X} \defeq \sum_{1 \leq i \leq n} \rv{X}_i$, and
    $\mu \defeq \E[\rv{X}]$.
    Then, for any real $\delta \in [0,1]$,
    \begin{align}
        \Pr\Bigl[\bigl|\rv{X}-\mu\bigr| \geq \delta \mu\Bigr]
        \leq 2 \exp\left(-\frac{\delta^2}{3}\mu\right).
    \end{align}
\end{lemma}

A function $g \colon \calX_1 \times \cdots \times \calX_n \to \bbR$ satisfies the 
\defi{bounded differences property} with bounds $c_1,\ldots, c_n$ if for each $i \in [n]$, 
replacing the assignment to the \nth{$i$} coordinate changes
the value of $g$ by at most $c_i$; namely,
for any $(x_1, \ldots, x_n), (y_1, \ldots, y_n) \in \calX_1 \times \cdots \times \calX_n$
such that $x_j = y_j$ for every $j \neq i$,
\begin{align}
    \bigl| g(x_1, \ldots, x_n) - g(y_1, \ldots, y_n) \bigr| \leq c_i.
\end{align}

\begin{lemma}[McDiarmid's inequality \cite{mcdiarmid1989method}]
\label{lem:McDiarmid}
Let $g \colon \calX_1 \times \cdots \times \calX_n \to \bbR$ be a function that 
satisfies the bounded differences property with bounds $c_1,\ldots, c_n$.
Let $\rv{x}_1, \ldots, \rv{x}_n$ be $n$ independent random variables with 
$\rv{x}_i \in \calX_i$ for every $i \in [n]$, and
$\rv{Y} \defeq g(\rv{x}_1, \ldots, \rv{x}_n)$.
Then, for any real $t > 0$,
\begin{align}
    \Pr\Bigl[
        \bigl|\rv{Y} - \E[\rv{Y}]\bigr| \geq t
    \Bigr]
    \leq 2 \exp\left(
        -\frac{2t^2}{\sum_{1 \leq i \leq n}c_i^2}
    \right).
\end{align}
\end{lemma}

\section{Tolerant $q$-query Direct Product Tester}
\label{sec:TDP}

In this section, we introduce and analyze the tolerant $q$-query direct product tester.
Let $F \colon \binom{V}{k} \times \Rnd \to \Sigma^k$ be a $k$-set function, where $\Rnd$ is the set of strings.
We will think of $F$ as the following randomized oracle:
for a given set $X \in \binom{V}{k}$,
the oracle samples a string $\rv{\rnd} \in \Rnd$ and returns $F(X;\rv{\rnd})$.
For each integer $q \geq 2$,
the \defi{tolerant $q$-query direct product tester} $\T_q$,
parameterized by the \defi{intersection parameter} $\ell$ and the \defi{tolerance parameter} $m$ with $1 \leq m \leq \ell \leq k$,
is described below.

\begin{itembox}[l]{\textbf{Tolerant $q$-query direct product tester $\T_q$ for a $k$-set function}}
\begin{algorithmic}[1]
    \item[\textbf{Input:}]
        positive integers $\ell$ and $m$ with $1 \leq m \leq \ell \leq k$.
    \item[\textbf{Oracle access:}]
        a $k$-set function $F \colon \binom{V}{k} \times \Rnd \to \Sigma^k$.
    \State sample $\rv{X}_1$ from $\binom{V}{k}$ uniformly.
    \For{\textbf{each} $i \in [q-1]$}
        \State sample $\rv{I}_i$ from $\binom{\rv{X}_i}{\ell}$ uniformly.
        \State sample $\rv{A}_{i+1}$ from $\binom{V \setminus \rv{I}_i}{k-\ell}$ uniformly.
        \State let $\rv{X}_{i+1} \defeq \rv{I}_i \cup \rv{A}_{i+1}$.
    \EndFor
    \For{\textbf{each} $i \in [q]$}
        \State sample $\rv{\rnd}_i$ from $\Rnd$ uniformly.
        \State read $F(\rv{X}_i; \rv{\rnd}_i)$.
    \EndFor
    \For{\textbf{each} $i \in [q-1]$}
        \If{$\far{
            F(\rv{X}_i; \rv{\rnd}_i)|_{\rv{I}_{i}}}{
            F(\rv{X}_{i+1}; \rv{\rnd}_{i+1})|_{\rv{I}_{i}}}{
            m/\ell}$}
            \State \Return $0$.
        \EndIf
    \EndFor
    \State \Return $1$.
\end{algorithmic}
\end{itembox}

\noindent
The main result of this section is the following.

\begin{theorem}
\label{thm:TDP}
Let $q \geq 2$ be an integer.
Let $F \colon \binom{V}{k} \times \Rnd \to \Sigma^k$ be a $k$-set function and $n \defeq |V|$.
Let $\ell$ and $m$ be positive integers with
$1 \leq m \leq \ell \leq k$ such that
$\ell = \Theta(k^\e)$ for a real $\e \in (0,1)$ and
$m = \ell^{1-\Theta(1)}$.
Define 
\begin{align}
\label{eq:TDP:epsilon}
    \epsilon_{k,\ell,m}
    \defeq \tilde{\Theta}\left(\max\left\{
        \frac{\ell}{k},
        \frac{m^{\frac{1}{8}}}{\ell^{\frac{1}{8}}},
        \frac{1}{\ell^{\frac{1}{16}}}
    \right\}\right),
\end{align}
where $\tilde{\Theta}$ hides a $\polylog(k)$ factor.
If $k$ and $n$ are sufficiently large, the following hold:

\begin{description}
    \item[(Completeness)] 
        Suppose that there exists a function $f \colon V \to \Sigma$ such that
        \begin{align}
            p \defeq
            \Pr_{(\rv{X}, \rv{\rnd}) \in \binom{V}{k} \times \Rnd}\left[
                \close{F(\rv{X};\rv{\rnd})}{f^k(\rv{X})}{m/(2k)}
            \right]
            = \omega\left(\frac{\ell}{k}\right).
        \end{align}
        Then, $\T_q(\ell,m)$ accepts $F$ with probability at least
        \begin{align}
            p^q - \bigO\left(\tfrac{\ell}{k}\right)p.
        \end{align}
        Moreover, if $p = 1$, then $\T_q(\ell,m)$ accepts $F$ with probability $1$.
    \item[(Soundness)] 
        Suppose that $\T_q(\ell,m)$ accepts $F$
        with probability $p \geq \sqrt{\epsilon_{k,\ell,m}}$.
        Then, there exists a function $f \colon V \to \Sigma$ such that
        \begin{align}
            \Pr_{(\rv{X},\rv{\rnd}) \in \binom{V}{k} \times \Rnd}\left[
                \close{F(\rv{X};\rv{\rnd})}{f^k(\rv{X})}{\epsilon_{k,\ell,m}}
            \right] \geq p^{\frac{1}{q-1}}(1-k^{-\Omega(1)}).
        \end{align}
\end{description}
\end{theorem}

Hereafter,
let $F \colon \binom{V}{k} \times \Rnd \to \Sigma^k$ be a $k$-set function with $V \defeq [n]$.
Let $\ell$ and $m$ be positive integers with
$1 \leq m \leq \ell \leq k$ such that
$\ell = \Theta(k^\e)$ for a real $\e \in (0,1)$ and
$m = \ell^{1-\Theta(1)}$.
We write
\begin{align}
    \bigl(
        \sq{\rv{I}} = (\rv{I}_i)_{i \in [q-1]},
        \sq{\rv{X}} = (\rv{X}_i)_{i \in [q]},
        \sq{\rv{\rnd}} = (\rv{\rnd}_i)_{i \in [q]}
    \bigr) \sim \T_q(\ell,m)
\end{align}
for the random variables sampled by $\T_q(\ell,m)$.

The remainder of this section is organized as follows.
In \cref{sec:TDP:Johnson}, we characterize the query pattern of $\T_q(\ell,m)$.
In \cref{sec:TDP:complete}, we prove the completeness part of \cref{thm:TDP}.
In \cref{sec:TDP:sound}, we prove the soundness part of \cref{thm:TDP}.

\subsection{Characterization of the Query Pattern}
\label{sec:TDP:Johnson}

We prove the following characterization of the query pattern of $\T_q(\ell,m)$.

\begin{proposition}
\label{lem:Johnson}
    For any sufficiently large integer $k$ and
    for any set $\calS \subseteq \binom{V}{k} \times \Rnd$,
    \begin{align}
    \label{eq:Johnson:POI1}
    \begin{aligned}
        \Pr_{(\sq{\rv{I}}, \sq{\rv{X}}, \sq{\rv{\rnd}}) \sim \T_q(\ell,m)}\Biggl[
            \bigwedge_{1 \leq i \leq q} (\rv{X}_i, \rv{\rnd}_i) \in \calS
        \Biggr]
        & \leq
        \left(\frac{|\calS|}{\bigl|\binom{V}{k} \times \Rnd\bigr|}\right)^q
        + \bigO_q\left(\frac{\ell}{k}\right)
            \left(\frac{|\calS|}{\bigl|\binom{V}{k} \times \Rnd\bigr|}\right)
        + \bigO_{q,k}\left(\frac{1}{n}\right).
    \end{aligned}
    \end{align}
    Moreover, if $\frac{|\calS|}{\bigl|\binom{V}{k} \times \Rnd\bigr|} = \omega\left(\frac{\ell}{k}\right)$,
    then
    \begin{align}
        \label{eq:Johnson:POI2}
        \Pr_{(\sq{\rv{I}}, \sq{\rv{X}}, \sq{\rv{\rnd}}) \sim \T_q(\ell,m)}\Biggl[
            \bigwedge_{1 \leq i \leq q} (\rv{X}_i, \rv{\rnd}_i) \in \calS
        \Biggr]
        & \geq
        \left(\frac{|\calS|}{\bigl|\binom{V}{k} \times \Rnd\bigr|}\right)^q
        - \bigO_q\left(\frac{\ell}{k}\right)
            \left(\frac{|\calS|}{\bigl|\binom{V}{k} \times \Rnd\bigr|}\right)
        -\bigO_{q,k}\left(\frac{1}{n}\right).
    \end{align}
\end{proposition}

For the purpose of proving \cref{lem:Johnson}, we introduce the Johnson graph.

\begin{definition}
\label{def:Johnson}
For positive integers $n$, $k$, and $\ell$ with $n \geq k \geq \ell$,
the \defi{Johnson graph} is defined as
a graph $J(n,k,\ell)$ such that
\begin{align}
\begin{aligned}
    V(J(n,k,\ell)) & \defeq \binom{[n]}{k}, \\
    E(J(n,k,\ell)) & \defeq \Bigl\{
        \{X_1, X_2\} \in \tbinom{V(J(n,k,\ell))}{2}
        \Bigm|
        |X_1 \cap X_2| = \ell
    \Bigr\}.
\end{aligned}
\end{align}
\end{definition}

For a graph $G$ and a finite set $\Rnd$,
let $G \times \Rnd$ denote a graph such that
\begin{align}
\begin{aligned}
    V(G \times \Rnd) & \defeq V(G) \times \Rnd, \\
    E(G \times \Rnd) & \defeq \Bigl\{
        \bigl\{ (x_1, r_1), (x_2, r_2) \bigr\} \in \tbinom{V(G \times \Rnd)}{2}
        \Bigm|
        \{ x_1, x_2 \} \in E(G)
    \Bigr\}.
\end{aligned}
\end{align}

We first show that the probability of interest---the left-hand side of \cref{eq:Johnson:POI1,eq:Johnson:POI2}---can be approximated by the hitting probability of the length-$q$ random walk on $J(n,k,\ell)$.

\begin{lemma}
\label{lem:Johnson:lbub}
For any set $\calS \subseteq \binom{V}{k} \times \Rnd$,
\begin{align}
    \left|
        \Pr_{(\sq{\rv{I}}, \sq{\rv{X}}, \sq{\rv{\rnd}}) \sim \T_q(\ell,m)}\Biggl[
            \bigwedge_{1 \leq i \leq q} (\rv{X}_i, \rv{\rnd}_i) \in \calS
        \Biggr]
        - \Pr_{((\rv{X}_1, \rv{\rnd}_1), \ldots, (\rv{X}_q, \rv{\rnd}_q)) \sim \RW_q(J(n,k,\ell)\times\Rnd)}\Biggl[
            \bigwedge_{1 \leq i \leq q} (\rv{X}_i, \rv{\rnd}_i) \in \calS
        \Biggr]
    \right|
    \leq \bigO_{q,k}\left(\frac{1}{n}\right).
\end{align}
\end{lemma}
\begin{proof} 
Conditioned on the event that
$|\rv{X}_i \cap \rv{X}_{i+1}| = \ell$ for every $i \in [q-1]$,
we find $( \rv{X}_1, \ldots, \rv{X}_q )$ to be the length-$q$ random walk on $J(n,k,\ell)$; thus,
$((\rv{X}_1, \rv{\rnd}_1), \ldots, (\rv{X}_q, \rv{\rnd}_q))$
is equal to the length-$q$ random walk on $J(n,k,\ell) \times \Rnd$.
Therefore,
\begin{align}
\label{eq:Johnson:lbub}
\begin{aligned}
    & \Pr_{(\sq{\rv{I}}, \sq{\rv{X}}, \sq{\rv{\rnd}}) \sim \T_q(\ell,m)}\Biggl[
        \bigwedge_{1 \leq i \leq q} (\rv{X}_i, \rv{\rnd}_i) \in \calS
        \Biggm|
        \bigwedge_{1 \leq i \leq q-1}
        |\rv{X}_i \cap \rv{X}_{i+1}| = \ell
    \Biggr] \\
    & = \Pr_{((\rv{X}_1, \rv{\rnd}_1), \ldots, (\rv{X}_q, \rv{\rnd}_q)) \in \RW_q(J(n,k,\ell) \times \Rnd)}\Biggl[
        \bigwedge_{1 \leq i \leq q} (\rv{X}_i, \rv{\rnd}_i) \in \calS
    \Biggr].
\end{aligned}
\end{align}
By Markov's inequality, we derive for each $i \in [q-1]$,
\begin{align}
\begin{aligned}
    \Pr_{(\sq{\rv{I}}, \sq{\rv{X}}, \sq{\rv{\rnd}}) \sim \T_q(\ell,m)}\Bigl[
        |\rv{X}_i \cap \rv{X}_{i+1}| \neq \ell
    \Bigr]
    & = \Pr_{(\sq{\rv{I}}, \sq{\rv{X}}, \sq{\rv{\rnd}}) \sim \T_q(\ell,m)}\Bigl[
        |(\rv{X}_i \setminus \rv{I}_i) \cap (\rv{X}_{i+1} \setminus \rv{I}_i)| \geq 1
    \Bigr] \\
    & \leq \E_{(\sq{\rv{I}}, \sq{\rv{X}}, \sq{\rv{\rnd}}) \sim \T_q(\ell,m)}\Bigl[
        |(\rv{X}_i \setminus \rv{I}_i) \cap (\rv{X}_{i+1} \setminus \rv{I}_i)|
    \Bigr]
    = \frac{(k-\ell)^2}{n-k},
\end{aligned}
\end{align}
where we used the fact that
$\rv{X}_i \setminus \rv{I}_i$ and $\rv{X}_{i+1} \setminus \rv{I}_i$ are independent of each other
and uniformly distributed in $\binom{V \setminus \rv{I}_i}{k-\ell}$.
By the union bound, we obtain
\begin{align}
\label{eq:Johnson:lbub:cap}
\begin{aligned}
    \Pr_{(\sq{\rv{I}}, \sq{\rv{X}}, \sq{\rv{\rnd}}) \sim \T_q(\ell,m)}\Biggl[
        \bigwedge_{1 \leq i \leq q-1}
        |\rv{X}_i \cap \rv{X}_{i+1}| = \ell
    \Biggr]
    & \geq 1 - \Pr_{(\sq{\rv{I}}, \sq{\rv{X}}, \sq{\rv{\rnd}}) \sim \T_q(\ell,m)}\Biggl[
        \bigvee_{1 \leq i \leq q-1}
        |\rv{X}_i \cap \rv{X}_{i+1}| \neq \ell
    \Biggr] \\
    & \geq 1 - \sum_{1 \leq i \leq q-1}
    \Pr_{(\sq{\rv{I}}, \sq{\rv{X}}, \sq{\rv{\rnd}}) \sim \T_q(\ell,m)}\Bigl[
        |\rv{X}_i \cap \rv{X}_{i+1}| \neq \ell
    \Bigr] \\
    & \geq 1 - \sum_{1 \leq i \leq q-1} \frac{(k-\ell)^2}{n-k}
    = 1 - \bigO_{q,k}\left(\frac{1}{n}\right),
\end{aligned}
\end{align}
On the one hand, we have
\begin{align}
\begin{aligned}
    \text{\cref{eq:Johnson:lbub}}
    & \geq \Pr_{(\sq{\rv{I}}, \sq{\rv{X}}, \sq{\rv{\rnd}}) \sim \T_q(\ell,m)}\left[
        \bigwedge_{1 \leq i \leq q} (\rv{X}_i, \rv{\rnd}_i) \in \calS
        \text{ and }
        \bigwedge_{1 \leq i \leq q-1} |\rv{X}_i \cap \rv{X}_{i+1}| = \ell
    \right] \\
    & \geq \Pr_{(\sq{\rv{I}}, \sq{\rv{X}}, \sq{\rv{\rnd}}) \sim \T_q(\ell,m)}\Biggl[
        \bigwedge_{1 \leq i \leq q} (\rv{X}_i, \rv{\rnd}_i) \in \calS
    \Biggr]
    + \Pr_{(\sq{\rv{I}}, \sq{\rv{X}}, \sq{\rv{\rnd}}) \sim \T_q(\ell,m)}\Biggl[
        \bigwedge_{1 \leq i \leq q-1} |\rv{X}_i \cap \rv{X}_{i+1}| = \ell
    \Biggr] - 1 \\
    & \underbrace{\geq}_{\text{\cref{eq:Johnson:lbub:cap}}} \Pr_{(\sq{\rv{I}}, \sq{\rv{X}}, \sq{\rv{\rnd}}) \sim \T_q(\ell,m)}\Biggl[
        \bigwedge_{1 \leq i \leq q} (\rv{X}_i, \rv{\rnd}_i) \in \calS
    \Biggr] - \bigO_{q,k}\left(\frac{1}{n}\right),
\end{aligned}
\end{align}
which implies the ``upper-bound'' part.
On the other hand, we have
\begin{align}
\begin{aligned}
    & \Pr_{(\sq{\rv{I}}, \sq{\rv{X}}, \sq{\rv{\rnd}}) \sim \T_q(\ell,m)}\Biggl[
        \bigwedge_{1 \leq i \leq q} (\rv{X}_i, \rv{\rnd}_i) \in \calS
    \Biggr] \\
    & \geq \Pr_{(\sq{\rv{I}}, \sq{\rv{X}}, \sq{\rv{\rnd}}) \sim \T_q(\ell,m)}\Biggl[
        \bigwedge_{1 \leq i \leq q} (\rv{X}_i, \rv{\rnd}_i) \in \calS
        \text{ and }
        \bigwedge_{1 \leq i \leq q-1} |\rv{X}_i \cap \rv{X}_{i+1}| = \ell
    \Biggr] \\
    & \geq \Pr_{(\sq{\rv{I}}, \sq{\rv{X}}, \sq{\rv{\rnd}}) \sim \T_q(\ell,m)}\Biggl[
        \bigwedge_{1 \leq i \leq q} (\rv{X}_i, \rv{\rnd}_i) \in \calS
        \Biggm|
        \bigwedge_{1 \leq i \leq q-1} |\rv{X}_i \cap \rv{X}_{i+1}| = \ell
    \Biggr] \cdot
    \Pr_{(\sq{\rv{I}}, \sq{\rv{X}}, \sq{\rv{\rnd}}) \sim \T_q(\ell,m)}\Biggl[
        \bigwedge_{1 \leq i \leq q-1} |\rv{X}_i \cap \rv{X}_{i+1}| = \ell
    \Biggr] \\
    & \underbrace{\geq}_{\text{\cref{eq:Johnson:lbub:cap}}}
        \text{\cref{eq:Johnson:lbub}} \cdot \left(1-\bigO_{q,k}\left(\frac{1}{n}\right)\right)
        \geq \text{\cref{eq:Johnson:lbub}} - \bigO_{q,k}\left(\frac{1}{n}\right),
\end{aligned}
\end{align}
which implies the ``lower-bound'' part, as desired.
\end{proof}

We then bound the hitting probability of the length-$q$ random walk on the Johnson graph.
\begin{lemma}
\label{lem:Johnson:hitting}
    For a vertex set $\calS \subseteq V(J(n,k,\ell)\times\Rnd)$
    with density $p \defeq \frac{|\calS|}{\binom{n}{k}|\Rnd|}$,
    \begin{align}
        \Pr_{((\rv{X}_1, \rv{\rnd}_1), \ldots, (\rv{X}_q, \rv{\rnd}_q)) \sim \RW_q(J(n,k,\ell)\times\Rnd)}\Biggl[
            \bigwedge_{1 \leq i \leq q} (\rv{X}_i, \rv{\rnd}_i) \in \calS
        \Biggr]
        \leq p^q + \bigO_q\left(\frac{\ell}{k}\right)p.
    \end{align}
    Moreover, if $p = \omega\left(\frac{\ell}{k}\right)$, then
    \begin{align}
        \Pr_{((\rv{X}_1, \rv{\rnd}_1), \ldots, (\rv{X}_q, \rv{\rnd}_q)) \sim \RW_q(J(n,k,\ell)\times\Rnd)}\Biggl[
            \bigwedge_{1 \leq i \leq q} (\rv{X}_i, \rv{\rnd}_i) \in \calS
        \Biggr]
        \geq p^q - \bigO_q\left(\frac{\ell}{k}\right)p.
    \end{align}
\end{lemma}
\begin{proof} 
    Observe that
    $J(n,k,\ell) \times \Rnd$ is equal to
    the tensor product of $J(n,k,\ell)$ and
    the complete graph with self-loops $K_{|\Rnd|}^\circ$.
    Since
    $\lambda(J(n,k,\ell)) = \bigO\left(\frac{\ell}{k}\right)$ \cite[Lemma~5.3]{dinur2008locally},
    we have
    \begin{align}
        \lambda\bigl(J(n,k,\ell) \times \Rnd\bigr)
        = \lambda\bigl(J(n,k,\ell) \otimes K_{|\Rnd|}^\circ\bigr)
        = \max\bigl\{ \lambda(J(n,k,\ell)), \lambda(K_{|\Rnd|}^\circ) \bigr\}
        = \bigO\left(\tfrac{\ell}{k}\right).
    \end{align}
    Let $\lambda \defeq \lambda(J(n,k,\ell) \times \Rnd)$.
    On the one hand, by \cref{thm:hitting}, we have
    \begin{align}
    \begin{aligned}
        & \Pr_{((\rv{X}_1, \rv{\rnd}_1), \ldots, (\rv{X}_q, \rv{\rnd}_q)) \sim \RW_q(J(n,k,\ell) \times \Rnd)}\Biggl[
            \bigwedge_{1 \leq i \leq q} (\rv{X}_i, \rv{\rnd}_i) \in \calS
        \Biggr] \\
        & \leq p (p + \lambda)^{q-1}
        \leq p(p^{q-1} + 2^{q-1}\lambda)
        = p^q + \bigO_q\left(\tfrac{\ell}{k}\right)p.
    \end{aligned}
    \end{align}
    On the other hand, if $p = \omega\left(\frac{\ell}{k}\right)$,
    by \cref{thm:hitting}, we have
    \begin{align}
    \begin{aligned}
        & \Pr_{((\rv{X}_1,\rv{\rnd}_1), \ldots, (\rv{X}_q,\rv{\rnd}_q)) \sim \RW_q(J(n,k,\ell)\times\Rnd)}\Biggl[
            \bigwedge_{1 \leq i \leq q} (\rv{X}_i, \rv{\rnd}_i) \in \calS
        \Biggr] \\
        & \geq p (p - 2\lambda)^{q-1}
        \geq p \Bigl(p^{q-1} - 2^{q-1} \cdot \max\bigl\{2\lambda, (2\lambda)^{q-1}\bigr\}\Bigr)
        \geq p^q - \bigO_q\left(\tfrac{\ell}{k}\right)p,
    \end{aligned}
    \end{align}
    as desired.
\end{proof}

By \cref{lem:Johnson:lbub,lem:Johnson:hitting},
we obtain \cref{lem:Johnson}, as desired.

\subsection{\texorpdfstring{%
Completeness Part of \cref{thm:TDP}
}{%
Completeness Part of Theorem \ref{thm:TDP}
}}
\label{sec:TDP:complete}

We prove the completeness part of \cref{thm:TDP}.

\begin{proof}[Proof of the completeness part of \cref{thm:TDP}]
Suppose that there exists a function $f \colon V \to \Sigma$ such that
\begin{align}
    p \defeq
    \Pr_{(\rv{X}, \rv{\rnd}) \in \binom{V}{k} \times \Rnd}\left[
        \close{F(\rv{X};\rv{\rnd})}{f^k(\rv{X})}{m/(2k)}
    \right]
    = \omega\left(\frac{\ell}{k}\right).
\end{align}
Define
\begin{align}
    \calS \defeq \left\{
        (X, r) \in \tbinom{V}{k} \times \Rnd
        \Bigm|
        \close{F(X;r)}{f^k(X)}{m/(2k)}
    \right\}.
\end{align}
Note that $|\calS| = p \bigl|\binom{V}{k} \times \Rnd\bigr|$.
Observe that
\begin{align}
\begin{aligned}
    & \Pr\Bigl[\T_q^F(\ell,m) = 1\Bigr] \\
    & = \Pr_{(\sq{\rv{I}}, \sq{\rv{X}}, \sq{\rv{\rnd}}) \sim \T_q(\ell,m)}\Biggl[
        \bigwedge_{1 \leq i \leq q-1}
        \close{F(\rv{X}_i; \rv{\rnd}_i)|_{\rv{I}_i}}{F(\rv{X}_{i+1}; \rv{\rnd}_{i+1})|_{\rv{I}_i}}{m/\ell}
    \Biggr] \\
    & \geq \Pr_{(\sq{\rv{I}}, \sq{\rv{X}}, \sq{\rv{\rnd}}) \sim \T_q(\ell,m)}\Biggl[
        \bigwedge_{1 \leq i \leq q-1}
        \close{F(\rv{X}_i; \rv{\rnd}_i)|_{\rv{I}_i}}{f^k(\rv{X}_i)|_{\rv{I}_i}}{m/(2\ell)}
        \text{ and }
        \close{f^k(\rv{X}_{i+1})|_{\rv{I}_i}}{F(\rv{X}_{i+1}; \rv{\rnd}_{i+1})|_{\rv{I}_i}}{m/(2\ell)}
    \Biggr] \\
    & \geq \Pr_{(\sq{\rv{I}}, \sq{\rv{X}}, \sq{\rv{\rnd}}) \sim \T_q(\ell,m)}\Biggl[
        \bigwedge_{1 \leq i \leq q-1}
        \close{F(\rv{X}_i; \rv{\rnd}_i)}{f^k(\rv{X}_i)}{m/(2k)}
        \text{ and }
        \close{f^k(\rv{X}_{i+1})}{F(\rv{X}_{i+1}; \rv{\rnd}_{i+1})}{m/(2k)}
    \Biggr] \\
    & = \Pr_{(\sq{\rv{I}}, \sq{\rv{X}}, \sq{\rv{\rnd}}) \sim \T_q(\ell,m)}\Biggl[
        \bigwedge_{1 \leq i \leq q}
        (\rv{X}_i, \rv{\rnd}_i) \in \calS
    \Biggr].
\end{aligned}
\end{align}
By \cref{lem:Johnson},
for sufficiently large $n$,
we derive
\begin{align}
\begin{aligned}
    \Pr_{(\sq{\rv{I}}, \sq{\rv{X}}, \sq{\rv{\rnd}}) \sim \T_q(\ell,m)}\Biggl[
        \bigwedge_{1 \leq i \leq q}
        (\rv{X}_i, \rv{\rnd}_i) \in \calS
    \Biggr]
    & \geq \left(\frac{|\calS|}{\bigl|\binom{V}{k} \times \Rnd\bigr|}\right)^q
            - \bigO_q\left(\frac{\ell}{k}\right)
                \left(\frac{|\calS|}{\bigl|\binom{V}{k} \times \Rnd\bigr|}\right)
        -\bigO_{q,k}\left(\frac{1}{n}\right), \\
    & \geq p^q - \bigO_q\left(\frac{\ell}{k}\right)p,
\end{aligned}
\end{align}
as desired.
\end{proof}

\subsection{\texorpdfstring{%
Soundness Part of \cref{thm:TDP}
}{%
Soundness Part of Theorem \ref{thm:TDP}
}}
\label{sec:TDP:sound}

We prove the soundness part of \cref{thm:TDP}.
Suppose that $\T_q(\ell,m)$ accepts $F$ with probability $p = \omega\left(\frac{\ell}{k}\right)$.
We first show that there exists a function $f \colon V \to \Sigma$ such that
$F(\rv{X};\rv{\rnd})$ agrees with $f^k(\rv{X})$ on
all but a $\max\left\{\frac{\ell}{k},\frac{m}{\ell}\right\}^{\Omega(1)}$-fraction of coordinates
with probability $\Omega(p^6)$,
whose proof is deferred to \cref{app:IKW12}.

\begin{theorem}[$*$]
\label{thm:IKW12}
Let $q \geq 2$ be an integer.
Let $F \colon \binom{V}{k} \times \Rnd \to \Sigma^k$ be a $k$-set function and
$n \defeq |V|$.
Let $\ell$ and $m$ be positive integers with $1 \leq m \leq \ell \leq k$.
If $k$ and $n$ are sufficiently large, the following holds.
Suppose that $\T_q(\ell,m)$ accepts $F$
with probability $p = \omega\left(\frac{\ell}{k}\right)$.
Then, there exists a function $f \colon V \to \Sigma$ such that
\begin{align}
    \Pr_{(\rv{X}, \rv{\rnd}) \in \binom{V}{k} \times \Rnd}\left[
        \close{F(\rv{X};\rv{\rnd})}{f^k(\rv{X})}{\delta_{k,\ell,m}}
    \right] = \Omega(p^6),
\end{align}
where $\delta_{k,\ell,m}$ is defined as
\begin{align}
    \delta_{k,\ell,m}
    = \tilde{\Theta}\left(\max\left\{\frac{\ell}{k},\frac{m}{\ell}\right\}\right).
\end{align}
\end{theorem}

\begin{remark}
If $\T_q(\ell,m)$ accepts $F$ with probability at least $p$,
then $\T_2(\ell,m)$ also accepts $F$ with probability at least $p$.
Therefore, it is sufficient to prove \cref{thm:IKW12} only when $q = 2$.
The proof of \cref{thm:IKW12} with $q = 2$ is obtained by extending the analysis of
the 2-query direct product tester due to \cite[Theorem~1.3]{impagliazzo2012new}.
\end{remark}

By applying \cref{thm:IKW12}, we then show that
there exists a function $f \colon V \to \Sigma$ such that
$F(\rv{X};\rv{\rnd})$ agrees with $f^k(\rv{X})$ on
all but a $\max\left\{\frac{\ell}{k},\frac{m}{\ell}\right\}^{\Omega(1)}$-fraction of coordinates
with probability $\gtrapprox p^{\frac{1}{q-1}}$.

\begin{theorem}
\label{thm:DG08}
Let $q \geq 2$ be an integer.
Let $F \colon \binom{V}{k} \times \Rnd \to \Sigma^k$ be a $k$-set function and $n \defeq |V|$.
Let $\ell$ and $m$ be positive integers with
$1 \leq m \leq \ell \leq k$ such that
$\ell = \Theta(k^\e)$ for a real $\e \in (0,1)$ and
$m = \ell^{1-\Theta(1)}$.
Define 
\begin{align}
    \epsilon_{k,\ell,m}
    \defeq \tilde{\Theta}\left(\max\left\{
        \frac{\ell}{k},
        \frac{m^{\frac{1}{8}}}{\ell^{\frac{1}{8}}},
        \frac{1}{\ell^{\frac{1}{16}}}
    \right\}\right).
\end{align}
If $k$ and $n$ are sufficiently large, the following holds.
Suppose that $\T_q(\ell,m)$ accepts $F$ with probability $p \geq \sqrt{\epsilon_{k,\ell,m}}$.
Then, there exists a function $f \colon V \to \Sigma$ such that
\begin{align}
    \Pr_{\substack{
        (\rv{X},\rv{\rnd}) \in \binom{V}{k} \times \Rnd
    }}\left[
        \close{F(\rv{X};\rv{\rnd})}{f^k(\rv{X})}{\epsilon_{k,\ell,m}}
    \right] \geq p^{\frac{1}{q-1}}\bigl(1-k^{-\Omega(1)}\bigr).
\end{align}
\end{theorem}

The soundness part of \cref{thm:TDP} follows from \cref{thm:DG08}.
The remainder of this subsection is devoted to the proof of \cref{thm:DG08}.

\paragraph{Proof of \cref{thm:DG08}.}
Let $q \geq 2$ be an integer.
Let $F \colon \binom{V}{k} \times \Rnd \to \Sigma^k$ be a $k$-set function and $n \defeq |V|$.
Let $\ell$ and $m$ be positive integers with
$1 \leq m \leq \ell \leq k$ such that
$\ell = \Theta(k^\e)$ for some real $\e \in (0,1)$ and
$m = \ell^{1-\Theta(1)}$.
Define $\smol$ and $\step$ as
\begin{align}
\label{eq:DG08:smol}
\begin{aligned}
    \smol & \defeq \Theta\left(\max\left\{
        \frac{\ell}{k},
        \frac{m^{\frac{1}{8}}}{\ell^{\frac{1}{8}}},
        \frac{1}{\ell^{\frac{1}{16}}}
    \right\} \log k\right), \\
    \step & \defeq \smol^8.
\end{aligned}
\end{align}
Note that $\smol = k^{-\Theta(1)}$ and $\step = k^{-\Theta(1)}$.
Hereafter, we use $\T_q$ to denote $\T_q(\ell,m)$.
We assume that $k$ and $n$ are sufficiently large so that \cref{lem:Johnson,thm:IKW12} can be applied.

For a $k$-set function $F \colon \binom{V}{k} \times \Rnd \to \Sigma^k$,
a function $f \colon V \to \Sigma$, and
a real $\delta \in (0,1)$,
we define $\supp_\delta(F, f^k)$ as
the set of pairs $(X, \rnd) \in \binom{V}{k} \times \Rnd$ such that
$F(X; \rnd)$ and $f^k(X)$ are $\delta$-close; namely,
\begin{align}
    \supp_\delta(F, f^k) \defeq \left\{
        (X, \rnd) \in \tbinom{V}{k} \times \Rnd \Bigm|
        \close{F(X; \rnd)}{f^k(X)}{\delta}
    \right\}.
\end{align}

Suppose that $\T_q$ accepts $F$ with probability $p \geq \sqrt{\smol}$.
By \cref{thm:IKW12}, there exists a function $f \colon V \to \Sigma$ such that
\begin{align}
    \Pr_{(\rv{X},\rv{\rnd}) \in \binom{V}{k} \times \Rnd}\left[
        \close{F(\rv{X};\rv{\rnd})}{f^k(\rv{X})}{
        \tilde{\Theta}\left(\max\left\{\frac{\ell}{k}, \frac{m}{\ell}\right\}\right)}
    \right]
    = \Omega(p^6).
\end{align}
\noindent
Consider
$T$ $k$-set functions $F_1, \ldots, F_T \colon \binom{V}{k} \times \Rnd \to \Sigma^k$,
$T$ functions $f_1, \ldots, f_T \colon V \to \Sigma$, and
$T$ positive reals $\radi_1, \ldots, \radi_T$
generated by the following randomized procedure due to \cite{dinur2008locally}.

\begin{itembox}[l]{\textbf{Construction of $F_t$'s, $f_t$'s, and $\radi_t$'s \cite{dinur2008locally}}}
\begin{algorithmic}[1]
    \item[\textbf{Input:}] a $k$-set function $F \colon \binom{V}{k} \times \Rnd \to \Sigma^k$.
    \State let $F_1 \defeq F$.
    \For{\textbf{each} $t \geq 1$}
        \If{$\T_q$ accepts $F_t$ with probability at least $\smol$}
            \State let $f_t \colon V \to \Sigma$ be a function such that
            $\bigl|\supp_{\smol}(F_t, f_t^k)\bigr|$ is maximized.
            \LComment{$\bigl|\supp_{\smol}(F_t, f_t^k)\bigr| \geq \Omega(\smol^6) \bigl|\binom{V}{k} \times \Rnd\bigr|$.}
            \State let $\radi_t$ be the smallest real of the form $\smol + i \step$ for a nonnegative integer $i$ such that
            \begin{align}
            \label{eq:DG08:radi}
                \bigl|\supp_{\radi_t + \step}(F, f_t^k) \setminus \supp_{\radi_t}(F, f_t^k)\bigr|
                < \smol^7 \bigl|\tbinom{V}{k} \times \Rnd\bigr|.
            \end{align}
            \State let $F_{t+1} \colon \binom{V}{k} \times \Rnd \to \Sigma^k$
            be a $k$-set function determined randomly as follows:
            \begin{align}
                F_{t+1}(X;\rnd) \defeq 
                \begin{cases}
                    \text{a random value in } \Sigma^k & \text{if } (X,\rnd) \in \supp_{\radi_t}(F_t, f_t^k), \\
                    F_t(X;\rnd) & \text{otherwise}.
                \end{cases}
            \end{align}
        \Else
            \State let $T \defeq t-1$.
            \State \textbf{terminate}.
        \EndIf
    \EndFor
\end{algorithmic}
\end{itembox}

We first show that the above procedure terminates within $\bigO(\smol^{-6})$ iterations.
\begin{observation}[\cite{dinur2008locally}]
\label{obs:DG08:1:iter}
For every $t \in [T]$, $\radi_t \in [\smol, 2\smol]$.
Moreover, $T = \bigO(\smol^{-6})$ with high probability.
\end{observation}

\noindent
To prove \cref{obs:DG08:1:iter}, we use the following claim.

\begin{claim}
\label{clm:DG08:1:iter}
    Let $\rv{F} \colon \binom{V}{k}\times \Rnd \to \Sigma^k$ be a random function
    such that $\rv{F}(X;\rnd)$ is selected from $\Sigma^k$ uniformly.
    Then,
    \begin{align}
        \Pr_{\rv{F} \colon \binom{V}{k}\times \Rnd \to \Sigma^k}\left[
            \exists f \colon V \to \Sigma \text{ s.t. }
            \frac{|\supp_{2\smol}(\rv{F}, f^k)|}{\bigl|\binom{V}{k} \times \Rnd\bigr|}
            \geq \exp\bigl(-\Theta(k)\bigr)
        \right] \leq \exp\Bigl(-\Omega\bigl(\bigl|\tbinom{V}{k}\bigr|\bigr)\Bigr).
    \end{align}
\end{claim}
\begin{proof} 
Fix a function $f:V \to \Sigma$.
For each pair $(X,\rnd) \in \binom{V}{k} \times \Rnd$,
we introduce the indicator variable
\begin{align}
    \rv{I}_{X,\rnd} \defeq \left\llbracket \close{\rv{F}(X;\rnd)}{f^k(X)}{2\smol} \right\rrbracket.
\end{align}
Since $\rv{F}(X;\rnd)$ is uniformly distributed in $\Sigma^k$, we have
\begin{align}
\begin{aligned}
    \Pr_{\rv{F} \colon \binom{V}{k}\times \Rnd \to \Sigma^k}\bigl[
        \rv{I}_{X,\rnd} = 1
    \bigr]
    & \leq \binom{k}{k-2\smol k} \left(\frac{1}{|\Sigma|}\right)^{k-2\smol k}
    \leq \left(\frac{\rme k}{2\smol k}\right)^{2\smol k} \left(\frac{1}{2}\right)^{k-2\smol k} \\
    & \leq \exp\bigl(\Theta(\smol k \log \smol^{-1})\bigr) \exp\bigl(-\Theta(k)\bigr)
    \leq \exp\bigl(-\Theta(k)\bigr).
\end{aligned}
\end{align}
Therefore, we have
\begin{align}
\begin{aligned}
    \E_{\rv{F} \colon \binom{V}{k}\times \Rnd \to \Sigma^k}
        \Bigl[\bigl|\supp_{2\smol}(\rv{F}, f^k)\bigr|\Bigr]
    = \sum_{(X,r) \in \binom{V}{k} \times \Rnd}
        \Pr_{\rv{F} \colon \binom{V}{k}\times \Rnd \to \Sigma^k}
            \bigl[\rv{I}_{X,r} = 1\bigr]
    \leq \exp\bigl(-\Theta(k)\bigr) \cdot \bigl|\tbinom{V}{k} \times \Rnd\bigr|.
\end{aligned}
\end{align}
Since $\rv{I}_{X,r}$'s are independent, the Chernoff bound derives
\begin{align}
\begin{aligned}
    & \Pr_{\rv{F} \colon \binom{V}{k}\times \Rnd \to \Sigma^k}\left[
        \frac{|\supp_{2\smol}(\rv{F}, f^k)|}{\bigl|\binom{V}{k}\times \Rnd\bigr|}
        > 2\exp\bigl(-\Theta(k)\bigr)
    \right] \\
    & = \Pr_{\rv{F} \colon \binom{V}{k}\times \Rnd \to \Sigma^k}\left[
        \sum_{(X,\rnd) \in \binom{V}{k}\times \Rnd}
        \rv{I}_{X,\rnd}
        > 2\exp\bigl(-\Theta(k)\bigr) \cdot \bigl|\tbinom{V}{k}\times \Rnd\bigr|
    \right] \\
    & \leq 2 \exp\left(-\Omega\left(
        \exp\bigl(-\Theta(k)\bigr) \cdot \bigl|\tbinom{V}{k}\times \Rnd\bigr|
    \right) \right)
    \leq \exp\left(-\Omega\bigl( \bigl|\tbinom{V}{k}\bigr| \bigr)\right).
\end{aligned}
\end{align}
Taking a union bound over all possible functions $f\colon V \to \Sigma$, we derive
\begin{align}
\begin{aligned}
    \Pr_{\rv{F} \colon \binom{V}{k}\times \Rnd \to \Sigma^k}\left[
        \exists f \colon V \to \Sigma \text{ s.t. }
        \frac{\bigl|\supp_{2\smol}(\rv{F}, f^k)\bigr|}{\bigl|\binom{V}{k}\times \Rnd\bigr|}
            > 2\exp\bigl(-\Theta(k)\bigr)
    \right]
    & \leq \bigl|\Sigma^V\bigr| \cdot
        \exp\left(-\Omega\bigl( \bigl|\tbinom{V}{k}\bigr| \bigr) \right) \\
    & \leq \exp\left(-\Omega\bigl( \bigl|\tbinom{V}{k}\bigr| \bigr) \right),
\end{aligned}
\end{align}
as desired.
\end{proof}

\begin{proof}[Proof of \cref{obs:DG08:1:iter}]
Suppose that $\radi_t > 2\smol$, which implies
\begin{align}
    \bigl|\supp_{\radi_t + \step}(F,f_t^k)\bigr|
    = \bigl|\supp_{\smol}(F, f_t^k)\bigr|
    + \sum_{0 \leq i \leq \smol^{-7}}
        \underbrace{\bigl|\supp_{\smol + (i+1)\step}(F, f_t^k) \setminus \supp_{\smol + i\step}(F, f_t^k)\bigr|}_{\geq \smol^7 \left|\binom{V}{k} \times \Rnd\right| \text{ by assumption}}
    > \bigl|\tbinom{V}{k} \times \Rnd\bigr|.
\end{align}
This is a contradiction.
Therefore, $\radi_t \leq 2\smol$.

Let $\Phi_t$ be the set of pairs $(X,\rnd) \in \binom{V}{k} \times \Rnd$
such that the value of $F_t(X;\rnd)$ has been randomized so far.
By definition,
$\Phi_1 = \emptyset$ and
$\Phi_{t+1} = \Phi_t \cup \supp_{\radi_t}(F_t, f_t^k)$.
Since $F_t(X;\rnd)$ is uniformly distributed in $\Sigma^k$ for $(X,\rnd) \in \Phi_t$,
the restriction $F_t|_{\Phi_t} \colon \Phi_t \to \Sigma^k$ is a random function.
By \cref{clm:DG08:1:iter},
with probability $1 - \exp\left(-\Omega\bigl( \bigl|\tbinom{V}{k}\bigr| \bigr) \right)$,
for every function $f \colon V \to \Sigma$,
\begin{align}
    \bigl|\supp_{2\smol}(F_t, f^k) \cap \Phi_t\bigr|
    \leq \exp\bigl(-\Theta(k)\bigr) \cdot \bigl|\tbinom{V}{k} \times \Rnd\bigr|.
\end{align}
This implies
\begin{align}
\begin{aligned}
    |\Phi_{t+1}|
    & = |\Phi_t \cup \supp_{\radi_t}(F_t, f_t^k)| \\
    & = |\Phi_t| + |\supp_{\radi_t}(F_t, f_t^k)| - |\Phi_t \cap \supp_{\radi_t}(F_t, f_t^k)| \\
    & \geq |\Phi_t| + \Omega(\smol^6)\bigl|\tbinom{V}{k} \times \Rnd\bigr|
        - \exp\bigl(-\Theta(k)\bigr)\bigl|\tbinom{V}{k} \times \Rnd\bigr| \\
    & \geq |\Phi_t| + \Omega(\smol^6)\bigl|\tbinom{V}{k} \times \Rnd\bigr|.
\end{aligned}
\end{align}
Consequently, the procedure must terminate within $\bigO(\smol^{-6})$ iterations with high probability,
as desired.
\end{proof}

By \cref{obs:DG08:1:iter}, we can assume that $T = \bigO(\smol^{-6})$
and $\radi_t \in [\smol, 2\smol]$ for every $t \in [T]$.
We introduce the following additional notations.
\begin{definition}
For each $t \in [T]$, define
\begin{align}
\begin{aligned}
    \calS_t &
        \defeq \supp_{\radi_t}(F, f_t^k).
\end{aligned}
\end{align}
For each $\tau \subseteq [T]$, define
\begin{align}
\begin{aligned}
    \Cell_\tau &
        \defeq \Biggl(\bigcap_{t \in \tau}\calS_t\Biggr) \cap \Biggl(\bigcap_{t \notin \tau}\bar{\calS_t}\Biggr), 
        \text{ where }
        \bar{\calS_t} \defeq \bigl(\tbinom{V}{k} \times \Rnd\bigr) \setminus \calS_t, \\
    \cell_\tau &
        \defeq \frac{|\Cell_\tau|}{\bigl|\binom{V}{k} \times \Rnd\bigr|}.
\end{aligned}
\end{align}
Note that $\{\Cell_\tau\}_{\tau \subseteq [T]}$ forms a partition of $\binom{V}{k} \times \Rnd$
and thus $\sum_{\tau \subseteq [T]} \cell_\tau = 1$.
\end{definition}

We will show the following lemma, which can be used to prove \cref{thm:DG08}.
\begin{lemma}
\label{lem:DG08:2}
    \begin{align}
        \sum_{\tau \neq \emptyset} \cell_\tau^q
        \geq p(1-k^{-\Omega(1)}).
    \end{align}
\end{lemma}

\begin{proof}[Proof of \cref{thm:DG08}]
By \cref{lem:DG08:2}, we have
\begin{align}
    p(1-k^{-\Omega(1)})
    \leq \sum_{\tau \neq \emptyset} \cell_\tau^q
    \leq \max_{\tau^* \neq \emptyset} \bigl\{\cell_{\tau^*}^{q-1}\bigr\}
        \cdot \sum_{\tau \neq \emptyset} \cell_\tau
    \leq \max_{\tau^* \neq \emptyset} \bigl\{\cell_{\tau^*}^{q-1}\bigr\},
\end{align}
where we used the inequality that $\sum_{\tau \neq \emptyset} \cell_\tau \leq 1$.
Therefore, there exists $\tau^* \neq \emptyset$ such that
$\cell_{\tau^*}^{q-1} \geq p(1-k^{-\Omega(1)})$.
Since $\tau^* \neq \emptyset$,
there exists $t \in \tau^*$ such that
for every pair $(X, \rnd) \in \Cell_{\tau^*}$, we have
$(X, \rnd) \in \supp_{\radi_t}(F, f_t^k)$; i.e.,
$\close{F(X;\rnd)}{f_t^k(X)}{\bigO(\smol)}$.
Consequently, we obtain
\begin{align}
    \Pr_{(\rv{X}, \rv{\rnd}) \in \binom{V}{k} \times \Rnd}\left[
        \close{F(\rv{X};\rv{\rnd})}{f_t^k(\rv{X})}{\bigO(\smol)}
    \right]
    \geq \frac{|\Cell_{\tau^*}|}{\bigl|\binom{V}{k} \times \Rnd\bigr|}
    = \cell_{\tau^*}
    \geq \bigl(p(1-k^{-\Omega(1)})\bigr)^{\frac{1}{q-1}}
    \geq p^{\frac{1}{q-1}} (1-k^{-\Omega(1)}),
\end{align}
which completes the proof.
\end{proof}

In the remainder of this subsection, we prove \cref{lem:DG08:2}.
We first show the following lemma.

\begin{lemma}
\label{lem:DG08:1:rare}
    For each $t \in [T]$ and each $i \in [q-1]$,
    \begin{align}
    \label{eq:DG08:1:rare}
        \Pr_{(\sq{\rv{I}}, \sq{\rv{X}}, \sq{\rv{\rnd}}) \sim \T_q}\Bigl[
            \T_q^F = 1 \text{ and }
            (\rv{X}_i,\rv{\rnd}_i) \in \calS_t \text{ and }
            (\rv{X}_{i+1},\rv{\rnd}_{i+1}) \notin \calS_t
        \Bigr] & = \bigO(\smol^7), \\
    \label{eq:DG08:1:rare2}
        \Pr_{(\sq{\rv{I}}, \sq{\rv{X}}, \sq{\rv{\rnd}}) \sim \T_q}\Bigl[
            \T_q^F = 1 \text{ and }
            (\rv{X}_i,\rv{\rnd}_i) \notin \calS_t \text{ and }
            (\rv{X}_{i+1},\rv{\rnd}_{i+1}) \in \calS_t
        \Bigr] & = \bigO(\smol^7).
    \end{align}
\end{lemma}
\begin{proof} 
By conditioning on whether $(\rv{X}_{i+1},\rv{\rnd}_{i+1}) \in \supp_{\radi_t + \step}(F, f_t^k)$ or not,
we can rewrite the left-hand side of \cref{eq:DG08:1:rare} as
\begin{align}
\label{eq:DG08:1:rare:divide}
\begin{aligned}
    & (\text{left-hand side of \cref{eq:DG08:1:rare}}) \\
    & = \underbrace{\Pr_{(\sq{\rv{I}}, \sq{\rv{X}}, \sq{\rv{\rnd}}) \sim \T_q}\left[
        \begin{aligned}
            \T_q^F & = 1 \\
            (\rv{X}_i,\rv{\rnd}_i) & \in \calS_t \\
            (\rv{X}_{i+1},\rv{\rnd}_{i+1}) & \notin \calS_t \\
            (\rv{X}_{i+1},\rv{\rnd}_{i+1}) & \in \supp_{\radi_t + \step}(F, f_t^k)
        \end{aligned}
    \right]}_{\text{first term}}
    + \underbrace{\Pr_{(\sq{\rv{I}}, \sq{\rv{X}}, \sq{\rv{\rnd}}) \sim \T_q}\left[
        \begin{aligned}
            \T_q^F & = 1 \\
            (\rv{X}_i,\rv{\rnd}_i) & \in \calS_t \\
            (\rv{X}_{i+1},\rv{\rnd}_{i+1}) & \notin \calS_t \\
            (\rv{X}_{i+1},\rv{\rnd}_{i+1}) & \notin \supp_{\radi_t + \step}(F, f_t^k)
        \end{aligned}
    \right]}_{\text{second term}}.
\end{aligned}
\end{align}

\noindent
By the definition of $\radi_t$, the first term of \cref{eq:DG08:1:rare:divide} can be bounded from above as
\begin{align}
\begin{aligned}
    & (\text{first term of \cref{eq:DG08:1:rare:divide}}) \\
    & = \Pr_{(\sq{\rv{I}}, \sq{\rv{X}}, \sq{\rv{\rnd}}) \sim \T_q}\Bigl[
        \T_q^F = 1 \text{ and }
        (\rv{X}_i,\rv{\rnd}_i) \in \calS_t \text{ and }
        (\rv{X}_{i+1},\rv{\rnd}_{i+1}) \in \supp_{\radi_t + \step}(F, f_t^k) \setminus \supp_{\radi_t}(F, f_t^k)
    \Bigr] \\
    & \leq \Pr_{(\sq{\rv{I}}, \sq{\rv{X}}, \sq{\rv{\rnd}}) \sim \T_q}\Bigl[
        (\rv{X}_{i+1},\rv{\rnd}_{i+1}) \in \supp_{\radi_t + \step}(F, f_t^k) \setminus \supp_{\radi_t}(F, f_t^k)
    \Bigr] \\
    & = \frac{
        \bigl|\supp_{\radi_t + \step}(F, f_t^k) \setminus \supp_{\radi_t}(F, f_t^k)\bigr|}{
        \bigl|\binom{V}{k} \times \Rnd\bigr|}
    < \smol^7.
\end{aligned}
\end{align}

\noindent
In order to bound the second term,
we use the following two claims, which will be proved afterwards.
\begin{claim}
\label{clm:DG08:1:rare:close}
    For any pair $(X,\rnd) \in \supp_{\radi_t}(F, f_t^k)$,
    \begin{align}
        \Pr_{\rv{I} \in \binom{[k]}{\ell}}\left[
            \close{F(X;\rnd)|_{\rv{I}}}{f_t^k(X)|_{\rv{I}}}{\radi_t + \step/3}
        \right]
        \geq 1 - \exp\bigl(-\Omega(\step^2 \ell)\bigr).
    \end{align}
\end{claim}
\begin{claim}
\label{clm:DG08:1:rare:far}
    For any pair $(X,r) \notin \supp_{\radi_t + \step}(F, f_t^k)$,
    \begin{align}
        \Pr_{\rv{I} \in \binom{[k]}{\ell}}\left[
            \far{F(X;\rnd)|_{\rv{I}}}{f_t^k(X)|_{\rv{I}}}{\radi_t + 2\step/3}
        \right]
        \geq 1 - \exp\bigl(-\Omega(\step^2 \ell)\bigr).
    \end{align}
\end{claim}

\noindent
Conditioned on the two events that
\begin{align}
    & \qquad \close{F(\rv{X}_i;\rv{\rnd}_i)|_{\rv{I}_i}}{f_t^k(\rv{X}_i)|_{\rv{I}_i}}{\radi_t + \step/3}, 
    \label{eq:DG08:1:rare:close} \\
    & \far{F(\rv{X}_{i+1};\rv{\rnd}_{i+1})|_{\rv{I}_i}}{f_t^k(\rv{X}_{i+1})|_{\rv{I}_i}}{\radi_t + 2\step/3},
    \label{eq:DG08:1:rare:far}
\end{align}
we must have
$\far{F(\rv{X}_i;\rv{\rnd}_i)|_{\rv{I}_i}}{F(\rv{X}_{i+1};\rv{\rnd}_{i+1})|_{\rv{I}_i}}{\step/3}$.
Since 
\begin{align}
    \frac{\step}{3}
    \geq \frac{\smol^8}{3}
    \underbrace{=}_{\text{\cref{eq:DG08:smol}}}
        \Omega\left(\frac{m}{\ell}\log^8 k\right)
    = \omega\left(\frac{m}{\ell}\right),
\end{align}
we also have
$\far{F(\rv{X}_i;\rv{\rnd}_i)|_{\rv{I}_i}}{F(\rv{X}_{i+1};\rv{\rnd}_{i+1})|_{\rv{I}_i}}{m/\ell}$; i.e.,
$\T_q$ would reject.
In other words, $\T_q$ would accept only if \cref{eq:DG08:1:rare:close} or \cref{eq:DG08:1:rare:far} does not hold.
By \cref{clm:DG08:1:rare:close,clm:DG08:1:rare:far} and the union bound, 
the second term of \cref{eq:DG08:1:rare:divide} can be bounded from above as
\begin{align}
\begin{aligned}
    & (\text{second term of \cref{eq:DG08:1:rare:divide}}) \\
    & = \Pr_{(\sq{\rv{I}}, \sq{\rv{X}}, \sq{\rv{\rnd}}) \sim \T_q}\Bigl[
        \T_q^F = 1 \text{ and }
        (\rv{X}_i,\rv{\rnd}_i) \in \supp_{\radi_t}(F, f_t^k) \text{ and }
        (\rv{X}_{i+1},\rv{\rnd}_{i+1}) \notin \supp_{\radi_t + \step}(F, f_t^k)
    \Bigr] \\
    & \leq \Pr_{(\sq{\rv{I}}, \sq{\rv{X}}, \sq{\rv{\rnd}}) \sim \T_q}\left[
        \begin{aligned}
        & \text{\cref{eq:DG08:1:rare:close} or \cref{eq:DG08:1:rare:far} does not hold} \\
        & (\rv{X}_i,\rv{\rnd}_i) \in \supp_{\radi_t}(F, f_t^k) \\
        & (\rv{X}_{i+1},\rv{\rnd}_{i+1}) \notin \supp_{\radi_t + \step}(F, f_t^k)
        \end{aligned}
    \right] \\
    & \leq \Pr_{(\sq{\rv{I}}, \sq{\rv{X}}, \sq{\rv{\rnd}}) \sim \T_q}\left[
        \begin{aligned}
        & \text{\cref{eq:DG08:1:rare:close} does not hold} \\
        & (\rv{X}_i,\rv{\rnd}_i) \in \supp_{\radi_t}(F, f_t^k) \\
        & (\rv{X}_{i+1},\rv{\rnd}_{i+1}) \notin \supp_{\radi_t + \step}(F, f_t^k)
        \end{aligned}
    \right]
    + \Pr_{(\sq{\rv{I}}, \sq{\rv{X}}, \sq{\rv{\rnd}}) \sim \T_q}\left[
        \begin{aligned}
        & \text{\cref{eq:DG08:1:rare:far} does not hold} \\
        & (\rv{X}_i,\rv{\rnd}_i) \in \supp_{\radi_t}(F, f_t^k) \\
        & (\rv{X}_{i+1},\rv{\rnd}_{i+1}) \notin \supp_{\radi_t + \step}(F, f_t^k)
        \end{aligned}
    \right] \\
    & \leq \Pr_{(\sq{\rv{I}}, \sq{\rv{X}}, \sq{\rv{\rnd}}) \sim \T_q}\Bigl[
        \text{\cref{eq:DG08:1:rare:close} does not hold}
        \Bigm|
        (\rv{X}_i,\rv{\rnd}_i) \in \supp_{\radi_t}(F, f_t^k)
    \Bigr] \\
    & \qquad + \Pr_{(\sq{\rv{I}}, \sq{\rv{X}}, \sq{\rv{\rnd}}) \sim \T_q}\Bigl[
        \text{\cref{eq:DG08:1:rare:far} does not hold}
        \Bigm|
        (\rv{X}_{i+1},\rv{\rnd}_{i+1}) \notin \supp_{\radi_t + \step}(F, f_t^k)
    \Bigr] \\
    & \leq 2 \exp\bigl(-\Omega(\step^2 \ell)\bigr).
\end{aligned}
\end{align}

Since
\begin{align}
    \exp\bigl(-\Omega(\step^2 \ell)\bigr)
    = \exp\bigl(-\Omega(\smol^{16} \ell)\bigr)
    \underbrace{=}_{\text{\cref{eq:DG08:smol}}}
    \exp\left(-\Omega\left(\tfrac{\log^{16} k}{\ell} \ell\right)\right)
    = \exp\bigl(-\Omega(\log^{16} k)\bigr)
    = \bigO(\smol^7),
\end{align}
we have
\begin{align}
    (\text{left-hand side of \cref{eq:DG08:1:rare}})
    \leq \smol^7 + 2\exp\bigl(-\Omega(\step^2 \ell)\bigr)
    \leq \bigO(\smol^7),
\end{align}
as desired.
\cref{eq:DG08:1:rare2} can be shown similarly.
\end{proof}

What remains to be done is to prove \cref{clm:DG08:1:rare:close,clm:DG08:1:rare:far}.

\begin{proof}[Proof of \cref{clm:DG08:1:rare:close}]
For any pair $(X,\rnd) \in \supp_{\radi_t}(F, f_t^k)$, let 
\begin{align}
    \rv{Z}
    \defeq \ell \cdot \dist\bigl(F(X;\rnd)|_{\rv{I}}, f_t^k(X)|_{\rv{I}}\bigr)
    = \sum_{i \in \rv{I}} \bigl\llbracket
        F(X;\rnd)|_i \neq f_t^k(X)|_i
    \bigr\rrbracket
\end{align}
for a random set $\rv{I} \in \binom{[k]}{\ell}$.
Observe that $\rv{Z}$ follows a hypergeometric distribution where
the population size is $k$,
the number of successes is $k \cdot \dist\bigl(F(X;\rnd), f_t^k(X)\bigr)$, and
the sample size is $\ell$.
By assumption, $\E[\rv{Z}] \leq \radi_t \ell$.
By Hoeffding's inequality, we have
\begin{align}
\begin{aligned}
    \Pr_{\rv{I} \in \binom{[k]}{\ell}}\left[
        \far{F(X;\rnd)|_{\rv{I}}}{f_t^k(X)|_{\rv{I}}}{\radi_t + \step/3}
    \right]
    \leq \Pr\Bigl[
        \rv{Z} \geq \E[\rv{Z}] + \tfrac{\step}{3} \ell
    \Bigr]
    \leq 2\exp\left(-2 \left(\tfrac{\step}{3}\right)^2 \ell\right)
    = \exp\bigl(-\Omega(\step^2 \ell)\bigr),
\end{aligned}
\end{align}
as desired.
\end{proof}

\begin{proof}[Proof of \cref{clm:DG08:1:rare:far}]
For any pair $(X,\rnd) \notin \supp_{\radi_t + \step}(F, f_t^k)$, let 
\begin{align}
    \rv{Z}
    \defeq \ell \cdot \dist\bigl(F(X;\rnd)|_{\rv{I}}, f_t^k(X)|_{\rv{I}}\bigr)
    = \sum_{\rv{i} \in \rv{I}} \bigl\llbracket
        F(X;\rnd)|_i \neq f_t^k(X)|_i
    \bigr\rrbracket
\end{align}
for a random set $\rv{I} \in \binom{[k]}{\ell}$.
Observe that $\rv{Z}$ follows a hypergeometric distribution where
the population size is $k$,
the number of successes is $k \cdot \dist\bigl(F(X;\rnd), f_t^k(X)\bigr)$, and
the sample size is $\ell$.
By assumption, $\E[\rv{Z}] > (\radi_t + \step) \ell$.
By Hoeffding's inequality, we have
\begin{align}
\begin{aligned}
    \Pr_{\rv{I} \in \binom{[k]}{\ell}}\left[
        \close{F(X;\rnd)|_{\rv{I}}}{f_t^k(X)|_{\rv{I}}}{\radi_t + 2\step/3}
    \right]
    \leq \Pr\Bigl[
        \rv{Z} \leq \E[\rv{Z}] - \tfrac{\step}{3} \ell
    \Bigr] 
    \leq 2\exp\left(-2 \left(\tfrac{\step}{3}\right)^2 \ell\right) 
    = \exp\bigl(-\Omega(\step^2 \ell)\bigr),
\end{aligned}
\end{align}
as desired.
\end{proof}

By applying \cref{lem:DG08:1:rare},
we show that conditioned on the event that $\T_q$ accepts $F$,
all $q$ queries $(\rv{X}_1, \rv{\rnd}_1), \ldots, (\rv{X}_q, \rv{\rnd}_q)$
belong to the same nonempty $\Cell_\tau$ with high probability.

\begin{lemma}
\label{lem:DG08:1}
    \begin{align}
        \Pr\Bigl[\T_q^F = 1\Bigr]
        \leq \bigO_q(\smol)
        + \sum_{\tau \neq \emptyset}
        \Pr_{(\sq{\rv{I}}, \sq{\rv{X}}, \sq{\rv{\rnd}}) \sim \T_q}\Biggl[
            \T_q^F = 1 \text{ and }
            \bigwedge_{1 \leq i \leq q} (\rv{X}_i, \rv{\rnd}_i) \in \Cell_\tau
        \Biggr],
    \end{align}
\end{lemma}
\begin{proof} 
For each pair $(X,\rnd) \in \binom{V}{k} \times \Rnd$,
let $\tau(X,\rnd)$ denote the unique $\tau \subseteq [T]$ such that $(X,\rnd) \in \Cell_\tau$.
The event that $\T_q$ accepts $F$ can be partitioned into the following three cases:
\begin{description}
    \item[(C1)] $\tau(\rv{X}_i,\rv{\rnd}_i) = \emptyset$ for every $i \in [q]$:
        Observing that $F_t(X;\rnd) = F(X;\rnd)$ for every pair $(X,\rnd) \in \Cell_{\emptyset}$,
        we have
        \begin{align}
        \begin{aligned}
            \Pr_{(\sq{\rv{I}}, \sq{\rv{X}}, \sq{\rv{\rnd}}) \sim \T_q}\Bigl[
                \T_q^F = 1
                \text{ and }
                \text{(C1) holds}
            \Bigr]
            & = \Pr_{(\sq{\rv{I}}, \sq{\rv{X}}, \sq{\rv{\rnd}}) \sim \T_q}\Bigl[
                \T_q^{F_t} = 1
                \text{ and }
                \text{(C1) holds}
            \Bigr] \\
            & \leq \Pr\Bigl[\T_q^{F_t} = 1\Bigr]
            < \smol,
        \end{aligned}
        \end{align}
        where the last inequality holds due to the termination condition of the procedure.
    \item[(C2)]
        $\tau(\rv{X}_{j_1},\rv{\rnd}_{j_1}) \neq \tau(\rv{X}_{j_2},\rv{\rnd}_{j_2})$
        for some $j_1 \neq j_2$:
        There exists $i \in [q-1]$ such that
        $\tau(\rv{X}_i, \rv{\rnd}_i) \neq \tau(\rv{X}_{i+1}, \rv{\rnd}_{i+1})$.
        By the definition of $\Cell_\tau$,
        there exists some
        $t \in \tau(\rv{X}_i, \rv{\rnd}_i) \triangle \tau(\rv{X}_{i+1}, \rv{\rnd}_{i+1}) \neq \emptyset$\footnote{
            $A \triangle B$ denotes the symmetric difference of $A$ and $B$.
        }
        such that
        \begin{itemize}
            \item $(\rv{X}_i,\rv{\rnd}_i) \in \supp_{\radi_t}(F, f_t^k)$ and $(\rv{X}_{i+1},\rv{\rnd}_{i+1}) \notin \supp_{\radi_t}(F, f_t^k)$, or
            \item $(\rv{X}_i,\rv{\rnd}_i) \notin \supp_{\radi_t}(F, f_t^k)$ and $(\rv{X}_{i+1},\rv{\rnd}_{i+1}) \in \supp_{\radi_t}(F, f_t^k)$.
        \end{itemize}
        By applying \cref{lem:DG08:1:rare} and the union bound, we have
        \begin{align}
        \begin{aligned}
            & \Pr_{(\sq{\rv{I}}, \sq{\rv{X}}, \sq{\rv{\rnd}}) \sim \T_q}\Bigl[
                \T_q^F = 1 \text{ and }
                \text{(C2) holds}
            \Bigr] \\
            & \leq \sum_{t \in [T]} \sum_{i \in [q-1]}
            \Pr_{(\sq{\rv{I}}, \sq{\rv{X}}, \sq{\rv{\rnd}}) \sim \T_q}\left[
                \begin{aligned}
                    \T_q^F & = 1 \\
                    (\rv{X}_i,\rv{\rnd}_i) & \in \supp_{\radi_t}(F, f_t^k) \\
                    (\rv{X}_{i+1},\rv{\rnd}_{i+1}) & \notin \supp_{\radi_t}(F, f_t^k)
                \end{aligned}
            \right]
            + \Pr_{(\sq{\rv{I}}, \sq{\rv{X}}, \sq{\rv{\rnd}}) \sim \T_q}\left[
                \begin{aligned}
                    \T_q^F & = 1 \\
                    (\rv{X}_i,\rv{\rnd}_i) & \notin \supp_{\radi_t}(F, f_t^k) \\
                    (\rv{X}_{i+1},\rv{\rnd}_{i+1}) & \in \supp_{\radi_t}(F, f_t^k)
                \end{aligned}
            \right] \\
            & \leq T \cdot (q-1) \cdot 2 \bigO(\smol^7)
            \underbrace{\leq}_{T = \bigO(\smol^{-6})}
                \bigO_q(\smol).
        \end{aligned}
        \end{align}
    \item[(C3)]
        $\tau(\rv{X}_i,\rv{\rnd}_i) = \tau$ for every $i \in [q]$ and for some $\tau \neq \emptyset$:
        Since $\Cell_\tau$'s are pairwise disjoint, we have
        \begin{align}
            \Pr_{(\sq{\rv{I}}, \sq{\rv{X}}, \sq{\rv{\rnd}}) \sim \T_q}\Bigl[
                \T_q^F = 1 \text{ and }
                \text{(C3) holds}
            \Bigr]
            = \sum_{\tau \neq \emptyset}
            \Pr_{(\sq{\rv{I}}, \sq{\rv{X}}, \sq{\rv{\rnd}}) \sim \T_q}\Biggl[
                \T_q^F = 1\text{ and }
                \bigwedge_{1 \leq i \leq q}
                (\rv{X}_i, \rv{\rnd}_i) \in \Cell_\tau
            \Biggr].
        \end{align}
\end{description}
Consequently, 
\begin{align}
\begin{aligned}
    \Pr_{(\sq{\rv{I}}, \sq{\rv{X}}, \sq{\rv{\rnd}}) \sim \T_q}\Bigl[\T_q^F = 1\Bigr]
    \leq \bigO_q(\smol)
        + \sum_{\tau \neq \emptyset}
        \Pr_{(\sq{\rv{I}}, \sq{\rv{X}}, \sq{\rv{\rnd}}) \sim \T_q}\Biggl[
            \T_q^F = 1 \text{ and }
            \bigwedge_{1 \leq i \leq q} (\rv{X}_i, \rv{\rnd}_i) \in \Cell_\tau
        \Biggr],
\end{aligned}
\end{align}
as desired.
\end{proof}

We are now ready to prove \cref{lem:DG08:2} by applying \cref{lem:Johnson,lem:DG08:1}.

\begin{proof}[Proof of \cref{lem:DG08:2}]
By applying \cref{lem:Johnson} to each $\Cell_\tau$, we have
\begin{align}
\begin{aligned}
    \Pr_{(\sq{\rv{I}}, \sq{\rv{X}}, \sq{\rv{\rnd}}) \sim \T_q}\Biggl[
        \bigwedge_{1 \leq i \leq q} (\rv{X}_i, \rv{\rnd}_i) \in \Cell_\tau
    \Biggr]
    & \leq \left(\frac{|\Cell_\tau|}{\bigl|\binom{V}{k} \times \Rnd\bigr|}\right)^q
        +  \bigO_q\left(\frac{\ell}{k}\right)
            \left(\frac{|\Cell_\tau|}{\bigl|\binom{V}{k} \times \Rnd\bigr|}\right)
        + \bigO_{q,k}\left(\frac{1}{n}\right) \\
    & = \cell_\tau^q
        + \bigO_q\left(\frac{\ell}{k}\right) \cell_\tau
        + \bigO_{q,k}\left(\frac{1}{n}\right).
\end{aligned}
\end{align}
Therefore, by \cref{lem:DG08:1}, we have
\begin{align}
\begin{aligned}
    \Pr\Bigl[\T_q^F = 1\Bigr]
    & \leq \bigO_q(\smol)
    + \sum_{\tau \neq \emptyset} \Pr_{(\sq{\rv{I}}, \sq{\rv{X}}, \sq{\rv{\rnd}}) \sim \T_q}\Biggl[
        \T_q^F = 1 \text{ and }
        \bigwedge_{1 \leq i \leq q} (\rv{X}_i,\rv{\rnd}_i) \in \Cell_\tau
    \Biggr] \\
    & \leq \bigO_q(\smol)
    + \sum_{\tau \neq \emptyset} \Pr_{(\sq{\rv{I}}, \sq{\rv{X}}, \sq{\rv{\rnd}}) \sim \T_q}\Biggl[
        \bigwedge_{1 \leq i \leq q} (\rv{X}_i,\rv{\rnd}_i) \in \Cell_\tau
    \Biggr] \\
    & \leq \bigO_q(\smol)
        + \sum_{\tau \neq \emptyset} \left[
            \cell_\tau^q
            + \bigO_q\left(\frac{\ell}{k}\right) \cell_\tau
            + \bigO_{q,k}\left(\frac{1}{n}\right)
        \right] \\
    & = \bigO_q(\smol)
        + \sum_{\tau \neq \emptyset}
            \cell_\tau^q
            + \bigO_q\left(\frac{\ell}{k}\right) \sum_{\tau \neq \emptyset} \cell_\tau
            + \bigO_{q,k}\left(\frac{1}{n}\right) 2^T \\
    & \leq \bigO_q(\smol)
        + \bigO_q\left(\frac{\ell}{k}\right)
        + \bigO_{q,k}\left(\frac{1}{n}\right)
        + \sum_{\tau \neq \emptyset} \cell_\tau^q,
\end{aligned}
\end{align}
where the last inequality holds because
\begin{align}
    \sum_{\tau \neq \emptyset} \cell_\tau
    \leq 1 \text{ and }
    \bigO_{q,k}\left(\frac{1}{n}\right) 2^T
    \underbrace{=}_{T = \poly(k)} \bigO_{q,k}\left(\frac{1}{n}\right).
\end{align}
Since
$\Pr[\T_q^F = 1] = p \geq \sqrt{\smol}$ and
$\smol = \Omega\left(\frac{\ell}{k}\right)$ as defined by \cref{eq:DG08:smol},
we derive that for sufficiently large $k$ and $n$,
\begin{align}
\begin{aligned}
    \sum_{\tau \neq \emptyset} \cell_\tau^q
    & \geq \Pr\Bigl[\T_q^F = 1\Bigr]
        - \bigO_q(\smol)
        - \bigO_q\left(\frac{\ell}{k}\right)
        - \bigO_{q,k}\left(\frac{1}{n}\right) \\
    & \geq p\left(1
        - \bigO_q\left(\frac{\smol}{p}\right)
        - \bigO_q\left(\frac{\ell}{k p}\right)
        - \bigO_{q,k}\left(\frac{1}{pn}\right)
    \right) \\
    & \geq p\left(1
        -\bigO_q(\sqrt{\smol})
        - \bigO_q\left(\sqrt{\frac{\ell}{k}}\right)
        - \bigO_{q,k}\left(\frac{1}{pn}\right)
    \right)
    \\
    & \geq p(1-k^{-\Omega(1)}),
\end{aligned}
\end{align}
which completes the proof.
\end{proof}

\section{\texorpdfstring{%
$\PSPACE$-hardness of $\left(\frac{1}{2^{q-1}}+\epsilon\right)$-factor Approximation for \MMqCSPReconf
}{%
PSPACE-hardness of (1/2\textasciicircum(q-1)+ε)-factor Approximation for Maxmin q-CSP Reconfiguration
}}
\label{sec:PSPACE}

In this section, we prove \cref{thm:intro:PSPACE}; i.e.,
\MMqCSPReconf is $\PSPACE$-hard to approximate within a factor of $\frac{1}{2^{q-1}}+\epsilon$ for any small real $\epsilon > 0$.

\begin{theorem}
\label{thm:PSPACE}
For
any integer $q \geq 2$ and
any real $\epsilon > 0$,
there exists a positive integer $\sigma$
such that
\prb{Gap$_{1,\frac{1}{2^{q-1}}+\epsilon}$ $q$-CSP$_\sigma$ Reconfiguration}
is $\PSPACE$-hard.
Moreover, the same hardness result holds even if the underlying hypergraph is regular.
\end{theorem}

The proof of \cref{thm:PSPACE} is based on
the following gap-preserving reduction from \prb{Maxmin 2-CSP Reconfiguration} to \prb{Maxmin $q$-CSP Reconfiguration}.

\begin{lemma}
\label{lem:PSPACE}
    For
    any reals $s \in (0,1)$ and $\epsilon > 0$, and
    any positive integers $q \geq 2$, $\sigma$, and $\Delta$,
    there exists a positive integer $k$ such that
    there exists a polynomial-time reduction from
    \prb{Gap$_{1,s}$ 2-CSP$_\sigma$ Reconfiguration}
    whose underlying graph is $\Delta$-regular to
    \prb{Gap$_{1,\frac{1}{2^{q-1}}+\epsilon}$ $q$-CSP$_{\sigma^{2k}}$ Reconfiguration}
    whose underlying hypergraph is regular.
\end{lemma}

\begin{proof}[Proof of \cref{thm:PSPACE}]
    By the PCRP theorem \cite{hirahara2024probabilistically,guruswami2025inapproximability} and
    gap-preserving reductions \cite{ohsaka2023gap},
    \prb{Gap$_{1,s}$ 2-CSP$_3$ Reconfiguration} on $\Delta$-regular 2-CSP instances is $\PSPACE$-hard
    for some real $s \in (0,1)$ and some positive integer $\Delta$.
    By \cref{lem:PSPACE},
    for any integer $q \geq 2$ and any real $\epsilon > 0$,
    there exists a positive integer $k$ such that
    \prb{Gap$_{1,\frac{1}{2^{q-1}}+\epsilon}$ $q$-CSP$_{3^{2k}}$ Reconfiguration} on regular $q$-CSP instances
    is $\PSPACE$-hard, as desired.
\end{proof}

The remainder of this section is devoted to the proof of \cref{lem:PSPACE}.

\paragraph{Additional Notations.}
Some additional notations are introduced.
Let $G=(V,E)$ be a graph.
For a vertex $x$ of $G$, let $E(x)$ denote the set of edges incident to $x$; namely,
$E(x) \defeq \{xy \in E\}$.
For a vertex $k$-tuple $\sq{x} = (x_1, \ldots, x_k)$ of $G$,
let $E^k(\sq{x})$ denote the set of edge $k$-tuples
obtained by selecting $k$ edges incident to $x_1, \ldots, x_k$; namely,
\begin{align}
    E^k(\sq{x})
    \defeq \prod_{1 \leq i \leq k} E(x_i)
    = \bigl\{ (e_1, \ldots, e_k) \in E^k \bigm| \forall i \in [k], \; e_i \in E(x_i) \bigr\}.
\end{align}
For a vertex function $f \colon V \to \Sigma$,
we write $f_E \colon E \to \Sigma^2$ for the \defi{induced edge function},
which is defined as
\begin{align}
    f_E(xy) \defeq \bigl(f(x), f(y)\bigr)
    \text{ for each edge } xy \in E.
\end{align}

\noindent
Let $F \colon E^k \to (\Sigma^2)^k$ be an edge $k$-tuple function; namely,
for an edge $k$-tuple $\sq{e}$,
a pair of symbols in $\Sigma^2$ are assigned to the endpoints of each edge in $\sq{e}$.
Suppose that
$F(\sq{e}) = ((\alpha_1, \beta_1), \ldots, (\alpha_k, \beta_k))$
for an edge $k$-tuple $\sq{e} = (e_1, \ldots, e_k) = (x_1y_1, \ldots, x_k y_k) \in E^k$.
Then, we use $F(\sq{e})[e_i] \in \Sigma^2$ to denote
a pair of symbols assigned to the endpoints of $e_i$; namely,
$F(\sq{e})[e_i] = (\alpha_i, \beta_i)$.
By abuse of notation,
for a vertex $k$-tuple $\sq{x} = (x_1, \ldots, x_k) \in V^k$,
we use $F(\sq{e})[x_i] \in \Sigma$ to denote
a symbol assigned to the endpoint $x_i$ of $e_i$; namely,
$F(\sq{e})[x_i] \defeq \alpha_i$, and
use $F(\sq{e})[\sq{x}] \in \Sigma^k$ to denote
a $k$-tuple of symbols assigned to the endpoints $x_1, \ldots, x_k$ of $e_1, \ldots, e_k$; namely,
$F(\sq{e})[\sq{x}] \defeq (\alpha_1, \ldots, \alpha_k)$.
For a $k$-tuple $\sq{x} = (x_1, \ldots, x_k)$ and a permutation $\pi \in \sym_k$,
let $\sq{x} \circ \pi$ denote a $k$-tuple such that
$(\sq{x} \circ \pi)_i \defeq x_{\pi(i)}$ for each $i \in [k]$; namely,
\begin{align}
    \sq{x} \circ \pi
    \defeq (x_{\pi(1)}, \ldots, x_{\pi(k)}).
\end{align}

\paragraph{Tolerant $q$-query Direct Product Tester for Edge $k$-tuple Functions.}

Before describing the gap-preserving reduction,
we introduce the tolerant $q$-query direct product tester for an edge $k$-tuple function,
denoted by $\W_q$.
Let
$G = (V,E)$ be a $\Delta$-regular graph,
$F \colon E^k \to (\Sigma^2)^k$ be an edge $k$-tuple function, and
$\ell \defeq \Theta(\sqrt{k})$.
Intuitively, $F$ is supposed to be
the $k$-wise direct product of the induced edge function of some assignment $f \colon V \to \Sigma$; namely,
$f_E^k \colon E^k \to (\Sigma^2)^k$.
Then, $\W_q$ tests if this is the case.

\begin{itembox}[l]{\textbf{Tolerant $q$-query direct product tester $\W_q$ for an edge $k$-tuple function}}
\begin{algorithmic}[1]
    \item[\textbf{Input:}]
        a $\Delta$-regular graph $G = (V,E)$.
    \item[\textbf{Oracle access:}]
        an edge $k$-tuple function $F \colon E^k \to (\Sigma^2)^k$.
    \State let $\ell \defeq \Theta(\sqrt{k})$.
    \State sample $\sq{\rv{x}}_1$ from $V^k$ uniformly.
    \For{\textbf{each} $i \in [q-1]$}
        \State sample $\rv{I}_i$ from $\binom{[k]}{\ell}$ uniformly.
        \State sample $\sq{\rv{x}}_{i+1}$ from $V^k$ uniformly
            conditioned on $\sq{\rv{x}}_{i+1}|_{\rv{I}_i} = \sq{\rv{x}}_i|_{\rv{I}_i}$.
    \EndFor
    \For{\textbf{each} $i \in [q]$}
        \State sample $\sq{\rv{e}}_i$ from $E^k(\sq{\rv{x}}_i)$ uniformly.
        \State sample $\rv{\pi}_i$ from $\sym_k$ uniformly.
        \State read $F(\sq{\rv{e}}_i \circ \rv{\pi}_i)$.
    \EndFor
    \For{\textbf{each} $i \in [q-1]$} \label{linum:W:start}
        \If{$|\{ \rv{x}_{i,1}, \ldots, \rv{x}_{i,k}, \rv{x}_{i+1,1}, \ldots, \rv{x}_{i+1,k} \}| < 2k - \ell$}
            \Comment{duplicates found between $\sq{\rv{x}}_i$ and $\sq{\rv{x}}_{i+1}$.}
            \State \textbf{continue}.
        \EndIf
        \If{$\far{
        \bigl(F(\sq{\rv{e}}_i \circ \rv{\pi}_i) \circ \rv{\pi}_i^{-1}\bigr)\bigl[\sq{\rv{x}}_i|_{\rv{I}_i}\bigr]}{
        \bigl(F(\sq{\rv{e}}_{i+1} \circ \rv{\pi}_{i+1}) \circ \rv{\pi}_{i+1}^{-1}\bigr)\bigl[\sq{\rv{x}}_{i+1}|_{\rv{I}_i}\bigr]}{
        1/\ell}$}
            \State \Return $0$.
        \EndIf
    \EndFor \label{linum:W:end}
    \State \Return $1$.
\end{algorithmic}
\end{itembox}

\noindent
Our verifier $\W_q$ has the following completeness and soundness,
whose proof is obtained by applying \cref{thm:TDP} and is deferred to \cref{sec:PSPACE:W}.

\begin{lemma}
\label{lem:PSPACE:W}
Let $\epsilon > 0$ be a real,
$G=(V,E)$ be a $\Delta$-regular $n$-vertex graph, and
$F \colon E^k \to (\Sigma^2)^k$ be an edge $k$-tuple function.
If $k$ and $n$ are sufficiently large, the following hold:
\begin{description}
    \item[(Completeness)]
        Suppose that $F \in \mix(f_E^k, g_E^k)$
        for two functions $f, g \colon V \to \Sigma$ that differ in at most a single vertex.
        Then, $\W_q$ accepts $F$ with probability $1$.
    \item[(Soundness)]
        Suppose that $\W_q$ accepts $F$ with probability at least $\frac{1}{2^{q-1}}+\epsilon$.
        Then, there exists a function $f \colon V \to \Sigma$ such that
        \begin{align}
        \label{eq:PSPACE:W:soundness}
            \Pr_{\sq{\rv{e}} \in E^k}\left[
                \close{
                    F(\sq{\rv{e}})}{
                    f_E^k(\sq{\rv{e}})}{
                    k^{-\e}}
            \right] \geq \frac{1}{2}+\frac{\epsilon}{4},
        \end{align}
        where $\e \in (0,1)$ is a universal constant.
\end{description}
\end{lemma}

\subsection{\texorpdfstring{%
Proof of \cref{lem:PSPACE}
}{%
Proof of Lemma \ref{lem:PSPACE}
}}

\paragraph{Reduction.}

Our gap-preserving reduction from
\MMtwoCSPReconf to \MMqCSPReconf is described as follows.
Let
$s \in (0,1)$ and $\epsilon > 0$ be reals, and
$q \geq 2$, $\sigma$ and $\Delta$ be positive integers.
Let $(G,f_\sss,f_\ttt)$ be an instance of
\prb{Gap$_{1,s}$ 2-CSP$_\sigma$ Reconfiguration}, where
$G = (V,E,\Sigma,\Psi = (\psi_e)_{e \in E})$
is a satisfiable $\Delta$-regular $n$-vertex 2-CSP instance with $|\Sigma| = \sigma$, and
$f_\sss, f_\ttt \colon V \to \Sigma$ are a pair of its satisfying assignments.
We construct a new instance of \MMqCSPReconf
in the form of $q$-query verifier.
Let $k$ be a positive integer.
Without loss of generality, we can assume that
$k$ and $n$ are sufficiently large so that \cref{lem:PSPACE:W} can be applied.

Our $q$-query verifier $\V_q$ for \prb{Maxmin 2-CSP Reconfiguration} is described below.
Given oracle access to an edge $k$-tuple function $F \colon E^k \to (\Sigma^2)^k$,
$\V_q$ first tests if $F$ is close to $f_E^k$ for some assignment $f \colon V \to \Sigma$
by running $\W_q$ and then tests if the $qk$ edges selected by $\W_q$ are satisfied by $F$.

\begin{itembox}[l]{\textbf{$q$-query verifier $\V_q$ for \prb{Maxmin 2-CSP Reconfiguration}}}
\begin{algorithmic}[1]
    \item[\textbf{Input:}]
        a regular 2-CSP instance $G = (V,E,\Sigma,\Psi = (\psi_e)_{e \in E})$.
    \item[\textbf{Oracle access:}]
        an edge $k$-tuple function $F \colon E^k \to (\Sigma^2)^k$.
    \State run $\W_q(V,E)$ on $F$.
    \If{run of $\W_q(V,E)$ returned $0$}
        \State \Return $0$.
    \EndIf
    \For{\textbf{each} $i \in [q]$}
        \State let $\sq{\rv{e}}_i \defeq (\rv{e}_{i,1}, \ldots, \rv{e}_{i,k})$
            be the \nth{$i$} edge $k$-tuple selected by $\W_q(V,E)$.
        \For{\textbf{each} $j \in [k]$}
            \If{$\psi_{\rv{e}_{i,j}}\bigl(F(\sq{\rv{e}}_i)[\rv{e}_{i,j}]\bigr) = 0$}
                \State \Return $0$.
            \EndIf
        \EndFor
    \EndFor
    \State \Return $1$.
\end{algorithmic}
\end{itembox}

\noindent
The starting and ending functions
$F_\sss, F_\ttt \colon E^k \to (\Sigma^2)^k$ are defined as
$F_\sss \defeq (f_\sss)_E^k$ and
$F_\ttt \defeq (f_\ttt)_E^k$; namely,
for each edge $k$-tuple $(x_1y_1, \ldots, x_k y_k) \in E^k$,
\begin{align}
\begin{aligned}
    F_\sss(x_1y_1, \ldots, x_k y_k)
    & \defeq \Bigl(\bigl(f_\sss(x_1), f_\sss(y_1)\bigr), \ldots, \bigl(f_\sss(x_k), f_\sss(y_k)\bigr)\Bigr), \\
    F_\ttt(x_1y_1, \ldots, x_k y_k)
    & \defeq \Bigl(\bigl(f_\ttt(x_1), f_\ttt(y_1)\bigr), \ldots, \bigl(f_\ttt(x_k), f_\ttt(y_k)\bigr)\Bigr).
\end{aligned}
\end{align}
Observe that $\V_q$ accepts both $F_\sss$ and $F_\ttt$ with probability $1$ by \cref{lem:PSPACE:W},
which completes the description of the reduction.

For a function $F \colon E^k \to (\Sigma^2)^k$,
the \defi{value} $\val_{\V_q}(F)$ is defined as the probability that $\V_q$ accepts $F$.
For a reconfiguration sequence $\sq{F} = (F^{(1)}, \ldots, F^{(T)})$,
the \defi{value} $\val_{\V_q}(\sq{F})$
is defined as the minimum value over all functions in $\sq{F}$; namely,
\begin{align}
    \val_{\V_q}(\sq{F})
    \defeq \min_{1 \leq t \leq T} \bigl\{ \val_{\V_q}(F^{(t)}) \bigr\}.
\end{align}
The \defi{optimal value} $\opt_{\V_q}(F_\sss \reco F_\ttt)$
is defined as the maximum value over all possible reconfiguration sequences $\sq{F}$ from $F_\sss$ to $F_\ttt$; namely,
\begin{align}
    \opt_{\V_q}(F_\sss \reco F_\ttt)
    \defeq \max_{\sq{F} = (F_\sss, \ldots, F_\ttt)} \bigl\{\val_{\V_q}(\sq{F})\bigr\}.
\end{align}

\paragraph{Correctness.}
We first show the completeness.

\begin{lemma}
\label{lem:PSPACE:completeness}
If $\opt_G(f_\sss \reco f_\ttt) = 1$,
then $\opt_{\V_q}(F_\sss \reco F_\ttt) = 1$.
\end{lemma}
\begin{proof} 
Suppose first that $f_\sss$ and $f_\ttt$ differ in a single vertex.
Consider a trivial reconfiguration sequence $\sq{F}$ from $F_\sss$ to $F_\ttt$ obtained by the following procedure.

\begin{itembox}[l]{\textbf{Reconfiguration sequence $\sq{F}$ from $F_\sss$ to $F_\ttt$}}
\begin{algorithmic}[1]
    \LComment{start with $F_\sss$.}
    \For{\textbf{each} edge $k$-tuple $\sq{e} \in E^k$ such that $F_\sss(\sq{e}) \neq F_\ttt(\sq{e})$}
        \State change the current assignment to $\sq{e}$ from $F_\sss(\sq{e})$ to $F_\ttt(\sq{e})$.
    \EndFor
    \LComment{end with $F_\ttt$.}
\end{algorithmic}
\end{itembox}

\noindent
We show that any intermediate assignment $F^\circ$ in $\sq{F}$ is accepted by $\V_q$ with probability $1$.
By construction, $F^\circ \in \mix(F_\sss, F_\ttt) = \mix\bigl((f_\sss)_E^k, (f_\ttt)_E^k\bigr)$.
Since $f_\sss$ and $f_\ttt$ differ in a single vertex,
by \cref{lem:PSPACE:W},
$\W_q$ accepts $F^\circ$ with probability $1$.
Since $f_\sss$ and $f_\ttt$ satisfy $G$,
we have $\psi_{e_i}\bigl(F^\circ(e_1, \ldots, e_k)[e_i]\bigr) = 1$
for every edge $k$-tuple $(e_1, \ldots, e_k) \in E^k$ and every $i \in [k]$.
Therefore, $\V_q$ accepts $F^\circ$ with probability $1$.
In particular, $\opt_{\V_q}(F_\sss \reco F_\ttt) = 1$, as desired.

Suppose now that $f_\sss$ and $f_\ttt$ differ in multiple vertices.
Since $\opt_G(f_\sss \reco f_\ttt) = 1$,
there exists a reconfiguration sequence $\sq{f} = (f^{(1)}, \ldots, f^{(T)})$ from $f_\sss$ to $f_\ttt$ 
consisting of satisfying assignments for $G$.
For each $t \in [T]$, define $F^{(t)} \defeq (f^{(t)})_E^k$.
By applying the above argument to each pair of $f^{(t)}$ and $f^{(t+1)}$,
we have $\opt_{\V_q}(F^{(t)} \reco F^{(t+1)}) = 1$,
implying that
$\opt_{\V_q}(F_\sss \reco F_\ttt) = 1$,
as desired.
\end{proof}

We then show the soundness.

\begin{lemma}
\label{lem:PSPACE:soundness}
    If $\opt_G(f_\sss \reco f_\ttt) < s$,
    then $\opt_{\V_q}(F_\sss \reco F_\ttt) < \frac{1}{2^{q-1}}+\epsilon$.
\end{lemma}

\noindent
By applying \cref{lem:PSPACE:completeness,lem:PSPACE:soundness}, we can prove \cref{lem:PSPACE}.

\begin{proof}[Proof of \cref{lem:PSPACE}]
Construct a $q$-CSP instance $H$ that represents $\V_q$,
whose vertex set is $E^k$ and alphabet is $(\Sigma^2)^k$.
By \cref{lem:PSPACE:completeness},
if $\opt_G(f_\sss \reco f_\ttt) = 1$, then $\opt_H(F_\sss \reco F_\ttt) = 1$.
By \cref{lem:PSPACE:soundness},
if $\opt_G(f_\sss \reco f_\ttt) < s$, then $\opt_H(F_\sss \reco F_\ttt) < \frac{1}{2^{q-1}} + \epsilon$.
In order to show that the underlying hypergraph of $H$ is regular,
we show that $\W_q$ is ``regular'' in a sense that
any edge $k$-tuple $\sq{e} \in E^k$ appears in $\W_q$'s query with the same probability.
Let
$(\sq{\rv{x}}_1, \ldots, \sq{\rv{x}}_q, \sq{\rv{e}}_1, \ldots, \sq{\rv{e}}_q, \rv{\pi}_1, \ldots, \rv{\pi}_q) \sim \W_q$ denote
$q$ vertex $k$-tuples, $q$ edge $k$-tuples, and $q$ permutations selected by $\W_q$.
Observe that any vertex $k$-tuple $\sq{x} \in V^k$
appears in $\{ \sq{\rv{x}}_1, \ldots, \sq{\rv{x}}_q \}$ with the same probability.
Since each $\sq{\rv{e}}_i$ is uniformly distributed in $E^k$ as
the underlying graph of $G$ is regular,
any edge $k$-tuple $\sq{e} \in E^k$
appears in $\{ \sq{\rv{e}}_1, \ldots, \sq{\rv{e}}_q \}$ with the same probability.
Since each $\sq{\rv{e}}_i \circ \rv{\pi}_i$ is obtained by randomly ordering $\sq{\rv{e}}_i$,
any edge $k$-tuple $\sq{e} \in E^k$
appears in $\{ \sq{\rv{e}}_1 \circ \rv{\pi}_1, \ldots, \sq{\rv{e}}_q \circ \rv{\pi}_q \}$ with the same probability.
Therefore, each ``vertex'' $\sq{e} \in E^k$ of $H$ appears in the same number of hyperedges of $H$.
Consequently, there exists a polynomial-time reduction from
\prb{Gap$_{1,s}$ 2-CSP$_\sigma$ Reconfiguration} whose underlying graph is $\Delta$-regular to
\prb{Gap$_{1,\frac{1}{2^{q-1}} + \epsilon}$ $q$-CSP$_{\sigma^{2k}}$ Reconfiguration}
whose underlying hypergraph is regular,
as desired.
\end{proof}

The remainder of this section is devoted to the proof of \cref{lem:PSPACE:soundness}.

\paragraph{Proof of \cref{lem:PSPACE:soundness}.}
Assume that $\opt_G(f_\sss \reco f_\ttt) < s$.
Let $\sq{F} = (F^{(1)}, \ldots, F^{(T)})$ be any reconfiguration sequence from $F_\sss$ to $F_\ttt$.
We would like to show that $\val_{\V_q}(\sq{F}) < \frac{1}{2^{q-1}}+\epsilon$.
Suppose first that there exists some function $F^{(t)}$ in $\sq{F}$ such that
\begin{align}
    \Pr\Bigl[\W_q^{F^{(t)}} = 1\Bigr] < \frac{1}{2^{q-1}} + \epsilon.
\end{align}
Since $\Pr[\V_q^{F^{(t)}} = 1] \leq \Pr[\W_q^{F^{(t)}} = 1]$ by construction,
$\V_q$ accepts $F^{(t)}$ with probability below $\frac{1}{2^{q-1}}+\epsilon$; i.e.,
$\val_{\V_q}(\sq{F}) < \frac{1}{2^{q-1}}+\epsilon$, as desired.
Hereafter, we assume that for every function $F^{(t)}$ in $\sq{F}$,
\begin{align}
    \Pr\Bigl[\W_q^{F^{(t)}} = 1\Bigr] \geq \frac{1}{2^{q-1}} + \epsilon.
\end{align}
By applying \cref{lem:PSPACE:W} to each function $F^{(t)}$ in $\sq{F}$,
there exists an assignment $f^{(t)} \colon V \to \Sigma$ such that
\begin{align}
\label{eq:PSPACE:soundness:ft}
    \Pr_{\sq{\rv{e}} \in E^k}\left[
        \close{
            F^{(t)}(\sq{\rv{e}})
        }{
            (f^{(t)})_E^k(\sq{\rv{e}})
        }{k^{-\e}}
    \right]
    \geq \frac{1}{2}+\frac{\epsilon}{4}.
\end{align}
Consider a sequence $\sq{f} = (f^{(1)}, \ldots, f^{(T)})$ over $\Sigma^V$ such that
$f^{(1)} \defeq f_\sss$, $f^{(T)} \defeq f_\ttt$, and
$f^{(t)}$ for each $2 \leq t \leq T-1$ is any function satisfying \cref{eq:PSPACE:soundness:ft}.
Although $\sq{f}$ is not necessarily a reconfiguration sequence,
$f^{(t)}$ and $f^{(t+1)}$ can be made arbitrarily close by taking $k$ sufficiently large,
as described below.

\begin{claim}
\label{clm:PSPACE:soundness:F1F2f1f2}
    Let $F_1,F_2 \colon E^k \to (\Sigma^2)^k$ be any edge $k$-tuple functions, and
    $f_1,f_2 \colon V \to \Sigma$ be any vertex functions.
    Suppose that $F_1$ and $F_2$ differ in at most a single coordinate, and
    \begin{align}
    \begin{aligned}
        \Pr_{\sq{\rv{e}} \in E^k}\left[
            \close{
                F_1(\sq{\rv{e}})}{
                f_{1,E}^k(\sq{\rv{e}})}{
                k^{-\e}}
        \right] & \geq \frac{1}{2}+\frac{\epsilon}{4}, \\
        \Pr_{\sq{\rv{e}} \in E^k}\left[
            \close{
                F_2(\sq{\rv{e}})}{
                f_{2,E}^k(\sq{\rv{e}})}{
                k^{-\e}}
        \right] & \geq \frac{1}{2}+\frac{\epsilon}{4},
    \end{aligned}
    \end{align}
    where $f_{1,E}, f_{2,E} \colon E \to \Sigma^2$ denote the induced edge functions of $f_1,f_2$, respectively, and
    $\e \in (0,1)$ is a universal constant appearing in \cref{lem:PSPACE:W}.
    Then,
    \begin{align}
        \dist(f_1,f_2)
        \leq \delta(k,\epsilon)
        \defeq \sqrt{\frac{1}{2k}\ln \frac{2}{\epsilon}} + 3k^{-\e}.
    \end{align}
\end{claim}
\begin{proof} 
By the triangle inequality, for any edge $k$-tuple $\sq{e} \in E^k$, we have
\begin{align}
\label{eq:PSPACE:soundness:F1F2f1f2:triangle}
\begin{aligned}
    \dist\bigl(f_{1,E}^k(\sq{e}), f_{2,E}^k(\sq{e})\bigr)
    & \leq \dist\bigl(f_{1,E}^k(\sq{e}), F_1(\sq{e})\bigr)
    + \dist\bigl(F_1(\sq{e}), F_2(\sq{e})\bigr)
    + \dist\bigl(F_2(\sq{e}), f_{2,E}^k(\sq{e})\bigr) \\
    & \leq \dist\bigl(f_{1,E}^k(\sq{e}), F_1(\sq{e})\bigr)
    + \frac{1}{|E|^k}
    + \dist\bigl(F_2(\sq{e}), f_{2,E}^k(\sq{e})\bigr) \\
    & \underbrace{\leq}_{|E| \geq 2}
    \dist\bigl(F_1(\sq{e}), f_{1,E}^k(\sq{e})\bigr)
    + \dist\bigl(F_2(\sq{e}), f_{2,E}^k(\sq{e})\bigr)
    + k^{-\e}.
\end{aligned}
\end{align}
By assumption and \cref{eq:PSPACE:soundness:F1F2f1f2:triangle}, we have
\begin{align}
\begin{aligned}
    \Pr_{\sq{\rv{e}} \in E^k}\left[
        \close{
            f_{1,E}^k(\sq{\rv{e}})}{
            f_{2,E}^k(\sq{\rv{e}})}{
            3k^{-\e}}
    \right]
    & \geq \Pr_{\sq{\rv{e}} \in E^k}\left[
        \close{
            F_1(\sq{\rv{e}})}{
            f_{1,E}^k(\sq{\rv{e}})}{
            k^{-\e}}
        \text{ and }
        \close{
            F_2(\sq{\rv{e}})}{
            f_{2,E}^k(\sq{\rv{e}})}{
            k^{-\e}}
    \right] \\
    & \geq \underbrace{\Pr_{\sq{\rv{e}} \in E^k}\left[
        \close{
            F_1(\sq{\rv{e}})}{
            f_{1,E}^k(\sq{\rv{e}})}{
            k^{-\e}}
    \right]}_{\geq \frac{1}{2}+\frac{\epsilon}{4}}
    + \underbrace{\Pr_{\sq{\rv{e}} \in E^k}\left[
        \close{
            F_2(\sq{\rv{e}})}{
            f_{2,E}^k(\sq{\rv{e}})}{
            k^{-\e}}
    \right]}_{\geq \frac{1}{2}+\frac{\epsilon}{4}} - 1
    \geq \frac{\epsilon}{2}. 
\end{aligned}
\end{align}
On the other hand, 
$\dist\bigl(f_{1,E}^k(\sq{\rv{e}}), f_{2,E}^k(\sq{\rv{e}})\bigr)$ over
a random edge $k$-tuple $\sq{\rv{e}} = (\rv{x}_1\rv{y}_1, \ldots, \rv{x}_k\rv{y}_k) \in E^k$
can be thought of as 
the sum of $k$ independent random variables $\rv{Z}_1, \ldots, \rv{Z}_k$ such that
\begin{align}
    \rv{Z}_i
    \defeq \frac{1}{k} \bigl\llbracket
        f_{1,E}(\rv{x}_i\rv{y}_i) \neq f_{2,E}(\rv{x}_i\rv{y}_i)
    \bigr\rrbracket
    = \frac{1}{k} \bigl\llbracket
        f_1(\rv{x}_i) \neq f_2(\rv{x}_i) \text{ or } f_1(\rv{y}_i) \neq f_2(\rv{y}_i)
    \bigr\rrbracket.
\end{align}
Let $\delta \defeq \dist(f_1,f_2)$.
Note that $0 \leq \rv{Z}_i \leq \frac{1}{k}$ for each $i \in [k]$, and
that $\E\bigl[\dist(f_{1,E}^k(\sq{\rv{e}}), f_{2,E}^k(\sq{\rv{e}}))\bigr]
= \E\left[\sum_{1 \leq i \leq k} \rv{Z}_i\right] \geq \delta$
because
\begin{align}
\begin{aligned}
    \E\bigl[\rv{Z}_i\bigr]
    & = \E_{\rv{x}_i\rv{y}_i \in E} \Bigl[
        \tfrac{1}{k} \bigl\llbracket
            f_1(\rv{x}_i) \neq f_2(\rv{x}_i) \text{ or } f_1(\rv{y}_i) \neq f_2(\rv{y}_i)
        \bigr\rrbracket
    \Bigr] \\
    & \geq \tfrac{1}{k} \Pr_{\rv{x}_i\rv{y}_i \in E}\bigl[f_1(\rv{x}_i) \neq f_2(\rv{x}_i)\bigr]
    = \tfrac{1}{k} \Pr_{\rv{x}_i \in V}\bigl[f_1(\rv{x}_i) \neq f_2(\rv{x}_i)\bigr]
    = \frac{\delta}{k},
\end{aligned}
\end{align}
where the second-to-the-last equality holds because the underlying graph of $G$ is regular.
By Hoeffding's inequality, for any real $t > 0$, we have
\begin{align}
    \Pr\left[
        \sum_{1 \leq i \leq k} \rv{Z}_i
        \leq \E\Biggl[\sum_{1 \leq i \leq k} \rv{Z}_i\Biggr] - t
    \right]
    \leq \exp\left(-\frac{2t^2}{\sum_{1 \leq i \leq k} \left(\frac{1}{k}\right)^2}\right)
    = \exp(-2k t^2).
\end{align}
Setting $t \defeq \delta - 3k^{-\e}$, we obtain
\begin{align}
    \Pr_{\sq{\rv{e}} \in E^k}\left[
        \close{f_{1,E}^k(\sq{\rv{e}})}{f_{2,E}^k(\sq{\rv{e}})}{3k^{-\e}}
    \right]
    \leq \exp\bigl(-2k(\delta-3k^{-\e})^2\bigr).
\end{align}
Consequently, we must have
\begin{align}
\begin{aligned}
    & \exp\bigl(-2k(\delta-3k^{-\e})^2\bigr)
    \geq \Pr_{\sq{\rv{e}} \in E^k}\left[
        \close{
            f_{1,E}^k(\sq{\rv{e}})}{
            f_{2,E}^k(\sq{\rv{e}})
            }{3k^{-\e}}
    \right]
    \geq \frac{\epsilon}{2}, \\
    & \therefore \dist(f_1,f_2)
        = \delta
        \leq \sqrt{\frac{1}{2k}\ln \frac{2}{\epsilon}} + 3k^{-\e},
\end{aligned}
\end{align}
as desired.
\end{proof}

Since each quadruple of $F^{(t)}$, $F^{(t+1)}$, $f^{(t)}$, and $f^{(t+1)}$ satisfy the condition of \cref{clm:PSPACE:soundness:F1F2f1f2},
we have $\dist(f^{(t)}, f^{(t+1)}) \leq \delta(k,\epsilon)$.
We then show that if
every adjacent pair of assignments in $\sq{f}$ are sufficiently close,
then some function in $\sq{f}$ has a value below $\frac{1+s}{2}$.

\begin{claim}
\label{clm:PSPACE:soundness:interpolate}
Let $\sq{f} = (f^{(1)}, \ldots, f^{(T)})$ be a sequence over $\Sigma^V$
from $f_\sss$ to $f_\ttt$.
If $\dist(f^{(t)}, f^{(t+1)}) \leq \gamma$ for every $t \in [T-1]$,
where $\gamma \defeq \frac{1-s}{4}$,
then there exists a function $f^{(t)}$ in $\sq{f}$ with value less than $\frac{1+s}{2} < 1$.
\end{claim}
\begin{proof} 
For each $t \in [T-1]$,
let $\sq{g}_t$ be a trivial reconfiguration sequence from $f^{(t)}$ to $f^{(t+1)}$
obtained by merely changing the assignments of at most $\gamma n$ vertices.
Concatenating $\sq{g}_t$ for every $t \in [T-1]$,
we obtain a reconfiguration sequence $\sq{g}$ from $f_\sss$ to $f_\ttt$.
By assumption, $\val_G(\sq{g}) < s$.
In particular, there exists some assignment $g^\circ$ in $\sq{g}$ such that
$\val_G(g^\circ) < s$.
Suppose that $g^\circ$ appears in $\sq{g}_t$.
Observe that for two assignments for $G$ that differ in at most a single vertex,
their values differ in at most $\frac{\Delta}{|E|}$.
Since $g^\circ$ and $f^{(t)}$ differ in at most $\gamma n$ vertices, we derive
\begin{align}
    \val_G(f^{(t)})
    \leq \val_G(g^\circ) + \frac{\Delta}{|E|}\gamma n
    \underbrace{=}_{|E| = \frac{|V|\Delta}{2}}
        \val_G(g^\circ) + \frac{2}{|V|} \frac{1-s}{4} n
    < s + \frac{1-s}{2}
    = \frac{1+s}{2},
\end{align}
as desired.
\end{proof}

By taking $k$ sufficiently large so that
\begin{align}
    \delta(k,\epsilon)
    = \sqrt{\frac{1}{2k}\ln \frac{2}{\epsilon}} + 3k^{-\e}
    \leq \frac{1-s}{4},
\end{align}
\cref{clm:PSPACE:soundness:interpolate} derives that
there exists some assignment $f^{(t)}$ such that
$\val_G(f^{(t)}) < \frac{1+s}{2}$.
We are now ready to show that $\V_q$ accepts $F^{(t)}$ with probability at most $\frac{1}{2^{q-1}} + \epsilon$.

\begin{claim}
\label{clm:PSPACE:soundness:last}
    Let $F \colon E^k \to (\Sigma^2)^k$ be an edge $k$-tuple function and
    $f \colon V \to \Sigma$ be a vertex function
    such that
    \begin{align}
    \begin{aligned}
        \Pr_{\sq{\rv{e}} \in E^k}\left[
            \close{F(\sq{\rv{e}})}{f_E^k(\sq{\rv{e}})}{k^{-\e}}
        \right] & \geq \frac{1}{2} + \frac{\epsilon}{4}, \\
        \val_G(f) & < \frac{1+s}{2}.
    \end{aligned}
    \end{align}
    Then, $\V_q$ accepts $F$ with probability less than $\frac{1}{2^{q-1}}$.
\end{claim}

In the proof of \cref{clm:PSPACE:soundness:last},
we use the following claim,
whose proof is obtained by applying \cref{lem:Johnson:hitting} and deferred to \cref{app:PSPACE}.
\begin{claim}[$*$]
\label{clm:PSPACE:soundness:hitting}
    For any set $\calS \subseteq E^k$,
    \begin{align}
        \Pr_{(\sq{\rv{e}}_1, \ldots, \sq{\rv{e}}_q, \rv{\pi}_1, \ldots, \rv{\pi}_q) \sim \W_q}\left[
            \bigwedge_{1 \leq i \leq q}
            \sq{\rv{e}}_i \circ \rv{\pi}_i \in \calS
        \right] \leq
        \left(\frac{|\calS|}{|E^k|}\right)^q
        + \bigO_q\left(\frac{\ell}{k}\right) \left(\frac{|\calS|}{|E^k|}\right)
        + \bigO_{k,q}\left(\frac{1}{n}\right),
    \end{align}
    where $(\sq{\rv{e}}_1, \ldots, \sq{\rv{e}}_q, \rv{\pi}_1, \ldots, \rv{\pi}_q) \sim \W_q$
    denote $q$ edge $k$-tuples and $q$ permutations selected by $\W_q$.
\end{claim}

\begin{proof}[Proof of \cref{clm:PSPACE:soundness:last}]
Define
$\calD$ as the set of edge $k$-tuples $\sq{e} \in E^k$ such that
$F(\sq{e})$ does not approximately agree with $f_E^k(\sq{e})$, and
$S$ as the set of edges $e$ satisfied by $f$; namely,
\begin{align}
\begin{aligned}
    \calD & \defeq \left\{
        \sq{e} \in E^k \;\middle|\; \far{F(\sq{e})}{f_E^k(\sq{e})}{k^{-\e}}
    \right\}, \\
    S & \defeq \Bigl\{
        e \in E \Bigm| \psi_e\bigl(f_E(e)\bigr) = 1
    \Bigr\}.
\end{aligned}
\end{align}
Let $(\sq{\rv{e}}_1, \ldots, \sq{\rv{e}}_q, \rv{\pi}_1, \ldots, \rv{\pi}_q) \sim \W_q$.
Observe that for $\V_q$ to accept $F$,
(at least) either of the following two conditions must hold:
\begin{description}
    \item[(C1)] For every $i \in [q]$,
        it holds that $\sq{\rv{e}}_i \circ \rv{\pi}_i \in \calD$.
    \item[(C2)] For some $i \in [q]$,
        at least $k - k^{1-\e}$ edges of $\sq{\rv{e}}_i \circ \rv{\pi}_i$ are in $S$.
\end{description}
To see why this is true, suppose that
$\sq{\rv{e}}_i \circ \rv{\pi}_i \notin \calD$
for some $i \in [q]$.
Denote $(\rv{e}_{i,1}, \ldots, \rv{e}_{i,k}) \defeq \sq{\rv{e}}_i \circ \rv{\pi}_i$.
Since $\close{F(\sq{\rv{e}}_i \circ \rv{\pi}_i)}{f_E^k(\sq{\rv{e}}_i \circ \rv{\pi}_i)}{k^{-\e}}$ by assumption,
    $F(\sq{\rv{e}}_i \circ \rv{\pi}_i)[\rv{e}_{i,j}] = f_E(\rv{e}_{i,j})$ holds for 
    at least $k - k^{1-\e}$ edges of $\rv{e}_{i,1}, \ldots, \rv{e}_{i,k}$.
Therefore, $F(\sq{\rv{e}}_i \circ \rv{\pi}_i)$ satisfies all of the $k$ edges $\rv{e}_{i,1}, \ldots, \rv{e}_{i,k}$
only if at least $k - k^{1-\e}$ of them are in $S$.

Since $\frac{|\calD|}{|E^k|} \leq \frac{1}{2}-\frac{\epsilon}{4}$ by assumption,
we apply \cref{clm:PSPACE:soundness:hitting} to $\calD$ and derive
\begin{align}
\begin{aligned}
    \Pr_{(\sq{\rv{e}}_1, \ldots, \sq{\rv{e}}_q, \rv{\pi}_1, \ldots, \rv{\pi}_q) \sim \W_q}
        \bigl[\text{(C1) holds}\bigr]
    & = \Pr_{(\sq{\rv{e}}_1, \ldots, \sq{\rv{e}}_q, \rv{\pi}_1, \ldots, \rv{\pi}_q) \sim \W_q}\Biggl[
        \bigwedge_{1 \leq i \leq q}
        \sq{\rv{e}}_i \circ \rv{\pi}_i \in \calD
    \Biggr] \\
    & \leq \left(\frac{|\calD|}{|E^k|}\right)^q
        + \bigO\left(\frac{\ell}{k}\right)\left(\frac{|\calD|}{|E^k|}\right)
        + \bigO_{q,k}\left(\frac{1}{n}\right) \\
    & \leq \left(\frac{1}{2}-\frac{\epsilon}{4}\right)^q
        + \bigO\left(\frac{\ell}{k}\right)
        + \bigO_{q,k}\left(\frac{1}{n}\right) \\
    & \leq \frac{1}{2^q}
        + \bigO_q\left(\frac{1}{\sqrt{k}}\right)
        + \bigO_{q,k}\left(\frac{1}{n}\right).
\end{aligned}
\end{align}
Since each $\sq{\rv{e}}_i \circ \rv{\pi}_i$ is uniformly distributed in $E^k$ and
$\frac{|S|}{|E|} \leq \frac{1+s}{2}$ by assumption,
we derive for sufficiently large $k$,
\begin{align}
\begin{aligned}
    \Pr_{(\sq{\rv{e}}_1, \ldots, \sq{\rv{e}}_q, \rv{\pi}_1, \ldots, \rv{\pi}_q) \sim \W_q}
        \bigl[\text{(C2) holds}\bigr]
    & = \Pr_{(\sq{\rv{e}}_1, \ldots, \sq{\rv{e}}_q, \rv{\pi}_1, \ldots, \rv{\pi}_q) \sim \W_q}\Biggl[
        \bigvee_{1 \leq i \leq q}
        \left|\bigl\{
            \rv{e}_{i,j} \in \sq{\rv{e}}_i \circ \rv{\pi}_i
            \bigm|
            \rv{e}_{i,j} \in S
        \bigr\}\right|
        \geq k - k^{1-\e}
    \Biggr] \\
    & \leq \sum_{1 \leq i \leq q}
        \Pr_{(\rv{e}_1, \ldots, \rv{e}_k) \in E^k}\Bigl[
            |\{ \rv{e}_1, \ldots, \rv{e}_k \} \cap S|
                \geq k-k^{1-\e}
        \Bigr] \\
    & \leq q \binom{k}{k-k^{1-\e}} \left(\frac{|S|}{|E|}\right)^{k-k^{1-\e}} \\
    & \leq q \left(\frac{\rme k}{k^{1-\e}}\right)^{k^{1-\e}} \left(\frac{1+s}{2}\right)^{k-k^{1-\e}} \\
    & = q \cdot \exp\bigl(\Theta(k^{1-\e} \log k^{1-\e})\bigr) \cdot \exp\bigl(-\Theta(k)\bigr) \\
    & = q \cdot \exp\bigl(-\Omega(k)\bigr).
\end{aligned}
\end{align}
Consequently, for sufficiently large $k$ and $n$, we obtain
\begin{align}
\begin{aligned}
    \Pr\Bigl[\V_q^F = 1\Bigr]
    & \leq \Pr\bigl[\text{(C1) holds}\bigr] + \Pr\bigl[\text{(C2) holds}\bigr] \\
    & \leq \frac{1}{2^q} +
    \underbrace{
        \bigO_q\left(\frac{1}{\sqrt{k}}\right) + \bigO_{q,k}\left(\frac{1}{n}\right) + q \cdot \exp\bigl(-\Omega(k)\bigr)
    }_{\leq \frac{1}{2^q}}
    \leq \frac{1}{2^{q-1}},
\end{aligned}
\end{align}
as desired.
\end{proof}

By \cref{clm:PSPACE:soundness:last},
$\V_q$ accepts $F^{(t)}$ with probability at most $\frac{1}{2^{q-1}}$, implying that
$\val_{\V_q}(\sq{F}) \leq \val_{\V_q}(F^{(t)}) \leq \frac{1}{2^{q-1}} < \frac{1}{2^{q-1}} + \epsilon$,
which completes the proof of \cref{lem:PSPACE:soundness}.
\qed

\subsection{\texorpdfstring{%
Proof of \cref{lem:PSPACE:W}
}{%
Proof of Lemma \ref{lem:PSPACE:W}
}}
\label{sec:PSPACE:W}

We prove \cref{lem:PSPACE:W}, i.e., the completeness and soundness of $\W_q$.
Let $\epsilon > 0$ be a real.
Let $G = (V,E)$ be a $\Delta$-regular $n$-vertex graph, and
$F \colon E^k \to (\Sigma^2)^k$ be an edge $k$-tuple function.
Hereafter, we assume that $k$ and $n$ are sufficiently large so that \cref{thm:TDP} can be applied.
We first prove the completeness part.

\begin{proof}[Proof of the completeness part of \cref{lem:PSPACE:W}]

Suppose that $F \in \mix(f_E^k, g_E^k)$ for two functions $f,g \colon V \to \Sigma$
that differ in at most a single vertex.
We will show that
$\W_q$ never rejects during the iteration in lines \ref{linum:W:start}--\ref{linum:W:end}
for each $i \in [q-1]$
by the following case analysis:

\begin{description}
    \item[(Case 1)]
        Suppose that $f = g$.
        Then, $F$ is the $k$-wise direct product of $f_E$; thus,
        \begin{align}
            \bigl(F(\sq{\rv{e}}_i \circ \rv{\pi}_i) \circ \rv{\pi}_i^{-1}\bigr)\bigl[\sq{\rv{x}}_i|_{\rv{I}_i}\bigr]
            =
            \bigl(F(\sq{\rv{e}}_{i+1} \circ \rv{\pi}_{i+1}) \circ \rv{\pi}_{i+1}^{-1}\bigr)\bigl[\sq{\rv{x}}_{i+1}|_{\rv{I}_i}\bigr],
        \end{align}
        implying that $\W_q$ does not reject.
    \item[(Case 2)]
        Suppose that $f$ and $g$ differ in a single vertex, say $x^* \in V$.
        Note that for each edge $k$-tuple $\sq{e} = (x_1y_1, \ldots, x_k y_k) \in E^k$,
        we have $F(\sq{e})[x_i] = f(x_i) = g(x_i)$ whenever $x_i \neq x^*$.
    \begin{description}
        \item[(Case 2-1)] 
            Suppose that
            $|\{ \rv{x}_{i,1}, \ldots, \rv{x}_{i,k}, \rv{x}_{i+1,1}, \ldots, \rv{x}_{i+1,k} \}| < 2k-\ell$.
            Obviously, $\W_q$ does not reject.
        \item[(Case 2-2)] 
            Suppose that
            $|\{ \rv{x}_{i,1}, \ldots, \rv{x}_{i,k}, \rv{x}_{i+1,1}, \ldots, \rv{x}_{i+1,k} \}| \geq 2k-\ell$.
            Since $\rv{x}_{i,j} = \rv{x}_{i+1,j}$ for every $j \in \rv{I}_i$, we have
            $|\{ \rv{x}_{i,1}, \ldots, \rv{x}_{i,k}, \rv{x}_{i+1,1}, \ldots, \rv{x}_{i+1,k} \}| = 2k-\ell$.
            Thus, $x^*$ appears in $\sq{\rv{x}}_i|_{\rv{I}_i} = \sq{\rv{x}}_{i+1}|_{\rv{I}_i}$ at most once.
            \begin{description}
                \item[(Case 2-2-1)] 
                    If $x^*$ does not appear in $\sq{\rv{x}}_i|_{\rv{I}_i} = \sq{\rv{x}}_{i+1}|_{\rv{I}_i}$,
                    then we have
                    \begin{align}
                        \bigl(F(\sq{\rv{e}}_i \circ \rv{\pi}_i) \circ \rv{\pi}_i^{-1}\bigr)\bigl[\sq{\rv{x}}_i|_{\rv{I}_i}\bigr]
                        =
                        \bigl(F(\sq{\rv{e}}_{i+1} \circ \rv{\pi}_{i+1}) \circ \rv{\pi}_{i+1}^{-1}\bigr)\bigl[\sq{\rv{x}}_{i+1}|_{\rv{I}_i}\bigr],
                    \end{align}
                    implying that $\W_q$ does not reject.
                \item[(Case 2-2-2)] 
                    If $x^*$ appears in $\sq{\rv{x}}_i|_{\rv{I}_i} = \sq{\rv{x}}_{i+1}|_{\rv{I}_i}$ once,
                    then there exists a unique $j^* \in \rv{I}_i$ such that $x_{i,j^*} = x_{i+1,j^*} = x^*$.
                    For each $j \in \rv{I}_i \setminus \{j^*\}$,
                    we have
                    \begin{align}
                        \bigl(F(\sq{\rv{e}}_i \circ \rv{\pi}_i) \circ \rv{\pi}_i^{-1}\bigr)\bigl[\rv{x}_{i,j}\bigr] =
                        \bigl(F(\sq{\rv{e}}_{i+1} \circ \rv{\pi}_{i+1}) \circ \rv{\pi}_{i+1}^{-1}\bigr)\bigl[\rv{x}_{i+1,j}\bigr].
                    \end{align}
                    Therefore, we have
                    \begin{align}
                        \close{
                            \bigl(F(\sq{\rv{e}}_i \circ \rv{\pi}_i) \circ \rv{\pi}_i^{-1}\bigr)\bigl[\sq{\rv{x}}_i|_{\rv{I}_i}\bigr]}{
                            \bigl(F(\sq{\rv{e}}_{i+1} \circ \rv{\pi}_{i+1}) \circ \rv{\pi}_{i+1}^{-1}\bigr)\bigl[\sq{\rv{x}}_{i+1}|_{\rv{I}_i}\bigr]}{
                            1/\ell},
                    \end{align}
                    implying that $\W_q$ does not reject.
                    \qedhere
            \end{description}
    \end{description}
\end{description}
\end{proof}

In the remainder of this subsection, we prove the soundness part of \cref{lem:PSPACE:W}.
Suppose that $\W_q$ accepts $F$ with probability at least $\frac{1}{2^{q-1}}+\epsilon$.
Let $\Ftuple \colon V^k \times [\Delta]^k \to \Sigma^k$ be a $k$-tuple function
whose value is determined based on $F$ by the following procedure.

\begin{itembox}[l]{\textbf{$k$-tuple function $\Ftuple \colon V^k \times [\Delta]^k \to \Sigma^k$}}
\begin{algorithmic}[1]
    \item[\textbf{Input:}]
        a vertex $k$-tuple $\sq{x} = (x_1, \ldots, x_k) \in V^k$, and
        an integer $k$-tuple $\sq{a} = (a_1, \ldots, a_k) \in [\Delta]^k$.
    \For{\textbf{each} $i \in [k]$}
        \State let $e_i \in E(x_i)$ be the \nth{$a_i$} edge incident to $x_i$.
    \EndFor
    \State \Return $F(e_1, \ldots, e_k)[\sq{x}]$.
\end{algorithmic}
\end{itembox}

\noindent
Let $\Ttuple$ be a tolerant $q$-query direct product tester for
a $k$-tuple function described as follows.

\begin{itembox}[l]{\textbf{Tolerant $q$-query direct product tester $\Ttuple$ for a $k$-tuple function}}
\begin{algorithmic}[1]
    \item[\textbf{Input:}]
        positive integers $\ell$ and $m$ with $1 \leq m \leq \ell \leq k$.
    \item[\textbf{Oracle access:}]
        a $k$-tuple function $\Ftuple \colon V^k \times \Rnd \to \Sigma^k$.
    \State sample $\sq{\rv{x}}_1 = (\rv{x}_{1,1}, \ldots, \rv{x}_{1,k})$ from $V^k$ uniformly.
    \For{\textbf{each} $i \in [q-1]$}
        \State sample $\rv{I}_i$ from $\binom{[k]}{\ell}$ uniformly.
        \State sample $\sq{\rv{x}}_{i+1} = (\rv{x}_{i+1,1}, \ldots, \rv{x}_{i+1,k})$ from $V^k$ uniformly
        conditioned on $\sq{\rv{x}}_{i+1}|_{\rv{I}_i} = \sq{\rv{x}}_i|_{\rv{I}_i}$.
    \EndFor
    \For{\textbf{each} $i \in [q]$}
        \State sample $\rv{\pi}_i$ from $\sym_k$ uniformly.
        \State sample $\rv{\rnd}_i$ from $\Rnd$ uniformly.
        \State read $\Ftuple(\sq{\rv{x}}_i \circ \rv{\pi}_i; \rv{\rnd}_i)$.
    \EndFor
    \For{\textbf{each} $i \in [q-1]$}
        \If{$\far{
        \bigl(\Ftuple(\sq{\rv{x}}_i \circ \rv{\pi}_i; \rv{\rnd}_i) \circ \rv{\pi}_i^{-1}\bigr)|_{\rv{I}_i}}{
        \bigl(\Ftuple(\sq{\rv{x}}_{i+1} \circ \rv{\pi}_{i+1}; \rv{\rnd}_{i+1}) \circ \rv{\pi}_{i+1}^{-1}\bigr)|_{\rv{I}_i}}{
        m/\ell}$}
            \State \Return $0$.
        \EndIf
    \EndFor
    \State \Return $1$.
\end{algorithmic}
\end{itembox}

\noindent
Consider running $\Ttuple(\ell,1)$ on $\Ftuple$, where $\Rnd \defeq [\Delta]^k$.
We show that
$\W_q$ and $\Ttuple(\ell,1)$ have nearly the same acceptance probability,
whose proof is deferred to \cref{app:PSPACE}.

\begin{claim}[$*$]
\label{clm:PSPACE:W:Ttuple}
    If $\W_q$ accepts $F$ with probability at least $p$,
    then $\Ttuple(\ell,1)$ accepts $\Ftuple$ with probability at least
    $p - \frac{qk^2}{n}$.
\end{claim}

Let $\Fset \colon \binom{V}{k} \times [\Delta]^k \times \sym_k \to \Sigma^k$
be a $k$-set function
whose value is determined based on $\Ftuple$ by the following procedure.

\begin{itembox}[l]{\textbf{$k$-set function $\Fset \colon \binom{V}{k} \times [\Delta]^k \times \sym_k \to \Sigma^k$}}
\begin{algorithmic}[1]
    \item[\textbf{Input:}]
        a $k$-set $X \in \binom{V}{k}$,
        an integer $k$-tuple $\sq{a} \in [\Delta]^k$, and
        a permutation $\pi \in \sym_k$.
    \State let $(x_1, \ldots, x_k)$ be a canonical ordering of $k$ vertices in $X$.
    \State let $X \circ \pi \defeq (x_{\pi(1)}, \ldots, x_{\pi(k)})$.
    \State \Return $\Ftuple(X \circ \pi; \sq{a}) \circ \pi^{-1}$.
\end{algorithmic}
\end{itembox}

\noindent
Consider running
$\T_q(\ell, 1)$ on $\Fset$, where $\Rnd \defeq [\Delta]^k \times \sym_k$.
(See \cref{sec:TDP} for the description of $\T_q$.)
We show that
$\Ttuple(\ell,1)$ and $\T_q(\ell,1)$ have nearly the same acceptance probability,
whose proof is well-known (cf.~\cite[Section~A]{dinur2014direct} and \cite[Footnote~5]{impagliazzo2012new})
and deferred to \cref{app:PSPACE}.

\begin{claim}[$*$]
\label{clm:PSPACE:W:Tset}
    If $\Ttuple(\ell,1)$ accepts $\Ftuple$ with probability at least $p$,
    then $\T_q(\ell,1)$ accepts $\Fset$ with probability at least $p-\frac{qk^2}{n}$.
\end{claim}

By applying \cref{thm:TDP}, we show that 
if $\T_q(\ell,1)$ accepts $\Fset$ with probability at least $p$,
then there exists a function $f \colon V \to \Sigma$ such that
$\Fset$ approximately agrees with $f^k$ with probability $\approx p^{\frac{1}{q-1}}$.

\begin{claim}
\label{clm:PSPACE:W:f}
    Let $\e \defeq \frac{1}{128}$.
    Suppose that $\T_q(\ell,1)$ accepts $\Fset$ with probability at least $p \geq k^{-\e}$.
    Then, there exists a function $f \colon V \to \Sigma$ such that
\begin{align}
    \Pr_{\substack{
        \rv{X} \in \binom{V}{k} \\
        (\sq{\rv{a}}, \rv{\pi}) \in [\Delta]^k \times \sym_k
    }}\left[
        \close{
            \Fset(\rv{X}; \sq{\rv{a}}, \rv{\pi})}{
            f^k(\rv{X})
        }{\delta}
    \right] \geq p^{\frac{1}{q-1}}(1-k^{-\Omega(1)}),
    \text{ where }
    \delta \defeq k^{-\e}.
\end{align}
\end{claim}
\begin{proof} 
Substituting $\Theta(\sqrt{k})$ for $\ell$ and $1$ for $m$,
we evaluate the value of $\epsilon_{k,\ell,m}$
(see \cref{eq:TDP:epsilon} for the definition) as
\begin{align}
    \epsilon_{k,\ell,m}
    = \tilde{\Theta}\left(\max\left\{
        \frac{\ell}{k},
        \frac{m^{\frac{1}{8}}}{\ell^{\frac{1}{8}}},
        \frac{1}{\ell^{\frac{1}{16}}}
    \right\}\right)
    = \tilde{\Theta}\left(\max\left\{
        \frac{1}{k^{\frac{1}{2}}},
        \frac{1}{k^{\frac{1}{16}}},
        \frac{1}{k^{\frac{1}{32}}}
    \right\}\right)
    = \tilde{\Theta}(\frac{1}{k^{\frac{1}{32}}}).
\end{align}
Suppose that $\T_q(\ell,1)$ accepts $\Fset$ with probability $p \geq k^{-\e}$.
Since $k^{-\e} \geq \sqrt{\epsilon_{k,\ell,m}} = \tilde{\Theta}(k^{-\frac{1}{64}})$,
by \cref{thm:TDP},
there exists a function $f \colon V \to \Sigma$ such that
\begin{align}
    \Pr_{\substack{
        \rv{X} \in \binom{V}{k} \\
        (\sq{\rv{a}}, \rv{\pi}) \in [\Delta]^k \times \sym_k
    }}\left[
        \close{
            \Fset(\rv{X}; \sq{\rv{a}}, \rv{\pi})}{
            f^k(\rv{X})
        }{k^{-\e}}
    \right] \geq p^{\frac{1}{q-1}}\bigl(1-k^{-\Omega(1)}\bigr),
\end{align}
as desired.
\end{proof}

We show that 
if $\Fset$ approximately agrees with $f^k$ with probability $p$,
then so does $\Ftuple$ with probability $\approx p$,
whose proof is well-known (cf.~\cite[Section~A]{dinur2014direct} and \cite[Footnote~5]{impagliazzo2012new})
and deferred to \cref{app:PSPACE}.

\begin{claim}[$*$]
\label{clm:PSPACE:W:close}
\begin{align}
    \Pr_{\substack{
        \rv{X} \in \binom{V}{k} \\
        (\sq{\rv{a}}, \rv{\pi}) \in [\Delta]^k \times \sym_k
    }}\left[
        \close{
            \Fset(\rv{X}; \sq{\rv{a}}, \rv{\pi})}{
            f^k(\rv{X})
        }{\delta}
    \right] \geq p
    \implies 
    \Pr_{\substack{
        \sq{\rv{x}} \in V^k \\
        \sq{\rv{a}} \in [\Delta]^k
    }}\left[
        \close{
            \Ftuple(\sq{\rv{x}}; \sq{\rv{a}})}{
            f^k(\sq{\rv{x}})}{
            \delta}
    \right] \geq p\left(1-\tfrac{k^2}{n}\right).
\end{align}
\end{claim}

We show that 
if $\Ftuple$ approximately agrees with $f^k$ with probability $p$,
then $F$ approximately agrees with $f_E^k$ with probability $\approx p$,
whose proof is based on \cite[Proof of Theorem~6.1]{impagliazzo2012new}.

\begin{claim}
\label{clm:PSPACE:W:fE}
\begin{align}
    \Pr_{\substack{
        \sq{\rv{x}} \in V^k \\
        \sq{\rv{a}} \in [\Delta]^k
    }}\left[
        \close{
            \Ftuple(\sq{\rv{x}}; \sq{\rv{a}})}{
            f^k(\sq{\rv{x}})}{
            \delta}
    \right] \geq p
    \implies
    \Pr_{\sq{\rv{e}} \in E^k}\left[
        \close{
            F(\sq{\rv{e}})}{
            f_E^k(\sq{\rv{e}})}{
            4\delta}
    \right]
    \geq p - \exp\bigl(-\Omega(\delta k)\bigr).
\end{align}
\end{claim}

\begin{proof} 
For an edge $k$-tuple
$\sq{e} = (e_1, \ldots, e_k) = (x_1y_1, \ldots, x_k y_k)$,
let $\prod \sq{e}$ denote the set of vertex $k$-tuples
obtained by selecting $k$ endpoints from $e_1, \ldots, e_k$; namely,
\begin{align}
    \prod \sq{e}
    \defeq \prod_{1 \leq i \leq k} \{ x_i, y_i \}
    = \{x_1, y_1\} \times \cdots \times \{x_k, y_k\}.
\end{align}
By assumption, we have
\begin{align}
    \Pr_{\substack{
        \sq{\rv{x}} = (\rv{x}_1, \ldots, \rv{x}_k) \in V^k \\
        \sq{\rv{a}} = (\rv{a}_1, \ldots, \rv{a}_k) \in [\Delta]^k
    }}\left[
        \close{
            \Ftuple(\sq{\rv{x}}; \sq{\rv{a}})}{
            f^k(\sq{\rv{x}})}{
            \delta}
    \right] \geq p,
    \text{ which implies }
    \Pr_{\substack{
        \sq{\rv{e}} = (\rv{e}_1, \ldots, \rv{e}_k) \in E^k \\
        \sq{\rv{x}} = (\rv{x}_1, \ldots, \rv{x}_k) \in \prod \sq{\rv{e}}
    }}\left[
        \close{
            F(\sq{\rv{e}})[\sq{\rv{x}}]}{
            f^k(\sq{\rv{x}})}{
            \delta}
    \right] \geq p,
\end{align}
where we used the fact that 
the \nth{$\rv{a}_i$} edge incident to $\rv{x}_i$ is uniformly distributed over $E$, where
$\rv{x}_i \in V$ and $\rv{a}_i \in [\Delta]$,
since the underlying graph of $G$ is $\Delta$-regular.
We will claim that
\begin{align}
    \Pr_{\substack{
        \sq{\rv{e}} \in E^k
    }}\left[
        \close{F(\sq{\rv{e}})}{f_E^k(\sq{\rv{e}})}{4\delta}
    \right]
    \geq 1-\frac{1-p}{1-\exp(-\delta k/6)},
\end{align}
which is sufficient to conclude the proof because
\begin{align}
\begin{aligned}
    1-\frac{1-p}{1-\exp(-\delta k/6)}
    & = \frac{p-\exp(-\delta k/6)}{1-\exp(-\delta k/6)} \\
    & \geq \bigl(p-\exp(-\delta k/6)\bigr) \bigl(1+\exp(-\delta k/6)\bigr) \\
    & \geq p-2\exp(-\delta k/6)
    = p - \exp\bigl(-\Omega(\delta k)\bigr).
\end{aligned}
\end{align}
For this purpose, we show the contrapositive; namely,
\begin{align}
\label{eq:PSPACE:W:conc:1}
\begin{aligned}
    \Pr_{\substack{
        \sq{\rv{e}} \in E^k
    }}\left[
        \far{F(\sq{\rv{e}})}{f_E^k(\sq{\rv{e}})}{4\delta}
    \right] \geq \frac{1-p}{1-\exp(-\delta k/6)}
    \implies
    \Pr_{\substack{
        \sq{\rv{e}} \in E^k \\
        \sq{\rv{x}} \in \prod \sq{\rv{e}}
    }}\left[
        \far{F(\sq{\rv{e}})[\sq{\rv{x}}]}{f^k(\sq{\rv{x}})}{\delta}
    \right] \geq 1-p.
\end{aligned}
\end{align}
Suppose that for a fixed edge $k$-tuple $\sq{e} = (e_1, \ldots, e_k) \in E^k$,
\begin{align}
    \far{F(\sq{e})}{f_E^k(\sq{e})}{4\delta}.
\end{align}
Observe that
$k \cdot \dist\bigl(F(\sq{e})[\sq{\rv{x}}], f^k(\sq{\rv{x}})\bigr)$
over a random vertex $k$-tuple $\sq{\rv{x}} = (\rv{x}_1, \ldots, \rv{x}_k) \in \prod \sq{e}$
can be thought of as
the sum of $k$ independent Bernoulli random variables $\rv{Z}_1, \ldots, \rv{Z}_k$ such that
\begin{align}
    \rv{Z}_i \defeq \bigl\llbracket
        F(\sq{e})[\rv{x}_i] \neq f(\rv{x}_i)
    \bigr\rrbracket.
\end{align}
Since $\mu \defeq \E\left[\sum_{1 \leq i \leq k}\rv{Z}_i\right] \geq 2 \delta k$,
by the Chernoff bound, we have
\begin{align}
    \Pr\left[
        \sum_{1 \leq i \leq k} \rv{Z}_i < \frac{\mu}{2}
    \right]
    < 2 \exp\left(-\frac{\left(\frac{1}{2}\right)^2}{3}\mu\right)
    < \exp(-\delta k/6),
\end{align}
implying that
\begin{align}
    \Pr_{\sq{\rv{x}} \in \prod \sq{e}}\left[
        \far{F(\sq{e})[\sq{\rv{x}}]}{f^k(\sq{\rv{x}})}{\delta}
    \right]
    > 1 - \exp(-\delta k/6).
\end{align}
Consequently, by the assumption of \cref{eq:PSPACE:W:conc:1}, we derive
\begin{align}
\begin{aligned}
    & \Pr_{\sq{\rv{e}} \in E^k}\left[
        \Pr_{\sq{\rv{x}} \in \prod \sq{\rv{e}}}\left[
            \far{F(\sq{\rv{e}})[\sq{\rv{x}}]}{f^k(\sq{\rv{x}})}{\delta}
        \right] \geq 1-\exp(-\delta k/6)
    \right] \geq \frac{1-p}{1-\exp(-\delta k/6)}, \\
    & \therefore \Pr_{\substack{
        \sq{\rv{e}} \in E^k \\ 
        \sq{\rv{x}} \in \prod \sq{\rv{e}}
    }}\left[
        \far{F(\sq{\rv{e}})[\sq{\rv{x}}]}{f^k(\sq{\rv{x}})}{\delta}
    \right] \geq \bigl(1-\exp(-\delta k/6)\bigr) \frac{1-p}{1-\exp(-\delta k/6)}
    = 1-p,
\end{aligned}
\end{align}
as desired.
\end{proof}

We are now ready to complete the proof of the soundness part of \cref{lem:PSPACE:W}.

\begin{proof}[Proof of the soundness part of \cref{lem:PSPACE:W}]
Suppose that $\W_q$ accepts $F$ with probability $\frac{1}{2^{q-1}}+\epsilon$.
By \cref{clm:PSPACE:W:Ttuple,clm:PSPACE:W:Tset,clm:PSPACE:W:f,clm:PSPACE:W:close,clm:PSPACE:W:fE},
there exists a function $f \colon V \to \Sigma$ such that
\begin{align}
\begin{aligned}
    \Pr_{\sq{\rv{e}} \in E^k}\left[
        \close{
            F(\sq{\rv{e}})}{
            f_E^k(\sq{\rv{e}})}{
            4\delta}
    \right]
    \geq \left(\frac{1}{2^{q-1}}+\epsilon-\frac{qk^2}{n}-\frac{qk^2}{n}\right)^{\frac{1}{q-1}}
        \left(1-k^{-\Omega(1)}\right)
        \left(1-\frac{k^2}{n}\right)
        - \exp\bigl(-\Omega(\delta k)\bigr),
\end{aligned}
\end{align}
which implies that for sufficiently large $k$ and $n$,
\begin{align}
    \Pr_{\sq{\rv{e}} \in E^k}\left[
        \close{
            F(\sq{\rv{e}})}{
            f_E^k(\sq{\rv{e}})}{
            k^{-\e/2}}
    \right]
    \geq \frac{1}{2}+\frac{\epsilon}{4},
\end{align}
completing the proof.
\end{proof}

\section{\texorpdfstring{%
$\NP$-membership of $\left(\frac{1}{2^{q-1}}-\epsilon\right)$-factor Approximation for \MMqCSPReconf
}{%
NP-membership of (1/2\textasciicircum(q-1)-ε)-factor Approximation for Maxmin q-CSP Reconfiguration
}}
\label{sec:NP}

In this section, we prove \cref{thm:intro:NP}; i.e.,
$\NP$-membership of a $\left(\frac{1}{2^{q-1}}-\epsilon\right)$-factor approximation for \MMqCSPReconf in the perfect completeness case.

\begin{theorem}
\label{thm:NP}
For any positive integers $q \geq 2$ and $\sigma$, and any real $\epsilon > 0$,
let $G=(V,E,\Sigma,\Psi)$ be a satisfiable $q$-CSP instance on $n$ vertices with $|\Sigma| = \sigma$, and
$f_\sss,f_\ttt \colon V \to \Sigma$ be a pair of its satisfying assignments.
If there exists a reconfiguration sequence from $f_\sss$ to $f_\ttt$ consisting only of satisfying assignments for $G$,
then there exists a reconfiguration sequence $\sq{f}$ from $f_\sss$ to $f_\ttt$ such that
\begin{align}
    \val_G(\sq{f}) \geq \frac{1}{2^{q-1}} - \epsilon,
\end{align}
and the length of $\sq{f}$ is at most $n^{\bigO_{q,\sigma,\epsilon}(1)}$.
In particular, for any positive integers $q \geq 2$ and $\sigma$, and any real $\epsilon > 0$,
\prb{Gap$_{1,\frac{1}{2^{q-1}}-\epsilon}$ $q$-CSP$_\sigma$ Reconfiguration} is in $\NP$.
\end{theorem}

Hereafter, we fix positive integers $q \geq 2$ and $\sigma$, and a real $\epsilon > 0$.
Let $G=(V,E,\Sigma,\Psi)$ be a satisfiable $q$-CSP instance with $|\Sigma|=\sigma$.
Let $n \defeq |V|$, $m \defeq |E|$, and
\begin{align}
\label{eq:NP:Delta}
    \Delta \defeq \frac{\epsilon^2 m}{100q \log n}.
\end{align}
We say that a vertex of $G$ is \defi{low degree} if its degree is at most $\Delta$ and \defi{high degree} otherwise.
The proof of \cref{thm:NP} is based on the following interpolation lemma for low-degree variables.

\begin{lemma}
\label{lem:NP:low-degree}
Let $f_\sss,f_\ttt \colon V \to \Sigma$ be a pair of satisfying assignments for $G$.
Define
\begin{align}
    D \defeq \bigl\{ x \in V \bigm| f_\sss(x) \neq f_\ttt(x) \bigr\}.
\end{align}
Suppose that every vertex $x \in D$ has degree at most $\Delta$.
Then, for any sufficiently large $n$,
there exists a reconfiguration sequence $\sq{f}$ from $f_\sss$ to $f_\ttt$ of length $|D|+1$ such that
\begin{align}
    \val_G(\sq{f}) \geq \frac{1}{2^{q-1}} - \epsilon.
\end{align}
\end{lemma}
\begin{proof}
Choose independent random labels $(\rv{\lambda}_x)_{x \in D}$,
each uniformly distributed on $(0,1)$.
For each real $\theta \in [0,1]$, define an assignment $\rv{h}_\theta \colon V \to \Sigma$ as
\begin{align}
    \rv{h}_\theta(x) \defeq
    \begin{cases}
        f_\ttt(x) & \text{if } \rv{\lambda}_x \leq \theta,\\
        f_\sss(x) & \text{otherwise.}
    \end{cases}
\end{align}
Note that $\rv{h}_0=f_\sss$ and $\rv{h}_1=f_\ttt$.
For each real $\theta \in [0,1]$, let
\begin{align}
    \rv{X}_\theta
    \defeq \bigl|\bigl\{e \in E \bigm| \rv{h}_\theta \text{ satisfies } e\bigr\}\bigr|
    = m \cdot \val_G(\rv{h}_\theta).
\end{align}
We show that there exists a realization of the random variables for which
$\rv{X}_\theta \geq \left(\frac{1}{2^{q-1}}-\epsilon\right)m$ holds for every $\theta \in [0,1]$.

\paragraph{(Step 1) Lower bound on $\E[\rv{X}_\theta]$ for fixed $\theta$.}
Fix a real $\theta \in [0,1]$ and a hyperedge $e \in E$.
Let $c_e \defeq |e \cap D|$.
If $c_e=0$, then $\rv{h}_\theta|_e = f_\sss|_e = f_\ttt|_e$; thus, $e$ is always satisfied.
Assume now that $c_e \geq 1$.
Since $f_\sss$ and $f_\ttt$ are satisfying assignments,
$\rv{h}_\theta$ certainly satisfies $e$ if each of the following two cases holds:
\begin{itemize}
    \item No vertex in $e \cap D$ has switched to its value in $f_\ttt$,
        which happens with probability $(1-\theta)^{c_e}$.
    \item Every vertex in $e \cap D$ has switched to its value in $f_\ttt$,
        which happens with probability $\theta^{c_e}$.
\end{itemize}
Since these two events are disjoint, we have
\begin{align}
    \Pr\bigl[\rv{h}_\theta \text{ satisfies } e\bigr]
    \geq (1-\theta)^{c_e} + \theta^{c_e}
    \underbrace{\geq}_{1 \leq c_e \leq q} \theta^q + (1-\theta)^q
    \geq \left(\frac{1}{2}\right)^q + \left(\frac{1}{2}\right)^q
    = \frac{1}{2^{q-1}},
\end{align}
where the second-to-last inequality holds because
the function $\phi_q(\theta) \defeq \theta^q + (1-\theta)^q$ attains its minimum at $\theta = \frac{1}{2}$.
Combining the cases $c_e=0$ and $c_e \geq 1$, we obtain
\begin{align}
\label{eq:NP:expectation}
    \Pr\bigl[\rv{h}_\theta \text{ satisfies } e\bigr] \geq \frac{1}{2^{q-1}},
    \qquad 
    \therefore \E\bigl[\rv{X}_\theta\bigr] \geq \frac{m}{2^{q-1}}.
\end{align}

\paragraph{(Step 2) Concentration on fixed $\theta$.}
For a fixed real $\theta \in [0,1]$,
regard $\rv{X}_\theta$ as a function $g_\theta \colon (0,1)^D \to \bbN$ of the random labels defined as
\begin{align}
    g_\theta\bigl((\rv{\lambda}_x)_{x \in D}\bigr) \defeq \rv{X}_\theta.
\end{align}
If only the coordinate $\rv{\lambda}_x$ is changed, then only the assignment to $x$ may change.
Therefore, $g_\theta$ satisfies the bounded difference property with bounds $(d(x))_{x \in D}$.
By applying McDiarmid's inequality to \cref{eq:NP:expectation}, we obtain
\begin{align}
\begin{aligned}
    \Pr\left[\rv{X}_\theta < \left(\frac{1}{2^{q-1}}-\frac{\epsilon}{2}\right)m\right]
    \leq \Pr\left[\rv{X}_\theta < \E[\rv{X}_\theta]-\frac{\epsilon m}{2}\right]
    \leq \exp\left(-\frac{2\left(\frac{\epsilon m}{2}\right)^2}{\sum_{x \in D} d(x)^2}\right).
\end{aligned}
\end{align}
Since $d(x) \leq \Delta$ for every vertex $x \in D$,
\begin{align}
    \sum_{x \in D} d(x)^2
    \leq \Delta\sum_{x \in D} d(x)
    \leq \Delta qm.
\end{align}
Hence, we derive
\begin{align}
    \Pr\left[\rv{X}_\theta < \left(\frac{1}{2^{q-1}}-\frac{\epsilon}{2}\right)m\right]
    \leq \exp\left(-\frac{\epsilon^2 m}{2\Delta q}\right)
    \underbrace{\leq}_{\text{\cref{eq:NP:Delta}}} n^{-50}.
    \label{eq:NP:fixed-s-tail}
\end{align}

\paragraph{(Step 3) One realization works for every $\theta \in [0,1]$.}
Let $M \defeq n^4$ and
\begin{align}
    \Gamma \defeq \left\{\frac{j}{M} \;\middle|\; 0 \leq j \leq M\right\}.
\end{align}
Let $\evt_1$ be the event that the value of $\rv{h}_\theta$ is at least 
$\frac{1}{2^{q-1}}-\frac{\epsilon}{2}$ for every $\theta \in \Gamma$; namely,
\begin{align}
    \evt_1 \defeq \left\{
        \rv{X}_\theta \geq \left(\frac{1}{2^{q-1}}-\frac{\epsilon}{2}\right)m \text{ for every } \theta \in \Gamma
    \right\}.
\end{align}
By applying the union bound to \cref{eq:NP:fixed-s-tail}, we derive
\begin{align}
\label{eq:NP:E1}
    \Pr\bigl[\neg\evt_1\bigr]
    \leq (M+1)n^{-50}
    = n^{-\Omega(1)}.
\end{align}
For each $j \in [M]$, define the interval $I_j$ as
\begin{align}
    I_j \defeq \left(\frac{j-1}{M},\frac{j}{M}\right].
\end{align}
Set
\begin{align}
    B \defeq \left\lfloor \frac{\epsilon m}{4\Delta} \right\rfloor
    \geq \frac{25q\log n}{\epsilon}-1.
\end{align}
Let $\evt_2$ be the event that every interval $I_j$ contains at most $B$ labels; namely,
\begin{align}
    \evt_2 \defeq \Bigl\{
        \bigl|\bigl\{ x \in D \mid \rv{\lambda}_x \in I_j \bigr\}\bigr| \leq B
        \text{ for every } j \in [M]
    \Bigr\}.
\end{align}
For each $j \in [M]$, we have
\begin{align}
\begin{aligned}
    \Pr\Bigl[
        \bigl|\bigl\{x \in D \mid \rv{\lambda}_x \in I_j\bigr\}\bigr| \geq B+1
    \Bigr]
    & = \Pr\left[
        \exists S \in \tbinom{S}{B+1} \text{ s.t.~}
        \forall x \in S, \; \rv{\lambda}_x \in I_j
    \right] \\
    & \leq \binom{|D|}{B+1}\left(\frac{1}{M}\right)^{B+1}
    \leq n^{B+1}\cdot n^{-4(B+1)}
    = n^{-3(B+1)},
\end{aligned}
\end{align}
which implies that
\begin{align}
\label{eq:NP:E2}
    \Pr\bigl[\neg\evt_2\bigr]
    \leq \sum_{j \in [M]} \Pr\Bigl[
        \bigl|\bigl\{x \in D \mid \rv{\lambda}_x \in I_j\bigr\}\bigr| \geq B+1
    \Bigr]
    \leq Mn^{-3(B+1)}
    = n^{-\Omega_{q,\epsilon}(\log n)}.
\end{align}
By \cref{eq:NP:E1,eq:NP:E2},
for any sufficiently large $n$, we have
\begin{align}
    \Pr\bigl[\evt_1 \wedge \evt_2 \bigr] > 0.
\end{align}
Fix a realization of the random labels $(\rv{\lambda}_x)_{x \in D}$ in this event.
We can safely assume that all labels are distinct.

We claim that
\begin{align}
    \rv{X}_\theta \geq \left(\frac{1}{2^{q-1}}-\epsilon\right)m
    \qquad\text{for every } \theta \in [0,1]. \label{eq:NP:all-s}
\end{align}
Fix a real $\theta \in [0,1]$, and let $\theta' \defeq \frac{\lfloor M\theta\rfloor}{M} \in \Gamma$
be the ``left endpoint'' of the interval containing $\theta$.
Since $\evt_1$ holds,
\begin{align}
    \rv{X}_{\theta'} \geq \left(\frac{1}{2^{q-1}}-\frac{\epsilon}{2}\right)m.
\end{align}
The assignments $\rv{h}_{\theta'}$ and $\rv{h}_\theta$ differ only on vertices $x \in D$ with $\rv{\lambda}_x \in (\theta',\theta]$.
Since $\evt_2$ holds, there are at most $B$ such vertices, denoted by $D_\theta \subseteq D$.
Therefore, we have
\begin{align}
\begin{aligned}
    & |\rv{X}_\theta - \rv{X}_{\theta'}|
    \leq \sum_{x \in D_\theta} d(x)
    \leq |D_\theta|\Delta
    \leq B\Delta
    \leq \frac{\epsilon m}{4}, \\
    & \implies \rv{X}_\theta
    \geq \rv{X}_{\theta'} - \frac{\epsilon m}{4}
    \geq \left(\frac{1}{2^{q-1}}-\frac{\epsilon}{2}\right)m - \frac{\epsilon m}{4}
    \geq \left(\frac{1}{2^{q-1}}-\epsilon\right)m,
\end{aligned}
\end{align}
as desired.

\paragraph{(Step 4) Extract a discrete reconfiguration sequence.}
Let $N \defeq |D|$.
Order the vertices of $D$ as $v_1,\ldots,v_N$ so that
\begin{align}
    \rv{\lambda}_{v_1} < \cdots < \rv{\lambda}_{v_N}.
\end{align}
For each $0 \leq t \leq N$, we define an assignment $\rv{f}^{(t)} \colon V \to \Sigma$ as
\begin{align}
    \rv{f}^{(t)}(x) \defeq
    \begin{cases}
        f_\ttt(x) & \text{if } x \in \{v_1,\ldots,v_t\}, \\
        f_\sss(x) & \text{otherwise.}
    \end{cases}
\end{align}
Then, $\sq{\rv{f}} \defeq (\rv{f}^{(0)},\rv{f}^{(1)},\ldots,\rv{f}^{(N)})$
is a reconfiguration sequence from $f_\sss$ to $f_\ttt$, whose length is $N+1$.
By applying \cref{eq:NP:all-s} to each assignment $\rv{f}^{(t)}$,
we obtain
\begin{align}
    \val_G(\sq{\rv{f}}) \geq \frac{1}{2^{q-1}} - \epsilon,
\end{align}
as desired.
\end{proof}

We are now ready to prove \cref{thm:NP}.

\begin{proof}[Proof of \cref{thm:NP}]
We can safely assume that $n$ is sufficiently large so that 
\cref{lem:NP:low-degree} holds.
Without loss of generality, we can assume that each vertex of $G$ appears in some hyperedge of $G$,
implying that $qm \geq n$.
Therefore, for sufficiently large $n$,
\begin{align}
    m \geq \frac{n}{q} > \frac{100q\log n}{\epsilon^2}.
\end{align}

Let $\sq{f} = (f^{(1)}, \ldots, f^{(T)})$
be a reconfiguration sequence from $f_\sss$ to $f_\ttt$ consisting only of satisfying assignments.
Partition the vertices of $G$ into
\begin{align}
\begin{aligned}
    H & \defeq \bigl\{v \in V \bigm| d(v) > \Delta \bigr\}, \\
    L & \defeq \bigl\{v \in V \bigm| d(v) \leq \Delta \bigr\}.
\end{aligned}
\end{align}
Since the sum of all vertex degrees is $qm$,
the number of vertices in $H$ is
\begin{align}
    |H|
    \leq \frac{qm}{\Delta}
    = \frac{100q^2\log n}{\epsilon^2}
    = \bigO_{q,\epsilon}(\log n).
    \label{eq:NP:H-size}
\end{align}
Consider the following sparsification procedure,
which extracts a subsequence of $\sq{f} = (f^{(1)}, \ldots, f^{(T)})$.

\begin{itembox}[l]{\textbf{Sparsifying $\sq{f} = (f^{(1)}, \ldots, f^{(T)})$}}
\begin{algorithmic}[1]
    \State let $I \defeq [T]$.
    \For{\textbf{each} index $t_\ell \in I$}
        \State find the largest index $t_r \in I$ such that $f^{(t_\ell)}|_H = f^{(t_r)}|_H$.
        \If{$t_\ell+1 \leq t_r-1$}
            \State delete $t_\ell+1, \ldots, t_r-1$ from $I$.
        \EndIf
    \EndFor
    \State \Return the subsequence $(f^{(t)})_{t \in I}$.
\end{algorithmic}
\end{itembox}

\noindent
Let $\sq{g} = (g^{(1)}, \ldots, g^{(T')})$
denote the subsequence of $\sq{f}$ obtained by the above procedure.
Note that $g^{(1)} = f^{(1)}$ and $g^{(T')} = f^{(T)}$.
We first bound the length $T'$ of $\sq{g}$.
Let $\alpha \colon H \to \Sigma$ be a partial assignment to $H$.
By construction, $\alpha$ appears in $g^{(1)}|_H, \ldots, g^{(T')}|_H$ at most twice.
Otherwise, there exist $t_1, t_2, t_3 \in I$ such that
$t_1 < t_2 < t_3$ and
$f^{(t_1)}|_H = f^{(t_2)}|_H = f^{(t_3)}|_H$; thus, 
$t_2$ would have been removed from $I$ by the sparsification procedure.
Moreover, the number of such partial assignments $\alpha \colon H \to \Sigma$
is $|\Sigma^H| = n^{\bigO_{q,\sigma,\epsilon}(1)}$ by \cref{eq:NP:H-size}.
Therefore, $T' \leq 2|\Sigma^H| = n^{\bigO_{q,\sigma,\epsilon}(1)}$.

We now construct a polynomial-length reconfiguration sequence.
For each $t \in [T'-1]$,
either $g^{(t)}|_H = g^{(t+1)}|_H$, or
$g^{(t)}$ and $g^{(t+1)}$ differ in at most a single vertex.
In the former case, 
by applying \cref{lem:NP:low-degree},
we obtain a reconfiguration sequence $\sq{h}_t$ from $g^{(t)}$ to $g^{(t+1)}$ with value at least $\frac{1}{2^{q-1}}-\epsilon$,
whose length is at most $n+1$.
In the latter case,
$\sq{h}_t \defeq (g^{(t)}, g^{(t+1)})$ is already a valid reconfiguration sequence.
Concatenating $\sq{h}_t$ for every $t \in [T'-1]$,
we obtain a reconfiguration sequence $\sq{h}$ from $g^{(1)}$ to $g^{(T')}$ 
whose value is at least $\frac{1}{2^{q-1}}-\epsilon$.
Moreover, the length of $\sq{h}$ is at most $T'(n+1) = n^{\bigO_{q,\sigma,\epsilon}(1)}$,
as desired.
\end{proof}

\section{\texorpdfstring{%
$\left(\frac{1}{2^{q-1}}-\epsilon\right)$-factor Approximation for \MMqCSPReconf on Regular Instances
}{%
(1/2\textasciicircum(q-1)-ε)-factor Approximation for Maxmin q-CSP Reconfiguration on Regular Instances
}}
\label{sec:regular}

In this section, we prove \cref{thm:intro:regular};
i.e., we develop a deterministic $\left(\frac{1}{2^{q-1}}-\epsilon\right)$-factor approximation algorithm
for \MMqCSPReconf on regular instances.
The proof is based on a bucket-label construction together with an explicit exact $2q$-wise independent family.

\begin{theorem}
\label{thm:regular}
For an integer $q \geq 2$ and a real $\epsilon > 0$,
let $G$ be a satisfiable $n$-vertex regular $q$-CSP instance and
$f_\sss$ and $f_\ttt$ be a pair of its satisfying assignments.
If $n$ is sufficiently large,
then there exists a polynomial-length reconfiguration sequence $\sq{f}$ from $f_\sss$ to $f_\ttt$
such that
\begin{align}
    \val_G(\sq{f}) \geq \frac{1}{2^{q-1}} - \epsilon.
\end{align}
Moreover, such $\sq{f}$ can be found in deterministic polynomial time.
In particular,
there exists a deterministic $\left(\frac{1}{2^{q-1}}-\epsilon\right)$-factor approximation algorithm for \MMqCSPReconf on regular $q$-CSP instances.
\end{theorem}

The idea is as follows.
If we choose a uniformly random ordering of the vertices and change them one by one from $f_\sss$ to $f_\ttt$,
then with positive probability every intermediate assignment satisfies at least about a $2^{1-q}$ fraction of the constraints.
For the deterministic algorithm, instead of using a fully random ordering,
we partition the vertices into $L=\Theta(q/\epsilon)$ random groups and then change the groups in order.
Thus we only need to analyze the assignments at the $L+1$ group boundaries and the sizes of the $L$ groups.
These quantities are controlled by first and second moments of random variables depending on at most $2q$ labels,
so an explicit exact $2q$-wise independent label family is enough for the analysis,
and we can deterministically search over that family.

The next lemma gives the exact $k$-wise independent label family that we use.
It follows immediately from \cite[Corollary~3.34]{vadhan2012pseudorandomness};
for the classical low-degree polynomial construction underlying this result, see also \citet{joffe1974set}.

\begin{lemma}[Exact $k$-wise independent labels {\cite[Corollary~3.34]{vadhan2012pseudorandomness}}]
\label{lem:regular:kwise}
Let $k \geq 1$, and let $L$ be a power of two.
Let $M$ be the smallest power of two such that $M \geq \max\{n,L\}$.
Then there exists an explicit family $\calH$ of functions $\ell \colon [n] \to [L]$ such that
\begin{enumerate}
    \item $|\calH| = M^k$, and
    \item if $\ell$ is chosen uniformly from $\calH$, then for every $t \leq k$ and every distinct vertices $v_1,\ldots,v_t \in [n]$, the tuple
    \begin{align}
        \bigl(\ell(v_1),\ldots,\ell(v_t)\bigr)
    \end{align}
    is uniformly distributed over $[L]^t$.
\end{enumerate}
\end{lemma}

\begin{proof}[Proof Sketch]
This follows immediately from \cite[Corollary~3.34]{vadhan2012pseudorandomness}.
Take the explicit family of $k$-wise independent functions from $[M]$ to $[M]$ given there;
it has size $M^k$.
Restrict the domain to $[n]$ and compose the output with any fixed balanced map from $[M]$ to $[L]$.
Since $L$ divides $M$, every label in $[L]$ has exactly $M/L$ preimages,
so the resulting labels are uniform on $[L]$, and every $t$-tuple of distinct labels remains uniform over $[L]^t$ for $t \leq k$.
\end{proof}

\begin{proof}[Proof of \cref{thm:regular}]
If $\epsilon \geq 2^{1-q}$, then the inequality
\begin{align}
    \val_G(\sq{f}) \geq \frac{1}{2^{q-1}} - \epsilon
\end{align}
is trivial for every reconfiguration sequence because $\val_G(\sq{f}) \geq 0$.
Thus we may assume that $0 < \epsilon < 2^{1-q}$.
If $\Delta = 0$, then every assignment satisfies all constraints and the theorem is immediate.
Hence we may further assume that $\Delta \geq 1$.

Let $(G,f_\sss,f_\ttt)$ be an instance of \MMqCSPReconf,
where $G$ is a satisfiable $\Delta$-regular $q$-CSP instance on $n$ vertices and $m$ constraints,
and $f_\sss$ and $f_\ttt$ are satisfying assignments for $G$.
We identify $V(G)$ with $[n]$.
Since $G$ is $q$-ary and $\Delta$-regular,
\begin{align}
    m = \frac{n\Delta}{q}.
\end{align}

Set $L$ to be the smallest power of two such that $L \geq 4q/\epsilon$.
Let $M$ be the smallest power of two such that $M \geq \max\{n,L\}$,
and apply \cref{lem:regular:kwise} with $k=2q$.
Thus we obtain an explicit family $\calH$ of functions $\ell \colon [n] \to [L]$ such that
\begin{align}
    |\calH| = M^{2q} \leq (2\max\{n,L\})^{2q} = n^{\bigO_{q,\epsilon}(1)},
\end{align}
and a uniformly random $\ell \in \calH$ is exact $2q$-wise independent and uniform on $[L]$.

Fix $\ell \in \calH$.
For each $i \in [L]$, define the $i$th bucket by
\begin{align}
    B_i(\ell) \defeq \{v \in [n] \mid \ell(v)=i\},
\end{align}
and for each $0 \leq j \leq L$, define
\begin{align}
    A_j(\ell) \defeq \bigcup_{1 \leq i \leq j} B_i(\ell),
\end{align}
where $A_0(\ell) \defeq \emptyset$.
Let $\prec_\ell$ be the total order on $[n]$ obtained by sorting vertices first by increasing label $\ell(v)$ and then by their natural order in $[n]$.
Let
\begin{align}
    \sq{f}_\ell = \bigl(f^{(0)}_\ell,\ldots,f^{(n)}_\ell\bigr)
\end{align}
be the reconfiguration sequence that changes vertices from $f_\sss$ to $f_\ttt$ in the order $\prec_\ell$.
For each $0 \leq j \leq L$, let
\begin{align}
    t_j(\ell) \defeq |A_j(\ell)|.
\end{align}
Then $f^{(t_j(\ell))}_\ell$ agrees with $f_\ttt$ on $A_j(\ell)$ and with $f_\sss$ on $[n]\setminus A_j(\ell)$.

For $0 \leq j \leq L$, define
\begin{align}
    Y_j(\ell)
    &\defeq \sum_{e \in E(G)} \Bigl\llbracket e \subseteq A_j(\ell) \text{ or } e \cap A_j(\ell)=\emptyset \Bigr\rrbracket,
\end{align}
and for $i \in [L]$, define
\begin{align}
    Z_i(\ell) \defeq |B_i(\ell)|.
\end{align}
For convenience, let $Z_0(\ell) \defeq 0$.
Thus $Y_j(\ell)$ counts the constraints all of whose vertices lie on one side of the $j$th checkpoint, and $Z_i(\ell)$ is the size of the $i$th bucket.
We call $\ell$ \defi{good} if, for every $0 \leq i \leq L$,
\begin{align}
    Y_i(\ell) \geq \left(\frac{1}{2^{q-1}}-\frac{\epsilon}{2}\right)m
    \qquad\text{and}\qquad
    Z_i(\ell) \leq \frac{2n}{L}.
\end{align}
We show that a uniformly random $\ell \in \calH$ is good with positive probability.

\paragraph{Checkpoint bounds.}
Fix $0 \leq j \leq L$.
For each constraint $e \in E(G)$, define
\begin{align}
    I_{e,j}(\ell) \defeq \Bigl\llbracket e \subseteq A_j(\ell) \text{ or } e \cap A_j(\ell)=\emptyset \Bigr\rrbracket.
\end{align}
Then $Y_j = \sum_{e \in E(G)} I_{e,j}$.
Because $\ell$ is exact $2q$-wise independent and uniform on $[L]$,
for every fixed constraint $e$ the labels on the $q$ vertices of $e$ are fully independent and uniform.
Hence
\begin{align}
    \E[I_{e,j}] = \left(\frac{j}{L}\right)^q + \left(1-\frac{j}{L}\right)^q,
\end{align}
and therefore
\begin{align}
    \E[Y_j]
    = m\left[\left(\frac{j}{L}\right)^q + \left(1-\frac{j}{L}\right)^q\right]
    \geq \frac{m}{2^{q-1}},
\end{align}
because the function $p \mapsto p^q + (1-p)^q$ is convex on $[0,1]$ and is minimized at $p=1/2$.

We next bound the variance of $Y_j$.
If two distinct constraints $e,e' \in E(G)$ are disjoint, then $I_{e,j}$ and $I_{e',j}$ are independent,
because they depend on the labels on at most $2q$ vertices and $\ell$ is exact $2q$-wise independent.
Thus
\begin{align}
\begin{aligned}
    \operatorname{Var}(Y_j)
    &= \sum_{e \in E(G)} \operatorname{Var}(I_{e,j})
       + \sum_{\substack{e,e' \in E(G) \\ e \neq e'}} \operatorname{Cov}(I_{e,j},I_{e',j}) \\
    &\leq \sum_{e \in E(G)} \operatorname{Var}(I_{e,j})
       + \sum_{\substack{e,e' \in E(G) \\ e \neq e',\ e \cap e' \neq \emptyset}}
            \bigl|\operatorname{Cov}(I_{e,j},I_{e',j})\bigr| \\
    &\leq m + mq(\Delta-1).
\end{aligned}
\end{align}
Indeed, $\operatorname{Var}(I_{e,j}) \leq 1$ for every $e$, and the number of ordered pairs $(e,e')$ of distinct intersecting constraints is at most $mq(\Delta-1)$:
for each constraint $e$ and each of its $q$ vertices, there are at most $\Delta-1$ other constraints containing that vertex.
Using $m=n\Delta/q$ and $\Delta \geq 1$, we obtain
\begin{align}
    mq(\Delta-1) \leq mq\Delta = \frac{q^2m^2}{n}
    \qquad\text{and}\qquad
    m \geq \frac{n}{q}.
\end{align}
Hence Chebyshev's inequality gives
\begin{align}
\begin{aligned}
    \Pr_{\ell \in \calH}\left[
        Y_j(\ell) < \left(\frac{1}{2^{q-1}}-\frac{\epsilon}{2}\right)m
    \right]
    &\leq \Pr_{\ell \in \calH}\left[
        Y_j(\ell) < \E[Y_j] - \frac{\epsilon}{2}m
    \right] \\
    &\leq \frac{4\operatorname{Var}(Y_j)}{\epsilon^2m^2} \\
    &\leq \frac{4}{\epsilon^2m} + \frac{4q^2}{\epsilon^2n} \\
    &\leq \frac{4q(1+q)}{\epsilon^2n}.
\end{aligned}
\end{align}

\paragraph{Bucket-size bounds.}
Fix $i \in [L]$.
For each vertex $v \in [n]$, let
\begin{align}
    J_{v,i}(\ell) \defeq \Bigl\llbracket \ell(v)=i \Bigr\rrbracket.
\end{align}
Then $Z_i = \sum_{v \in [n]} J_{v,i}$.
Since each label is uniform on $[L]$ and any two vertex labels are independent,
\begin{align}
    \E[Z_i] = \frac{n}{L}
    \qquad\text{and}\qquad
    \operatorname{Var}(Z_i)
    = \sum_{v \in [n]} \operatorname{Var}(J_{v,i})
    \leq \frac{n}{L}.
\end{align}
Another application of Chebyshev's inequality yields
\begin{align}
    \Pr_{\ell \in \calH}\left[ Z_i(\ell) > \frac{2n}{L} \right]
    \leq \Pr_{\ell \in \calH}\left[ \left|Z_i(\ell)-\frac{n}{L}\right| \geq \frac{n}{L} \right]
    \leq \frac{L}{n}.
\end{align}

By the union bound,
\begin{align}
    \Pr_{\ell \in \calH}\Bigl[
        \exists i \in \{0,1,\ldots,L\} :
        \Bigl(
            Y_i(\ell) < \left(\frac{1}{2^{q-1}}-\frac{\epsilon}{2}\right)m
            \text{ or }
            Z_i(\ell) > \frac{2n}{L}
        \Bigr)
    \Bigr]
    \leq \frac{4(L+1)q(1+q)}{\epsilon^2n} + \frac{L^2}{n}.
\end{align}
Since $L = \bigO_{q,\epsilon}(1)$, this upper bound is smaller than $1$ for all sufficiently large $n$.
Therefore, for all sufficiently large $n$, there exists a good labeling $\ell^\star \in \calH$.
Because $\calH$ is explicit and has polynomial size, we can find such an $\ell^\star$ deterministically by enumerating all labelings in $\calH$ and checking the above conditions.
For a fixed labeling $\ell$, all values $Z_i(\ell)$ and $Y_j(\ell)$ can be computed in $\bigO\bigl(L(m+n)\bigr)$ time, so this search runs in polynomial time.

It remains to show that the reconfiguration sequence $\sq{f}_{\ell^\star}$ has value at least $\frac{1}{2^{q-1}}-\epsilon$.
Write $t_j \defeq t_j(\ell^\star)$ for $0 \leq j \leq L$.
Fix $j \in \{0,1,\ldots,L-1\}$.
At time $t_j$, the assignment $f^{(t_j)}_{\ell^\star}$ agrees with $f_\ttt$ on $A_j(\ell^\star)$ and with $f_\sss$ on $[n]\setminus A_j(\ell^\star)$.
Hence every constraint counted by $Y_j(\ell^\star)$ is certainly satisfied by $f^{(t_j)}_{\ell^\star}$,
because all its vertices lie entirely on one side of the cut and both $f_\sss$ and $f_\ttt$ satisfy that constraint.
Now let $t$ be any time with $t_j \leq t \leq t_{j+1}$.
The assignments $f^{(t)}_{\ell^\star}$ and $f^{(t_j)}_{\ell^\star}$ differ only on vertices of the bucket $B_{j+1}(\ell^\star)$.
Thus only constraints incident to $B_{j+1}(\ell^\star)$ can change their satisfaction status, and there are at most $Z_{j+1}(\ell^\star)\Delta$ such constraints.
Consequently,
\begin{align}
\begin{aligned}
    \#\{\text{constraints satisfied by } f^{(t)}_{\ell^\star}\}
    &\geq Y_j(\ell^\star) - Z_{j+1}(\ell^\star)\Delta \\
    &\geq \left(\frac{1}{2^{q-1}}-\frac{\epsilon}{2}\right)m - \frac{2n\Delta}{L} \\
    &= \left(\frac{1}{2^{q-1}}-\frac{\epsilon}{2}\right)m - \frac{2q}{L}m \\
    &\geq \left(\frac{1}{2^{q-1}}-\epsilon\right)m,
\end{aligned}
\end{align}
where the last inequality uses $L \geq 4q/\epsilon$.
Since this holds for every block and every intermediate time,
\begin{align}
    \val_G(\sq{f}_{\ell^\star}) \geq \frac{1}{2^{q-1}} - \epsilon.
\end{align}
This proves the theorem for all sufficiently large $n$.

Finally, fix $q$ and $\epsilon$, and let $n_0=n_0(q,\epsilon)$ be large enough so that the argument above works for every $n \geq n_0$.
For instances with $n < n_0$, we solve the problem exactly.
Indeed, the reconfiguration graph on all assignments in $\Sigma^V$ has
$|\Sigma|^n = |\Sigma|^{\bigO_{q,\epsilon}(1)}$ vertices,
which is polynomial in the input size because $n$ is bounded by the constant $n_0$.
We can therefore construct this graph explicitly and compute, by a standard bottleneck-path algorithm, a reconfiguration sequence from $f_\sss$ to $f_\ttt$ of maximum value.
Combining this exact algorithm for $n<n_0$ with the algorithm above for $n\geq n_0$, we obtain a deterministic polynomial-time $\left(\frac{1}{2^{q-1}}-\epsilon\right)$-factor approximation algorithm for \MMqCSPReconf on regular $q$-CSP instances.
\end{proof}

\section*{Acknowledgments}
The authors used ChatGPT 5.5 Pro in the course of developing proofs of \cref{thm:NP,thm:regular}.
Before using ChatGPT, the authors proved \cref{lem:NP:low-degree} separately.
The authors then used ChatGPT to generate initial proof drafts for \cref{thm:NP,thm:regular}.
The drafts were subsequently rewritten and refined by the authors.
The authors verified the correctness and originality of all content.

\begin{sloppypar}
\printbibliography
\end{sloppypar}

\appendix
\section{\texorpdfstring{%
Proof of \cref{thm:IKW12}
}{%
Proof of Theorem \ref{thm:IKW12}
}}
\label{app:IKW12}

In this section, we prove \cref{thm:IKW12}.
It suffices to consider the case of $q = 2$.

\begin{theorem}
\label{thm:app:IKW12}
Let $F \colon \binom{V}{k} \times \Rnd \to \Sigma^k$ be a $k$-set function and
$n \defeq |V|$.
Let $\ell$ and $m$ be positive integers with $1 \leq m \leq \ell \leq k$.
If $k$ and $n$ are sufficiently large, the following holds.
Suppose that $\T_2(\ell,m)$ accepts $F$
with probability $p = \omega\left(\frac{\ell}{k}\right)$.
Then, there exists a function $f \colon V \to \Sigma$ such that
\begin{align}
    \Pr_{(\rv{X}, \rv{\rnd}) \in \binom{V}{k} \times \Rnd}\left[
        \close{F(\rv{X};\rv{\rnd})}{f^k(\rv{X})}{\delta_{k,\ell,m}}
    \right] \geq \Omega\bigl(p^6\bigr),
\end{align}
where $\delta_{k,\ell,m}$ is defined as
\begin{align}
    \delta_{k,\ell,m}
    = \tilde{\Theta}\left(\max\left\{\frac{\ell}{k},\frac{m}{\ell}\right\}\right).
\end{align}
\end{theorem}

Hereafter, fix a function $F \colon \binom{V}{k} \times \Rnd \to \Sigma^k$ and
positive integers $\ell$ and $m$ with $1 \leq m \leq \ell \leq k$.
Define $n \defeq |V|$ and $\sigma \defeq |\Sigma|$.
For $v \in V$ and $S \subseteq V$, define
\begin{align}
\begin{aligned}
    \tbinom{V}{k}_v & \defeq \left\{ X \in \tbinom{V}{k} \;\middle|\; v \in X \right\}, &
    \tbinom{V}{k}_S & \defeq \left\{ X \in \tbinom{V}{k} \;\middle|\; S \subseteq X \right\}.
\end{aligned}
\end{align}

\subsection{Goodness and Excellence}

We introduce goodness and excellence from \cite{impagliazzo2012new}.

\begin{definition}
Let $I_0 \in \binom{V}{\ell}$,
$A_0, A \in \binom{V \setminus I_0}{k-\ell}$, and
$\rnd_0, \rnd \in \Rnd$.
We say that
$(A,\rnd)$ is \defi{consistent} with $(I_0, A_0, \rnd_0)$ if
$\close{F(I_0 \cup A; \rnd)|_{I_0}}{F(I_0 \cup A_0; \rnd_0)|_{I_0}}{m/\ell}$.
We define $\Cons(I_0, A_0, \rnd_0)$ as
\begin{align}
\begin{aligned}
    \Cons(I_0, A_0, \rnd_0)
    & \defeq \left\{
        (A, \rnd) \in \tbinom{V \setminus I_0}{k-\ell} \times \Rnd
        \Bigm|
        (A, \rnd) \text{ is consistent with } (I_0,A_0,\rnd_0)
    \right\} \\
    & = \left\{
        (A, \rnd) \in \tbinom{V \setminus I_0}{k-\ell} \times \Rnd
        \Bigm|
        \close{F(I_0 \cup A; \rnd)|_{I_0}}{F(I_0 \cup A_0; \rnd_0)|_{I_0}}{m/\ell}
    \right\}.
\end{aligned}
\end{align}
\end{definition}

\begin{definition}
Let
$I_0 \in \binom{V}{\ell}$,
$A_0 \in \binom{V \setminus I_0}{k-\ell}$, and
$\rnd_0 \in \Rnd$.
We say that
$(I_0, A_0, \rnd_0)$ is \defi{$\epsilon$-good} if
\begin{align}
    \Pr_{(\rv{A}, \rv{\rnd}) \in \binom{V \setminus I_0}{k-\ell} \times \Rnd}\bigl[
        (\rv{A}, \rv{\rnd}) \in \Cons(I_0, A_0, \rnd_0)
    \bigr] \geq \epsilon.
\end{align}
\end{definition}

\begin{lemma}[Lemma 3.5 of \cite{impagliazzo2012new}]
\label{lem:IKW12:lem35}
Suppose that $\T_2(\ell,m)$ accepts $F$ with probability $\epsilon$.
Then, 
\begin{align}
    \Pr_{\substack{
        \rv{I}_0 \in \binom{V}{\ell} \\
        (\rv{A}_0, \rv{\rnd}_0) \in \binom{V\setminus\rv{I}_0}{k-\ell} \times \Rnd
    }}\Bigl[
        (\rv{I}_0, \rv{A}_0, \rv{\rnd}_0) \text{ is $\tfrac{\epsilon}{2}$-good}
    \Bigr] \geq \tfrac{\epsilon}{2}.
\end{align}
\end{lemma}
\begin{proof}
By an averaging argument, we derive
\begin{align}
\begin{aligned}
    & \Pr_{\substack{
        \rv{I}_0 \in \binom{V}{\ell} \\
        (\rv{A}_0, \rv{\rnd}_0), (\rv{A}_1, \rv{\rnd}_1) \in \binom{V\setminus\rv{I}_0}{k-\ell} \times \Rnd
    }}\left[
        \close{F(\rv{I}_0 \cup \rv{A}_0; \rv{\rnd}_0)|_{\rv{I}_0}}{F(\rv{I}_0 \cup \rv{A}_1; \rv{\rnd}_1)|_{\rv{I}_0}}{m/\ell}
    \right] \geq \epsilon, \\
    \implies & \Pr_{\substack{
        \rv{I}_0 \in \binom{V}{\ell} \\
        (\rv{A}_0, \rv{\rnd}_0) \in \binom{V\setminus\rv{I}_0}{k-\ell} \times \Rnd \\
    }}\left[
        \Pr_{(\rv{A}_1, \rv{\rnd}_1) \in \binom{V\setminus\rv{I}_0}{k-\ell} \times \Rnd}\left[
            \close{F(\rv{I}_0 \cup \rv{A}_0; \rv{\rnd}_0)|_{\rv{I}_0}}{F(\rv{I}_0 \cup \rv{A}_1; \rv{\rnd}_1)|_{\rv{I}_0}}{m/\ell}
        \right] \geq \tfrac{\epsilon}{2}
    \right] \geq \tfrac{\epsilon}{2}, \\
    \implies & \Pr_{\substack{
        \rv{I}_0 \in \binom{V}{\ell} \\
        (\rv{A}_0, \rv{\rnd}_0) \in \binom{V\setminus\rv{I}_0}{k-\ell} \times \Rnd
    }}\Bigl[
        (\rv{I}_0, \rv{A}_0, \rv{\rnd}_0) \text{ is $\tfrac{\epsilon}{2}$-good}
    \Bigr] \geq \tfrac{\epsilon}{2},
\end{aligned}
\end{align}
as desired.
\end{proof}

\begin{definition}
Let
$I_0 \in \binom{V}{\ell}$,
$A_0 \in \binom{V \setminus I_0}{k-\ell}$, and
$\rnd_0 \in \Rnd$.
We say that
$(I_0, A_0, \rnd_0)$ is \defi{$(\alpha, \gamma, \epsilon)$-excellent} if
$(I_0, A_0, \rnd_0)$ is $\epsilon$-good and
\begin{align}
    \Pr_{\substack{
        \rv{B} \in \binom{V\setminus I_0}{\ell} \\
        (\rv{A}_1, \rv{\rnd}_1) \in \binom{V\setminus I_0}{k-\ell}_{\rv{B}} \times \Rnd \\
        (\rv{A}_2, \rv{\rnd}_2) \in \binom{V\setminus I_0}{k-\ell}_{\rv{B}} \times \Rnd
    }}\left[
        \begin{aligned}
            (\rv{A}_1, \rv{\rnd}_1) \in \Cons(I_0, A_0, \rnd_0) \\
            (\rv{A}_2, \rv{\rnd}_2) \in \Cons(I_0, A_0, \rnd_0)
        \end{aligned}
        \text{ and }
        \far{F(I_0 \cup \rv{A}_1; \rv{\rnd}_1)|_{\rv{B}}}{F(I_0 \cup \rv{A}_2; \rv{\rnd}_2)|_{\rv{B}}}{\alpha}
    \right] \leq \gamma.
\end{align}
\end{definition}

\begin{lemma}[Lemma 3.6 of \cite{impagliazzo2012new}]
\label{lem:IKW12:lem36}
    Let $\alpha \geq \omega\bigl(\frac{m}{\ell}\bigr)$.
    Then,
    \begin{align}
        \Pr_{\substack{
            \rv{I}_0 \in \binom{V}{\ell} \\
            (\rv{A}_0, \rv{\rnd}_0) \in \binom{V\setminus\rv{I}_0}{k-\ell} \times \Rnd
        }}\Bigl[
            (\rv{I}_0, \rv{A}_0, \rv{\rnd}_0)
            \text{ is $\tfrac{\epsilon}{2}$-good but not $\left(\alpha, \gamma, \tfrac{\epsilon}{2}\right)$-excellent}
        \Bigr] < \frac{\exp\bigl(-\Omega(\alpha \ell)\bigr)}{\gamma}.
    \end{align}
\end{lemma}
\begin{proof}
We will bound the following probability:
\begin{align}
\label{eq:IKW12:lem36}
\begin{aligned}
    \Pr_{\substack{
            \rv{I}_0 \in \binom{V}{\ell} \\
            (\rv{A}_0, \rv{\rnd}_0) \in \binom{V\setminus\rv{I}_0}{k-\ell} \times \Rnd \\
            \rv{B} \in   \binom{V\setminus\rv{I}_0}{\ell} \\
            (\rv{A}_1, \rv{\rnd}_1) \in \binom{V\setminus\rv{I}_0}{k-\ell}_{\rv{B}} \times \Rnd \\
            (\rv{A}_2, \rv{\rnd}_2) \in \binom{V\setminus\rv{I}_0}{k-\ell}_{\rv{B}} \times \Rnd
        }}\left[
        \begin{aligned}
            (\rv{A}_1, \rv{\rnd}_1) \in \Cons(\rv{I}_0, \rv{A}_0, \rv{\rnd}_0) \\
            (\rv{A}_2, \rv{\rnd}_2) \in \Cons(\rv{I}_0, \rv{A}_0, \rv{\rnd}_0)
        \end{aligned}
        \text{ and }
        \far{F(\rv{I}_0 \cup \rv{A}_1; \rv{\rnd}_1)|_{\rv{B}}}{F(\rv{I}_0 \cup \rv{A}_2; \rv{\rnd}_2)|_{\rv{B}}}{\alpha}
    \right].
\end{aligned}
\end{align}
Observing that
\begin{align}
& \begin{aligned}
    & \bigl(
    (\rv{A}_1, \rv{\rnd}_1) \in \Cons(\rv{I}_0, \rv{A}_0, \rv{\rnd}_0) \text{ and }
    (\rv{A}_2, \rv{\rnd}_2) \in \Cons(\rv{I}_0, \rv{A}_0, \rv{\rnd}_0)
    \bigr) \\
    & \implies
    \close{F(\rv{I}_0 \cup \rv{A}_1; \rv{\rnd}_1)|_{\rv{I}_0}}{F(\rv{I}_0 \cup \rv{A}_2; \rv{\rnd}_2)|_{\rv{I}_0}}{2m/\ell},
\end{aligned} \\
& \begin{aligned}
    & \far{F(\rv{I}_0 \cup \rv{A}_1; \rv{\rnd}_1)|_{\rv{B}}}{F(\rv{I}_0 \cup \rv{A}_2; \rv{\rnd}_2)|_{\rv{B}}}{\alpha} \\
    & \implies
    \far{F(\rv{I}_0 \cup \rv{A}_1; \rv{\rnd}_1)|_{\rv{I}_0 \cup \rv{B}}}{F(\rv{I}_0 \cup \rv{A}_2; \rv{\rnd}_2)|_{\rv{I}_0 \cup \rv{B}}}{\alpha/2},
\end{aligned}
\end{align}
we derive that \cref{eq:IKW12:lem36} is at least
\begin{align}
\begin{aligned}
    \Pr_{\rv{I}_0, \rv{A}_0, \rv{\rnd}_0, \rv{B}, \rv{A}_1, \rv{\rnd}_1, \rv{A}_2, \rv{\rnd}_2}\left[
        \begin{aligned}
            & \close{F(\rv{I}_0 \cup \rv{A}_1; \rv{\rnd}_1)|_{\rv{I}_0}}{F(\rv{I}_0 \cup \rv{A}_2; \rv{\rnd}_2)|_{\rv{I}_0}}{2m/\ell} \\
            & \far{F(\rv{I}_0 \cup \rv{A}_1; \rv{\rnd}_1)|_{\rv{I}_0 \cup \rv{B}}}{F(\rv{I}_0 \cup \rv{A}_2; \rv{\rnd}_2)|_{\rv{I}_0 \cup \rv{B}}}{\alpha/2}
        \end{aligned}
    \right].
\end{aligned}
\end{align}
Conditioning $\rv{I}_0 \cup \rv{B}$ on the event that 
$\far{F(\rv{I}_0 \cup \rv{A}_1; \rv{\rnd}_1)|_{\rv{I}_0 \cup \rv{B}}}{F(\rv{I}_0 \cup \rv{A}_2; \rv{\rnd}_2)|_{\rv{I}_0 \cup \rv{B}}}{\alpha/2}$, let
\begin{align}
    \rv{Z}
    \defeq \ell \cdot \dist\bigl(F(\rv{I}_0 \cup \rv{A}_1; \rv{\rnd}_1)|_{\rv{I}_0}, F(\rv{I}_0 \cup \rv{A}_2; \rv{\rnd}_2)|_{\rv{I}_0}\bigr)
    = \sum_{i \in \rv{I}_0} \bigl\llbracket
        F(\rv{I}_0 \cup \rv{A}_1; \rv{\rnd}_1)|_i \neq F(\rv{I}_0 \cup \rv{A}_2; \rv{\rnd}_2)|_i
    \bigr\rrbracket.
\end{align}
Since
$\rv{I}_0$ is uniformly distributed over $\binom{\rv{I}_0\cup\rv{B}}{\ell}$,
$\rv{Z}$ follows a hypergeometric distribution where
the population size is $2\ell$,
the number of successes is
    $2\ell \cdot \dist\bigl(F(\rv{I}_0 \cup \rv{A}_1; \rv{\rnd}_1)|_{\rv{I}_0 \cup \rv{B}},F(\rv{I}_0 \cup \rv{A}_2; \rv{\rnd}_2)|_{\rv{I}_0 \cup \rv{B}}\bigr)$, and
the sample size is $\ell$.
Since $\alpha \geq \omega\bigl(\frac{m}{\ell}\bigr)$ by assumption and
$\E[\rv{Z}] \geq \frac{\alpha}{2}\ell$,
we have
\begin{align}
    \frac{2m}{\ell}\ell
    \leq \frac{\alpha}{200}\ell
    \leq \frac{\E[\rv{Z}]}{100}.
\end{align}
By the Chernoff bound, we have
\begin{align}
\begin{aligned}
    & \Pr_{\rv{I}_0, \rv{A}_0, \rv{\rnd}_0, \rv{B}, \rv{A}_1, \rv{\rnd}_1, \rv{A}_2, \rv{\rnd}_2}\left[
        \begin{aligned}
        & \close{F(\rv{I}_0 \cup \rv{A}_1; \rv{\rnd}_1)|_{\rv{I}_0}}{F(\rv{I}_0 \cup \rv{A}_2; \rv{\rnd}_2)|_{\rv{I}_0}}{2m/\ell} \\
        & \far{F(\rv{I}_0 \cup \rv{A}_1; \rv{\rnd}_1)|_{\rv{I}_0 \cup \rv{B}}}{F(\rv{I}_0 \cup \rv{A}_2; \rv{\rnd}_2)|_{\rv{I}_0 \cup \rv{B}}}{\alpha/2}
        \end{aligned}
    \right] \\
    & \leq \Pr_{\rv{I}_0, \rv{A}_0, \rv{\rnd}_0, \rv{B}, \rv{A}_1, \rv{\rnd}_1, \rv{A}_2, \rv{\rnd}_2}\left[
        \close{F(\rv{I}_0 \cup \rv{A}_1; \rv{\rnd}_1)|_{\rv{I}_0}}{F(\rv{I}_0 \cup \rv{A}_2; \rv{\rnd}_2)|_{\rv{I}_0}}{2m/\ell}
        \Bigm|
        \far{F(\rv{I}_0 \cup \rv{A}_1; \rv{\rnd}_1)|_{\rv{I}_0 \cup \rv{B}}}{F(\rv{I}_0 \cup \rv{A}_2; \rv{\rnd}_2)|_{\rv{I}_0 \cup \rv{B}}}{\alpha/2}
    \right] \\
    & \leq \Pr\Bigl[\rv{Z} \leq \tfrac{2m}{\ell}\ell\Bigr] \\
    & \leq \Pr\Bigl[\rv{Z} \leq \tfrac{1}{100}\E[\rv{Z}]\Bigr] \\
    & < \exp\bigl(-\Omega(\alpha \ell)\bigr).
\end{aligned}
\end{align}
By an averaging argument, we have
\begin{align}
\begin{aligned}
        & \Pr_{\rv{I}_0, \rv{A}_0, \rv{\rnd}_0}\left[
        \Pr_{\rv{B}, \rv{A}_1, \rv{\rnd}_1, \rv{A}_2, \rv{\rnd}_2}\left[
            \begin{aligned}
                (\rv{A}_1, \rv{\rnd}_1) \in \Cons(\rv{I}_0, \rv{A}_0, \rv{\rnd}_0) \\
                (\rv{A}_2, \rv{\rnd}_2) \in \Cons(\rv{I}_0, \rv{A}_0, \rv{\rnd}_0)
            \end{aligned}
            \text{ and }
            \far{F(\rv{I}_0 \cup \rv{A}_1; \rv{\rnd}_1)|_{\rv{B}}}{F(\rv{I}_0 \cup \rv{A}_2; \rv{\rnd}_2)|_{\rv{B}}}{\alpha}
        \right] \geq \gamma
    \right]
    < \frac{\exp\bigl(-\Omega(\alpha \ell)\bigr)}{\gamma}, \\
    & \implies
    \Pr_{\rv{I}_0, \rv{A}_0, \rv{\rnd}_0}\Bigl[
        (\rv{I}_0, \rv{A}_0, \rv{\rnd}_0)
        \text{ is $\tfrac{\epsilon}{2}$-good but not $\left(\alpha,\gamma,\tfrac{\epsilon}{2}\right)$-excellent}
    \Bigr]
    < \frac{\exp\bigl(-\Omega(\alpha \ell)\bigr)}{\gamma}.
\end{aligned}
\end{align}
as desired.
\end{proof}

As an immediate corollary of \cref{lem:IKW12:lem35,lem:IKW12:lem36}, we obtain the following.
\begin{corollary}[Corollary 3.7 of \cite{impagliazzo2012new}]
\label{lem:IKW12:cor37}
Suppose that $\T_2(\ell,m)$ accepts $f$ with probability $\epsilon$.
Then, 
\begin{align}
    \Pr_{\substack{
        \rv{I}_0 \in \binom{V}{\ell} \\
        (\rv{A}_0, \rv{\rnd}_0) \in \binom{V\setminus\rv{I}_0}{k-\ell} \times \Rnd
    }}\Bigl[
        (\rv{I}_0, \rv{A}_0, \rv{\rnd}_0) \text{ is $\left(\alpha, \gamma, \tfrac{\epsilon}{2}\right)$-excellent}
        \Bigm|
        (\rv{I}_0, \rv{A}_0, \rv{\rnd}_0) \text{ is $\tfrac{\epsilon}{2}$-good}
    \Bigr] \geq 1 - \frac{\exp\bigl(-\Omega(\alpha \ell)\bigr)}{\gamma \epsilon}.
\end{align}
\end{corollary}
\begin{proof}
By \cref{lem:IKW12:lem35,lem:IKW12:lem36}, we have
\begin{align}
\begin{aligned}
    & \Pr_{\substack{
        \rv{I}_0 \in \binom{V}{\ell} \\
        (\rv{A}_0, \rv{\rnd}_0) \in \binom{V\setminus\rv{I}_0}{k-\ell} \times \Rnd
    }}\Bigl[
        (\rv{I}_0, \rv{A}_0, \rv{\rnd}_0) \text{ is not $\left(\alpha, \gamma, \tfrac{\epsilon}{2}\right)$-excellent}
        \Bigm|
        (\rv{I}_0, \rv{A}_0, \rv{\rnd}_0) \text{ is $\tfrac{\epsilon}{2}$-good}
    \Bigr] \\
    & = \frac{\displaystyle
            \Pr_{\rv{I}_0, \rv{A}_0, \rv{\rnd}_0}\Bigl[
                (\rv{I}_0, \rv{A}_0, \rv{\rnd}_0) \text{ is $\tfrac{\epsilon}{2}$-good but not $\left(\alpha, \gamma, \tfrac{\epsilon}{2}\right)$-excellent}
            \Bigr]
        }{\displaystyle
            \Pr_{\rv{I}_0, \rv{A}_0, \rv{\rnd}_0}\Bigl[
                (\rv{I}_0, \rv{A}_0, \rv{\rnd}_0) \text{ is $\tfrac{\epsilon}{2}$-good}
            \Bigr]
        } \\
    & < \frac{\exp\bigl(-\Omega(\alpha \ell)\bigr)}{\gamma} \cdot \frac{2}{\epsilon}
    = \frac{\exp\bigl(-\Omega(\alpha \ell)\bigr)}{\gamma \epsilon},
\end{aligned}
\end{align}
as desired.
\end{proof}

\subsection{Excellence Implies Local Agreement}

We show that we can perform unique decoding on $\Cons(I_0, A_0, \rnd_0)$.
We first define the plurality function.

\begin{definition}
\label{lem:IKW12:plurality}
Let
$I_0 \in \binom{V}{\ell}$,
$A_0 \in \binom{V \setminus I_0}{k-\ell}$, and
$\rnd_0 \in \Rnd$.
The \defi{plurality function} of $(I_0,A_0,\rnd_0)$ is defined as
a function $f_{I_0,A_0,\rnd_0} \colon V \to \Sigma$
such that for each $v \in V$,
\begin{align}
    f_{I_0,A_0,\rnd_0}(v) \defeq \PLR_{(A,\rnd) \in \Cons(I_0, A_0, \rnd_0): A \ni v}\bigl\{
        F(I_0 \cup A; \rnd)|_v
    \bigr\}.
\end{align}
\end{definition}

\begin{lemma}[Lemma 3.8 of \cite{impagliazzo2012new}]
\label{lem:IKW12:lem38}
Suppose that $\T_2(\ell,m)$ accepts $F$ with probability
$\epsilon \geq \omega\bigl(\frac{\ell}{k}\bigr)$ and
that $(I_0,A_0,\rnd_0)$ is
$\left(\alpha,\gamma,\tfrac{\epsilon}{2}\right)$-excellent, where
\begin{align}
    \alpha = \omega\left(\max\left\{\frac{\ell}{k}, \frac{m}{\ell}\right\} \log\epsilon^{-1}\right)
    \text{ and }
    \gamma = o\bigl(\epsilon^3\bigr).
\end{align}
Then,
\begin{align}
\begin{aligned}
    & \Pr_{(\rv{A}, \rv{\rnd}) \in \Cons(I_0, A_0, \rnd_0)}\left[
        \far{F(I_0 \cup \rv{A}; \rv{\rnd})}{f_{I_0,A_0,\rnd_0}^k(I_0 \cup \rv{A})}{2\beta}
    \right] < \delta, \text{ where }
    \beta \defeq 16\sigma\alpha \text{ and }
    \delta \defeq \frac{256\sigma\gamma}{\epsilon^2}.
\end{aligned}
\end{align}
\end{lemma}

\begin{definition}
Let $F \colon \calD \to \Sigma^k$ be a function, where
$\calD$ is a subset of $\binom{V}{k} \times \Rnd$.
We say that $f$ is \defi{$(\alpha, \gamma)$-excellent} if
\begin{align}
    \Pr_{\substack{
        \rv{B} \in \binom{V}{\ell} \\
        (\rv{A}_1, \rv{\rnd}_1) \in \binom{V}{k}_{\rv{B}} \times \Rnd \\ 
        (\rv{A}_2, \rv{\rnd}_2) \in \binom{V}{k}_{\rv{B}} \times \Rnd
    }}\left[
        \begin{aligned}
            (\rv{A}_1, \rv{\rnd}_1) \in \calD \\
            (\rv{A}_2, \rv{\rnd}_2) \in \calD
        \end{aligned}
        \text{ and }
        \far{F(\rv{A}_1; \rv{\rnd}_1)|_{\rv{B}}}{F(\rv{A}_2; \rv{\rnd}_2)|_{\rv{B}}}{\alpha}
    \right] \leq \gamma.
\end{align}
\end{definition}

\begin{proposition}[Lemma 3.10 of \cite{impagliazzo2012new}]
\label{prp:IKW12:lem310}
    Let $F \colon \calD \to \Sigma^k$ be a function, where
    $\calD$ is a subset of $\binom{V}{k} \times \Rnd$.
    Suppose that $\frac{|\calD|}{\bigl|\binom{V}{k} \times \Rnd\bigr|} = \epsilon \geq \omega\bigl(\frac{\ell}{k}\bigr)$ and 
    $F$ is $(\alpha,\gamma)$-excellent,
    where
    \begin{align}
        \alpha = \omega\left(\max\left\{\frac{\ell}{k}, \frac{m}{\ell}\right\} \log\epsilon^{-1}\right)
        \text{ and }
        \gamma = o\bigl(\epsilon^3\bigr).
    \end{align}
    Let $f \colon V \to \Sigma$ be a function such that
    \begin{align}
        f(v) \defeq \PLR_{(A,\rnd) \in \calD: A \ni v}\bigl\{ F(A; \rnd)|_v \bigr\}.
    \end{align}
    Then,
    \begin{align}
        & \Pr_{(\rv{A}, \rv{\rnd}) \in \calD}\left[
            \far{F(\rv{A}; \rv{\rnd})}{f^k(\rv{A})}{\beta}
        \right] \leq \delta, \text{ where }
        \beta \defeq 16\sigma\alpha \text{ and }
        \delta \defeq \frac{64\sigma \gamma}{\epsilon^2}.
    \end{align}
\end{proposition}

By \cref{prp:IKW12:lem310}, we can prove \cref{lem:IKW12:lem38}.
\begin{proof}[Proof of \cref{lem:IKW12:lem38}]
    By applying \cref{prp:IKW12:lem310} to
    $f|_{\calD}$ with $\calD = \Cons(I_0,A_0,\rnd_0)$,
    $\frac{|\calD|}{\bigl|\binom{V\setminus I_0}{k-\ell} \times \Rnd\bigr|} \geq \frac{\epsilon}{2}$, and
    $k \gets k-\ell$.
    we obtain
    \begin{align}
        \close{F(I_0 \cup \rv{A}; \rv{\rnd})|_{\rv{A}}}{f^{k-\ell}(\rv{A})}{\beta}
        \implies
        \dist\bigl( F(I_0 \cup \rv{A}; \rv{\rnd}), f^k(I_0 \cup \rv{A}) \bigr)
        \leq \tfrac{k-\ell}{k}\cdot\beta + \tfrac{\ell}{k}\cdot 1
        \leq \bigl(1+o(1)\bigr)\beta
        \leq 2\beta,
    \end{align}
    as desired.
\end{proof}

\subsection{\texorpdfstring{%
Proof of \cref{prp:IKW12:lem310}
}{%
Proof of Proposition \ref{prp:IKW12:lem310}
}}
In this subsection, we prove \cref{prp:IKW12:lem310}.

For $\calD \subset \binom{V}{k} \times \Rnd$, $v \in V$, and $S \subseteq V$, define
\begin{align}
\begin{aligned}
    \calD_v & \defeq \bigl\{ (A, \rnd) \in \calD \bigm| v \in A \bigr\}, &
    \calD_S & \defeq \bigl\{ (A, \rnd) \in \calD \bigm| S \subseteq A \bigr\}.
\end{aligned}
\end{align}

\begin{definition}
For $(A,\rnd) \in \binom{V}{k} \times \Rnd$, define the event \defi{$\evtA(A,\rnd)$} as
\begin{align}
    \label{eq:IKW12:lem310:A}
    (A,\rnd) \in \calD
    \text{ and }
    \far{F(A; \rnd)}{f^k(A)}{\beta}.
\end{align}

For $B \in \binom{V}{\ell}$, define the event \defi{$\evtB(B)$} as
\begin{align}
\label{eq:IKW12:lem310:B}
\begin{aligned}
    \frac{|\calD_B|}{\bigl|\binom{V}{k}_B \times \Rnd\bigr|}
    & \geq \frac{\epsilon}{2} \;
    \left(
        \text{i.e., }
        \Pr_{(\rv{A}, \rv{\rnd}) \in \binom{V}{k}_B \times \Rnd}\bigl[
            (\rv{A}, \rv{\rnd}) \in \calD
        \bigr] \geq \frac{\epsilon}{2}
    \right),
    \text{ and } \\
    \Pr_{\rv{v} \in B}\Bigl[
        \Pr_{(\rv{A}, \rv{\rnd}) \in \calD_B}\bigl[
            F(\rv{A}; \rv{\rnd})|_{\rv{v}} = f(\rv{v})
        \bigr] \geq \tfrac{1}{2\sigma}
    \Bigr]
    & \geq 1 - \bigO\left(\tfrac{\log \epsilon^{-1}}{k/\ell}\right).
\end{aligned}
\end{align}
\end{definition}

\begin{remark}
The event $\evtB(B)$ can be thought of as a ``$B$-restricted'' version of the following assumptions:
\begin{align}
    \frac{|\calD|}{\bigl|\binom{V}{k} \times \Rnd\bigr|} \geq \epsilon
    \text{ and }
    \Pr_{(\rv{A}, \rv{\rnd}) \in \calD_v}\bigl[
        F(\rv{A})|_v = f(v)
    \bigr] \geq \tfrac{1}{\sigma}.
\end{align}
\end{remark}

\paragraph{(Step 1)}
We first show that $\evtB(\rv{B})$ holds with probability $\frac{1}{2}$.

\begin{lemma}[Proof of \cref{prp:IKW12:lem310}: Part 1]
\label{lem:IKW12:lem310:1}
    \begin{align}
        \Pr_{\rv{B} \in \binom{V}{\ell}}\bigl[
            \evtB(\rv{B})
        \bigr] \geq \frac{1}{2}.
    \end{align}
\end{lemma}

\begin{claim}[Claim 3.11 of \cite{impagliazzo2012new}]
\label{clm:IKW12:clm311}
    \begin{align}
        \Pr_{\rv{v} \in V}\left[
            \frac{|\calD_{\rv{v}}|}{\bigl|\binom{V}{k}_{\rv{v}}\bigr|} \geq \frac{\epsilon}{2}
        \right]
        = \Pr_{\rv{v} \in V}\left[
            \Pr_{(\rv{A}, \rv{\rnd}) \in \binom{V}{k}_{\rv{v}} \times \Rnd}\bigl[
                \rv{A} \in \calD
            \bigr] \geq \tfrac{\epsilon}{2}
        \right]
        \geq 1 - \bigO\left(\tfrac{\log \epsilon^{-1}}{k}\right).
    \end{align}
\end{claim}
\begin{proof}
By \cite[Lemma~2.6]{impagliazzo2012new},
we have
\begin{align}
    \Pr_{\rv{v} \in V}\left[
        \Pr_{(\rv{A}, \rv{\rnd}) \in \binom{V}{k}_{\rv{v}} \times \Rnd}\bigl[
            (\rv{A}, \rv{\rnd}) \in \calD
        \bigr] \geq \tfrac{\epsilon}{2}
    \right]
    \geq 1 - \bigO\left(\tfrac{\log \epsilon^{-1}}{k}\right),
\end{align}
as desired.
\end{proof}

\begin{claim}[Claim 3.12 of \cite{impagliazzo2012new}]
\label{clm:IKW12:clm312}
Suppose that $v \in V$ satisfies
$\frac{|\calD_v|}{\bigl|\binom{V}{k}_v \times \Rnd\bigr|} \geq \frac{\epsilon}{2}$.
Then,
\begin{align}
    \Pr_{\rv{B} \in \binom{V}{\ell}_v}\left[
        \Pr_{(\rv{A}, \rv{\rnd}) \in \calD_{\rv{B}}}\bigl[
            F(\rv{A}; \rv{\rnd})|_v = f(v)
        \bigr] \geq \tfrac{1}{2\sigma}
    \right] \geq 1 - \bigO\left(\tfrac{\log \epsilon^{-1}}{k/\ell}\right).
\end{align}
\end{claim}
\begin{proof}
Since $f(v)$ is determined based on the plurality, we have
\begin{align}
    \Pr_{(\rv{A}, \rv{\rnd}) \in \calD_v}\bigl[
        F(\rv{A})|_v = f(v)
    \bigr] \geq \tfrac{1}{\sigma}.
\end{align}
Define
\begin{align}
    \mu_1 & \defeq \Pr_{(\rv{A}, \rv{\rnd}) \in \binom{V}{k}_v \times \Rnd}\bigl[
        (\rv{A}, \rv{\rnd}) \in \calD
    \bigr], \\
    \mu_2 & \defeq \Pr_{(\rv{A}, \rv{\rnd}) \in \binom{V}{k}_v \times \Rnd}\bigl[
        (\rv{A}, \rv{\rnd}) \in \calD \text{ and } F(\rv{A}; \rv{\rnd})|_v = f(v)
    \bigr].
\end{align}
Note that $\frac{\mu_2}{\mu_1} \geq \frac{1}{\sigma}$.
By applying \cite[Corollary~2.7]{impagliazzo2012new},
we have
\begin{align}
    \Pr_{\rv{B} \in \binom{V}{\ell}_v}\left[
        \Pr_{(\rv{A}, \rv{\rnd}) \in \binom{V}{k}_{\rv{B}} \times \Rnd}\bigl[
            (\rv{A}, \rv{\rnd}) \in \calD
        \bigr] \in \bigl(0.8\mu_1, 1.2\mu_1\bigr)
    \right] & \geq 1 - \bigO\left(\tfrac{\log \mu_1^{-1}}{k/\ell}\right), \\
    \Pr_{\rv{B} \in \binom{V}{\ell}_v}\left[
        \Pr_{(\rv{A}, \rv{\rnd}) \in \binom{V}{k}_{\rv{B}} \times \Rnd}\bigl[
            (\rv{A}, \rv{\rnd}) \in \calD \text{ and } F(\rv{A}; \rv{\rnd})|_v = f(v)
        \bigr] \in \bigl(0.8\mu_2, 1.2\mu_2\bigr)
    \right] & \geq 1 - \bigO\left(\tfrac{\log \mu_2^{-1}}{k/\ell}\right).
\end{align}
Note that $\rv{A}$ and $\rv{B}$ always contain $v$.
By assumption, $\mu_1 \geq \frac{\epsilon}{2}$ and thus $\mu_2 \geq \frac{\epsilon}{2\sigma}$; i.e.,
$\mu_1 = \Omega(\epsilon)$ and $\mu_2 = \Omega(\epsilon)$.
By taking the union bound,
with probability at least $1 - \bigO\left(\tfrac{\log \epsilon^{-1}}{k/\ell}\right)$
over $\rv{B} \in \binom{V}{\ell}_v$,
we have
\begin{align}
    \Pr_{(\rv{A}, \rv{\rnd}) \in \calD_{\rv{B}}}\bigl[
        F(\rv{A}; \rv{\rnd})|_v = f(v)
    \bigr]
    = \frac{\displaystyle
        \Pr_{(\rv{A}, \rv{\rnd}) \in \binom{V}{k}_{\rv{B}} \times \Rnd}\bigl[
            (\rv{A}, \rv{\rnd}) \in \calD \text{ and } F(\rv{A}; \rv{\rnd})|_v = f(v)
        \bigr]
    }{\displaystyle
        \Pr_{(\rv{A}, \rv{\rnd}) \in \binom{V}{k}_{\rv{B}} \times \Rnd}\bigl[
            (\rv{A}, \rv{\rnd}) \in \calD
        \bigr]
    }
    \geq \frac{0.8\mu_2}{1.2\mu_1}
    \geq \frac{1}{2\sigma},
\end{align}
as desired.
\end{proof}

\begin{claim}[Claim 3.13 of \cite{impagliazzo2012new}]
\label{clm:IKW12:clm313}
    \begin{align}
        \Pr_{\substack{
            \rv{v} \in V \\
            \rv{B} \in \binom{V}{\ell}_{\rv{v}}
        }}\left[
            \Pr_{(\rv{A}, \rv{\rnd}) \in \calD_{\rv{B}}}\bigl[
                F(\rv{A}; \rv{\rnd})|_{\rv{v}} = f(\rv{v})
            \bigr] \geq \tfrac{1}{2\sigma}
        \right] \geq 1 - \bigO\left(\tfrac{\log \epsilon^{-1}}{k/\ell}\right).
    \end{align}
\end{claim}
\begin{proof}
    By applying the union bound to \cref{clm:IKW12:clm311,clm:IKW12:clm312},
    we obtain the desired result.
\end{proof}

\begin{claim}[Claim 3.14 of \cite{impagliazzo2012new}]
\label{clm:IKW12:clm314}
\begin{align}
    \Pr_{\rv{B} \in \binom{V}{\ell}}\biggl[
        \Pr_{\rv{v} \in \rv{B}}\Bigl[
            \Pr_{(\rv{A}, \rv{\rnd}) \in \calD_{\rv{B}}}\bigl[
                F(\rv{A}; \rv{\rnd})|_{\rv{v}} = f(\rv{v})
            \bigr] \geq \tfrac{1}{2\sigma}
        \Bigr] \geq 1 - \bigO\left(\tfrac{\log \epsilon^{-1}}{k/\ell}\right)
    \biggr] \geq 0.99.
\end{align}
\end{claim}
\begin{proof}
By swapping the order of choosing $\rv{v}$ and $\rv{B}$ and applying an averaging argument to \cref{clm:IKW12:clm313},
we obtain the desired result.
\end{proof}

\begin{claim}[Claim 3.15 of \cite{impagliazzo2012new}]
    \label{clm:IKW12:clm315}
    \begin{align}
        \Pr_{\rv{B} \in \binom{V}{\ell}}\left[
            \frac{|\calD_{\rv{B}}|}{\bigl|\binom{V}{k}_{\rv{B}} \times \Rnd\bigr|}
                \geq \frac{\epsilon}{2}
        \right]
        = \Pr_{\rv{B} \in \binom{V}{\ell}}\left[
            \Pr_{(\rv{A}, \rv{\rnd}) \in \binom{V}{k}_{\rv{B}} \times \Rnd}\bigl[
                (\rv{A}, \rv{\rnd}) \in \calD
            \bigr] \geq \tfrac{\epsilon}{2}
        \right]
        \geq 1 - \bigO\left(\tfrac{\log \epsilon^{-1}}{k/\ell}\right).
    \end{align}
\end{claim}
\begin{proof}
By applying \cite[Corollary~2.7]{impagliazzo2012new},
we have
\begin{align}
    \Pr_{\rv{B} \in \binom{V}{\ell}}\left[
        \Pr_{(\rv{A}, \rv{\rnd}) \in \binom{V}{k}_{\rv{B}} \times \Rnd}\bigl[
            (\rv{A},\rv{\rnd}) \in \calD
        \bigr]
        \geq \tfrac{\epsilon}{2}
    \right] \geq 1 - \bigO\left(\tfrac{\log \epsilon^{-1}}{k/\ell}\right),
\end{align}
as desired.
\end{proof}

We now prove \cref{lem:IKW12:lem310:1}.
\begin{proof}[Proof of \cref{lem:IKW12:lem310:1}]
    By applying the union bound to \cref{clm:IKW12:clm314,clm:IKW12:clm315},
    we have
    \begin{align}
        \Pr_{\rv{B} \in \binom{V}{\ell}}\bigl[
            \evtB(\rv{B})
        \bigr]
        \geq 0.99 - \bigO\left(\tfrac{\log \epsilon^{-1}}{k/\ell}\right)
        \geq 0.99 - o(1)
        \geq \frac{1}{2},
    \end{align}
    as desired.
\end{proof}

\paragraph{(Step 2)}
By applying \cref{lem:IKW12:lem310:1}, we next show the following.

\begin{lemma}[Proof of \cref{prp:IKW12:lem310}: Part 2]
\label{lem:IKW12:lem310:2}
\begin{align}
    \Pr_{\substack{
        (\rv{A}, \rv{\rnd}) \in \binom{V}{k} \times \Rnd \\
        \rv{B} \in \binom{\rv{A}}{\ell}
    }}\left[
        \evtB(\rv{B})
        \text{ and }
        \far{F(\rv{A}; \rv{\rnd})|_{\rv{B}}}{f^\ell(\rv{B})}{\beta/2}
        \Bigm|
        \evtA(\rv{A}, \rv{\rnd})
    \right] \geq \frac{1}{8}.
\end{align}
\end{lemma}
\begin{proof}
Condition $(\rv{A}, \rv{\rnd}) \in \binom{V}{k} \times \Rnd$
on the event that $\evtA(\rv{A}, \rv{\rnd})$, which
implies $\far{F(\rv{A}; \rv{\rnd})}{f^k(\rv{A})}{\beta}$.
By the Chernoff bound, we have
\begin{align}
    \Pr_{\substack{
        (\rv{A}, \rv{\rnd}) \in \binom{V}{k} \times \Rnd \\
        \rv{B} \in \binom{\rv{A}}{\ell}
    }}\left[
        \far{F(\rv{A}; \rv{\rnd})|_{\rv{B}}}{f^\ell(\rv{B})}{\beta/2}
        \;\middle|\;
        \evtA(\rv{A}, \rv{\rnd})
    \right] \geq 1 - \exp\bigl(-\Omega(\beta \ell)\bigr).
\end{align}
By applying \cite[Lemmas~2.3~and~2.4]{impagliazzo2012new} to \cref{lem:IKW12:lem310:1},
we have
\begin{align}
    \Pr_{\substack{
        (\rv{A}, \rv{\rnd}) \in \binom{V}{k} \times \Rnd \\
        \rv{B} \in \binom{\rv{A}}{\ell}
    }}\Bigl[
        \evtB(\rv{B})
        \Bigm|
        \evtA(\rv{A}, \rv{\rnd})
    \Bigr] \geq \frac{1}{6}.
\end{align}
By the union bound, we have
\begin{align}
\begin{aligned}
    & \Pr_{\substack{
        (\rv{A}, \rv{\rnd}) \in \binom{V}{k} \times \Rnd \\
        \rv{B} \in \binom{\rv{A}}{\ell}
    }}\left[
        \evtB(\rv{B})
        \text{ and }
        \far{F(\rv{A}; \rv{\rnd})|_{\rv{B}}}{f^\ell(\rv{B})}{\beta/2}
        \;\middle|\;
        \evtA(\rv{A}, \rv{\rnd})
    \right]
    \geq \frac{1}{6} - \exp\bigl(-\Omega(\beta \ell)\bigr)
    \geq \frac{1}{6} - o(1)
    \geq \frac{1}{8},
\end{aligned}
\end{align}
where we used the assumption that $\beta = \omega\left(\frac{\log \epsilon^{-1}}{\ell}\right)$,
as desired.
\end{proof}

\paragraph{(Step 3)}
We then show the following.

\begin{lemma}[Proof of \cref{prp:IKW12:lem310}: Part 3]
\label{lem:IKW12:lem310:3}
Suppose that
$(A_1, \rnd_1) \in \binom{V}{k} \times \Rnd$ satisfies $\evtA(A_1, \rnd_1)$, 
$B \in \binom{A_1}{\ell}$ satisfies $\evtB(B)$, and
$\far{F(A_1; \rnd_1)|_B}{f^\ell(B)}{\beta/2}$.
Then,
\begin{align}
    \Pr_{(\rv{A}_2, \rv{\rnd}_2) \in \binom{V}{k}_B \times \Rnd}\left[
        (\rv{A}_2, \rv{\rnd}_2) \in \calD \text{ and }
        \far{F(\rv{A}_2; \rv{\rnd}_2)|_B}{F(A_1; \rnd_1)|_B}{\beta/(16\sigma)}
    \right] > \frac{\epsilon}{8\sigma}.
\end{align}
\end{lemma}
\begin{proof}
Define
\begin{align}
    B^\circ \defeq \left\{
        v \in B
        \;\middle|\;
        F(A_1; \rnd_1)|_v \neq f(v)
        \text{ and }
        \Pr_{(\rv{A}_2, \rv{\rnd}_2) \in \calD_B}\bigl[
            F(\rv{A}_2; \rv{\rnd}_2)|_v = f(v)
        \bigr] \geq \tfrac{1}{2\sigma}
    \right\}.
\end{align}
Since $\far{F(A_1; \rnd_1)|_B}{f^\ell(B)}{\beta/2}$ and $\evtB(B)$ holds
by assumption, we have
\begin{align}
\label{eq:IKW12:lem310:3:Bcirc}
\begin{aligned}
    |B^\circ|
    & \geq |B| \left(
        \Pr_{\rv{v} \in B}\Bigl[ F(A_1; \rnd_1)|_{\rv{v}} \neq f(\rv{v}) \Bigr]
        + \Pr_{\rv{v} \in B}\Bigl[
            \Pr_{(\rv{A}_2, \rv{\rnd}_2) \in \calD_B}\bigl[
                F(\rv{A}_2; \rv{\rnd}_2)|_{\rv{v}} = f(\rv{v})
            \bigr] \geq \tfrac{1}{2\sigma}
        \Bigr]
        - 1
    \right) \\
    & \geq |B|\left(\tfrac{\beta}{2} + 1- \bigO\left(\tfrac{\log \epsilon^{-1}}{k/\ell}\right)-1\right)
    \geq |B|\tfrac{\beta}{2}\bigl(1 - o(1)\bigr)
    > \tfrac{\beta}{4}|B|,
\end{aligned}
\end{align}
where we used the assumption that $\beta = \omega\left(\tfrac{\log \epsilon^{-1}}{k/\ell}\right)$.
Observe that for every $v \in B^\circ$,
\begin{align}
    \Pr_{(\rv{A}_2, \rv{\rnd}_2) \in \calD_B}\Bigl[
        F(\rv{A}_2; \rv{\rnd}_2)|_v = f(v) \neq F(A_1; \rnd_1)|_v
    \Bigr] \geq \tfrac{1}{2\sigma},
\end{align}
implying that by an averaging argument,
\begin{align}
\begin{aligned}
    \Pr_{(\rv{A}_2, \rv{\rnd}_2) \in \calD_B}\Bigl[
        \underbrace{\Pr_{\rv{v} \in B^\circ}\bigl[
            F(\rv{A}_2; \rv{\rnd}_2)|_{\rv{v}} \neq F(A_1; \rnd_1)|_{\rv{v}}
        \bigr]}_{= \dist\left(F(\rv{A}_2; \rv{\rnd}_2)|_{B^\circ}, F(A_1; \rnd_1)|_{B^\circ}\right)}
        \geq \tfrac{1}{4\sigma}
    \Bigr] & \geq \tfrac{1}{4\sigma}, \\
    \implies \Pr_{(\rv{A}_2, \rv{\rnd}_2) \in \calD_B}\Bigl[
        \dist\bigl(F(\rv{A}_2; \rv{\rnd}_2)|_B, F(A_1; \rnd_1)|_B\bigr) > \tfrac{\beta}{16\sigma}
    \Bigr] & \geq \tfrac{1}{4\sigma},
\end{aligned}
\end{align}
where we used the inequality that
\begin{align}
\begin{aligned}        
    \dist\bigl(F(\rv{A}_2; \rv{\rnd}_2)|_B, F(A_1; \rnd_1)|_B\bigr)
    & \geq \tfrac{|B^\circ|}{|B|} \dist\bigl(F(\rv{A}_2; \rv{\rnd}_2)|_{B^\circ}, F(A_1; \rnd_1)|_{B^\circ}\bigr) \\
    & \underbrace{>}_{\text{\cref{eq:IKW12:lem310:3:Bcirc}}}
        \tfrac{\beta}{4} \dist\bigl(F(\rv{A}_2; \rv{\rnd}_2)|_{B^\circ}, F(A_1; \rnd_1)|_{B^\circ}\bigr).
\end{aligned}
\end{align}
Since $\frac{|\calD_B|}{\bigl|\binom{V}{k}_B \times \Rnd\bigr|} \geq \frac{\epsilon}{2}$ due to $\evtB(B)$, we have
\begin{align}
\begin{aligned}
    & \Pr_{(\rv{A}_2, \rv{\rnd}_2) \in \binom{V}{k}_B \times \Rnd}\left[
        (\rv{A}_2, \rv{\rnd}_2) \in \calD
        \text{ and }
        \far{F(\rv{A}_2; \rv{\rnd}_2)|_B}{F(A_1, \rnd_1)|_B}{\beta/(16\sigma)}
    \right] \\
    & = \Pr_{(\rv{A}_2, \rv{\rnd}_2) \in \calD_B}\left[
            \far{F(\rv{A}_2; \rv{\rnd}_2)|_B}{F(A_1; \rnd_1)|_B}{\beta/(16\sigma)}
        \right]
        \cdot \Pr_{(\rv{A}_2, \rv{\rnd}_2) \in \binom{V}{k}_B \times \Rnd}\left[
            (\rv{A}_2, \rv{\rnd}_2) \in \calD
        \right]
    \geq \frac{1}{4\sigma} \cdot \frac{\epsilon}{2}
    = \frac{\epsilon}{8\sigma},
\end{aligned}
\end{align}
as desired.
\end{proof}

\paragraph{Proof of \cref{prp:IKW12:lem310}.}
We are now ready to prove \cref{prp:IKW12:lem310}.
\begin{proof}[Proof of \cref{prp:IKW12:lem310}]
    Suppose for contradiction that
    \begin{align}
        \Pr_{(\rv{A}, \rv{\rnd}) \in \calD}\left[
            \far{F(\rv{A}; \rv{\rnd})}{f^k(\rv{A})}{\beta}
        \right] > \delta,
        \qquad \therefore
        \Pr_{(\rv{A}, \rv{\rnd}) \in \binom{V}{k} \times \Rnd}\bigl[
            \evtA(\rv{A}, \rv{\rnd})
        \bigr] > \delta\epsilon.
    \end{align}
    By \cref{lem:IKW12:lem310:3}, we have
    \begin{align}
        \Pr_{\substack{
            (\rv{A}_1, \rv{\rnd}_1) \in \binom{V}{k} \times \Rnd \\
            \rv{B} \in \binom{\rv{A}_1}{\ell} \\
            (\rv{A}_2, \rv{\rnd}_2) \in \binom{V}{k}_{\rv{B}} \times \Rnd
        }}\left[
            \begin{aligned}
                & (\rv{A}_2, \rv{\rnd}_2) \in \calD \\
                & \far{F(\rv{A}_1; \rv{\rnd}_1)|_{\rv{B}}}{F(\rv{A}_2; \rv{\rnd}_2)|_{\rv{B}}}{\beta/(16\sigma)}
            \end{aligned}
            \;\middle|\;
            \begin{aligned}
                & \evtA(\rv{A}_1, \rv{\rnd}_1) \\
                & \evtB(\rv{B}) \text{ and }
                    \far{F(\rv{A}; \rv{\rnd}_1)|_{\rv{B}}}{f^\ell(\rv{B})}{\beta/2}
            \end{aligned}
        \right] > \frac{\epsilon}{8\sigma}.
    \end{align}
    By \cref{lem:IKW12:lem310:2}, we have
    \begin{align}
        \Pr_{\substack{
            (\rv{A}_1, \rv{\rnd}_1) \in \binom{V}{k} \times \Rnd \\
            \rv{B} \in \binom{\rv{A}_1}{\ell} \\
            (\rv{A}_2, \rv{\rnd}_2) \in \binom{V}{k}_{\rv{B}} \times \Rnd
        }}\left[
            (\rv{A}_2, \rv{\rnd}_2) \in \calD
            \text{ and }
            \far{F(\rv{A}_1; \rv{\rnd}_1)|_{\rv{B}}}{F(\rv{A}_2; \rv{\rnd}_2)|_{\rv{B}}}{\beta/(16\sigma)}
            \;\middle|\;
            \evtA(\rv{A}_1, \rv{\rnd}_1)
        \right]
        > \frac{\epsilon}{8\sigma} \cdot \frac{1}{8}
        = \frac{\epsilon}{64\sigma}.
    \end{align}
    By assumption, we have
    \begin{align}
    \begin{aligned}
        \Pr_{\substack{
            (\rv{A}_1, \rv{\rnd}_1) \in \binom{V}{k} \times \Rnd \\
            \rv{B} \in \binom{\rv{A}_1}{\ell} \\
            (\rv{A}_2, \rv{\rnd}_2) \in \binom{V}{k}_{\rv{B}} \times \Rnd
        }}\left[
            (\rv{A}_2, \rv{\rnd}_2) \in \calD
            \text{ and }
            \far{F(\rv{A}_1; \rv{\rnd}_1)|_{\rv{B}}}{F(\rv{A}_2; \rv{\rnd}_2)|_{\rv{B}}}{\beta/(16\sigma)}
            \text{ and }
            \evtA(\rv{A}_1, \rv{\rnd}_1)
        \right] & > \frac{\epsilon}{64\sigma} \cdot \delta \epsilon, \\
        \implies \Pr_{\substack{
            \rv{B} \in \binom{V}{\ell} \\
            (\rv{A}_1, \rv{\rnd}_1) \in \binom{V}{k}_{\rv{B}} \times \Rnd \\
            (\rv{A}_2, \rv{\rnd}_2) \in \binom{V}{k}_{\rv{B}} \times \Rnd
        }}\left[
            (\rv{A}_1, \rv{\rnd}_1) \in \calD
            \text{ and }
            (\rv{A}_2, \rv{\rnd}_2) \in \calD
            \text{ and }
            \far{F(\rv{A}_1; \rv{\rnd}_1)|_{\rv{B}}}{F(\rv{A}_2; \rv{\rnd}_2)|_{\rv{B}}}{\alpha}
        \right] & > \gamma,
    \end{aligned}
    \end{align}
    which contradicts to the assumption that $F$ is $(\alpha,\gamma)$-excellent,
    as desired.
\end{proof}

\subsection{\texorpdfstring{%
Proof of \cref{thm:IKW12}
}{%
Proof of Theorem \ref{thm:IKW12}
}}

Suppose hereafter that
\begin{align}
\begin{aligned}
    \epsilon & = \omega\left(\frac{\ell}{k}\right), \\
    \alpha & = \omega\left(\max\left\{ \frac{\ell}{k}, \frac{m}{\ell} \right\}\log \epsilon^{-1} \right), &
    \beta & \defeq 16\sigma\alpha, \\
    \gamma & = o(\epsilon^3), &
    \delta & \defeq \frac{256\sigma\gamma}{\epsilon^2} = o(\epsilon).
\end{aligned}
\end{align}

Consider the following procedure $\Sample$ that generates $(\rv{I}_1, \rv{I}_2, \rv{A}_1, \rv{A}_2, \rv{X})$:
\begin{itembox}[l]{$\Sample$}
\begin{itemize}
    \item sample
    $\rv{I}_1 \in \binom{V}{\ell}$ and
    $\rv{I}_2 \in \binom{V}{\ell}$ such that
    $\rv{I}_1 \cap \rv{I}_2 = \emptyset$.
    \item sample
    $\rv{A}_1 \in \binom{V\setminus\rv{I}_1}{k-\ell}$ and
    $\rv{A}_2 \in \binom{V\setminus\rv{I}_2}{k-\ell}$.
    \item sample $\rv{X} \in \binom{V}{k}$ such that $\rv{X} \supset \rv{I}_1 \cup \rv{I}_2$.
\end{itemize}
\end{itembox}

For each $i \in [2]$, define $\rv{f}_i \colon V \to \Sigma$,
which depends on $\rv{I}_i$ and $\rv{A}_i$,
as follows:
\begin{align}
    \rv{f}_i(v) & \defeq \PLR_{(A, \rnd) \in \Cons(\rv{I}_i, \rv{A}_i) \times \Rnd : A \ni v}
    \bigl\{ F(\rv{I}_i \cup A; \rnd)|_v \bigr\}.
\end{align}

We prove the following claims about \Sample.

\begin{claim}[Claim 3.18 of \cite{impagliazzo2012new}]
\label{lem:IKW12:clm318}
    \begin{align}
        \Pr_{\substack{
            (\rv{I}_1, \rv{I}_2, \rv{A}_1, \rv{A}_2, \rv{X}) \sim \Sample \\
            \rv{\rnd}_1, \rv{\rnd}_2, \rv{\rnd}_1', \rv{\rnd}_2' \in \Rnd
        }}\left[
            \begin{aligned}
                (\rv{I}_1, \rv{A}_1, \rv{\rnd}_1) \text{ is $\left(\alpha,\gamma,\tfrac{\epsilon}{2}\right)$-excellent} \\
                (\rv{I}_2, \rv{A}_2, \rv{\rnd}_2) \text{ is $\left(\alpha,\gamma,\tfrac{\epsilon}{2}\right)$-excellent}
            \end{aligned}
            \text{ and }
            \begin{aligned}
                (\rv{X}\setminus\rv{I}_1, \rv{\rnd}_1') \in \Cons(\rv{I}_1, \rv{A}_1, \rv{\rnd}_1) \\
                (\rv{X}\setminus\rv{I}_2, \rv{\rnd}_2') \in \Cons(\rv{I}_2, \rv{A}_2, \rv{\rnd}_2)
            \end{aligned}
        \right]
        \geq \Omega(\epsilon^5).
    \end{align}
\end{claim}
\begin{proof}
Consider the following procedure that generates $\Sample'$: $(\rv{X}, \rv{\calP}, \rv{I}_1, \rv{I}_2, \rv{A}_1, \rv{A}_2)$:
\begin{itembox}[l]{$\Sample'$}
\begin{itemize}
    \item sample $\rv{X} \in \binom{V}{k}$ and
    its $k/\ell$-partition $\rv{\calP} = (\rv{P}_1, \ldots, \rv{P}_{k/\ell})$
    such that $|\rv{P}_1| = \cdots = |\rv{P}_{k/\ell}| = \ell$.
    \item sample $\{\rv{I}_1, \rv{I}_2\} \in \binom{\rv{\calP}}{2}$.
    \item sample
    $\rv{A}_1 \in \binom{V\setminus\rv{I}_1}{k-\ell}$ and
    $\rv{A}_2 \in \binom{V\setminus\rv{I}_2}{k-\ell}$.
\end{itemize}
\end{itembox}
Note that $\Sample$ and $\Sample'$ have the same distribution of $(\rv{X}, \rv{I}_1, \rv{I}_2, \rv{A}_1, \rv{A}_2)$.
By \cref{lem:IKW12:lem35,lem:IKW12:cor37}, we have
\begin{align}
    \Pr_{\substack{
        (\rv{X},\rv{\calP}) \sim \Sample' \\
        \rv{I} \in \rv{\calP}\\
        \rv{\rnd} \in \Rnd
    }}\bigl[
        (\rv{I}, \rv{X} \setminus \rv{I}, \rv{\rnd})
            \text{ is $\left(\alpha,\gamma,\tfrac{\epsilon}{2}\right)$-excellent}
    \bigr]
    \geq \frac{\epsilon}{2}\left(
        1-\frac{\exp\bigl(-\Omega(\alpha \ell)\bigr)}{\gamma \epsilon}
    \right)
    \geq \frac{\epsilon}{2} \bigl(1-o(1)\bigr)
    \geq \frac{\epsilon}{4},
\end{align}
implying that by an averaging argument,
\begin{align}
    \Pr_{(\rv{X}, \rv{\calP}) \sim \Sample'}\left[
        \Pr_{\substack{
            \rv{I} \in \rv{\calP} \\
            \rv{\rnd} \in \Rnd
        }}\bigl[
            (\rv{I}, \rv{X} \setminus \rv{I}, \rv{\rnd})
                \text{ is $\left(\alpha,\gamma,\tfrac{\epsilon}{2}\right)$-excellent}
        \bigr] \geq \tfrac{\epsilon}{8}
    \right] \geq \tfrac{\epsilon}{8}.
\end{align}
Conditioning $\rv{X}, \rv{\calP}$ on the above event, we have
\begin{align}
\begin{aligned}
    & \Pr_{\substack{
        (\rv{X}, \rv{\calP}, \rv{I}_1, \rv{I}_2) \sim \Sample' \\
        \rv{\rnd}_1', \rv{\rnd}_2' \in \Rnd
    }}\left[
        \begin{aligned}
            (\rv{I}_1, \rv{X} \setminus \rv{I}_1, \rv{\rnd}_1') \text{ is $\left(\alpha,\gamma,\tfrac{\epsilon}{2}\right)$-excellent} \\
            (\rv{I}_2, \rv{X} \setminus \rv{I}_2, \rv{\rnd}_2') \text{ is $\left(\alpha,\gamma,\tfrac{\epsilon}{2}\right)$-excellent}
        \end{aligned}
        \;\middle|\;
        \Pr_{\substack{
            \rv{I} \in \rv{\calP} \\
            \rv{\rnd} \in \Rnd
        }}\bigl[
            (\rv{I}, \rv{X} \setminus \rv{I}, \rv{\rnd}) \text{ is $\left(\alpha,\gamma,\tfrac{\epsilon}{2}\right)$-excellent}
        \bigr] \geq \tfrac{\epsilon}{8}
    \right] \\
    & \geq \tfrac{\epsilon}{8}\left(\tfrac{\epsilon}{8} - \tfrac{\ell}{k}\right)
    \geq \Omega(\epsilon^2),
\end{aligned}
\end{align}
where we used the assumption that $\epsilon = \omega\left(\frac{\ell}{k}\right)$.
Conditioning $\rv{X}, \rv{I}_1, \rv{I}_2$ on the above event,
we have
\begin{align}
\begin{aligned}
    & \Pr_{\substack{
        (\rv{X}, \rv{\calP}, \rv{I}_1, \rv{I}_2, \rv{A}_1, \rv{A}_2) \sim \Sample' \\
        \rv{\rnd}_1, \rv{\rnd}_2, \rv{\rnd}_1', \rv{\rnd}_2' \in \Rnd
    }}\left[
        \begin{aligned}
            (\rv{A}_1, \rv{\rnd}_1) \in \Cons(\rv{I}_1, \rv{X} \setminus \rv{I}_1, \rv{\rnd}_1') \\
            (\rv{A}_2, \rv{\rnd}_2) \in \Cons(\rv{I}_2, \rv{X} \setminus \rv{I}_2, \rv{\rnd}_2')
        \end{aligned}
        \;\middle|\;
        \begin{aligned}
            (\rv{I}_1, \rv{X} \setminus \rv{I}_1, \rv{\rnd}_1') \text{ is $\left(\alpha,\gamma,\tfrac{\epsilon}{2}\right)$-excellent} \\
            (\rv{I}_2, \rv{X} \setminus \rv{I}_2, \rv{\rnd}_2') \text{ is $\left(\alpha,\gamma,\tfrac{\epsilon}{2}\right)$-excellent}
        \end{aligned}
    \right] \\
    & \geq \frac{\epsilon}{2}\cdot\frac{\epsilon}{2}
    = \Omega(\epsilon^2).
\end{aligned}
\end{align}
Putting the above three inequalities together, we have
\begin{align}
    \Pr_{\substack{
        (\rv{X}, \rv{\calP}, \rv{I}_1, \rv{I}_2, \rv{A}_1, \rv{A}_2) \sim \Sample' \\
        \rv{\rnd}_1, \rv{\rnd}_2, \rv{\rnd}_1', \rv{\rnd}_2' \in \Rnd
    }}\left[
        \begin{aligned}
            (\rv{I}_1, \rv{X} \setminus \rv{I}_1, \rv{\rnd}_1') \text{ is $\left(\alpha,\gamma,\tfrac{\epsilon}{2}\right)$-excellent} \\
            (\rv{I}_2, \rv{X} \setminus \rv{I}_2, \rv{\rnd}_2') \text{ is $\left(\alpha,\gamma,\tfrac{\epsilon}{2}\right)$-excellent}
        \end{aligned}
        \text{ and }
        \begin{aligned}
            (\rv{A}_1, \rv{\rnd}_1) \in \Cons(\rv{I}_1, \rv{X} \setminus \rv{I}_1, \rv{\rnd}_1') \\
            (\rv{A}_2, \rv{\rnd}_2) \in \Cons(\rv{I}_2, \rv{X} \setminus \rv{I}_2, \rv{\rnd}_2')
        \end{aligned}
    \right] & \geq \Omega(\epsilon^5), \\
    \implies \Pr_{\substack{
            (\rv{I}_1, \rv{I}_2, \rv{A}_1, \rv{A}_2, \rv{X}) \sim \Sample \\
            \rv{\rnd}_1, \rv{\rnd}_2, \rv{\rnd}_1', \rv{\rnd}_2' \in \Rnd
        }}\left[
        \begin{aligned}
            (\rv{I}_1, \rv{A}_1, \rv{\rnd}_1) \text{ is $\left(\alpha,\gamma,\tfrac{\epsilon}{2}\right)$-excellent} \\
            (\rv{I}_2, \rv{A}_2, \rv{\rnd}_2) \text{ is $\left(\alpha,\gamma,\tfrac{\epsilon}{2}\right)$-excellent}
        \end{aligned}
        \text{ and }
        \begin{aligned}
            (\rv{X} \setminus \rv{I}_1, \rv{\rnd}_1') \in \Cons(\rv{I}_1, \rv{A}_1, \rv{\rnd}_1) \\
            (\rv{X} \setminus \rv{I}_2, \rv{\rnd}_2') \in \Cons(\rv{I}_2, \rv{A}_2, \rv{\rnd}_2)
        \end{aligned}
    \right] & \geq \Omega(\epsilon^5),
\end{align}
where we note that $\rv{A}_i$ and $\rv{X} \setminus \rv{I}_i$ can be swapped,
as desired.
\end{proof}

\begin{claim}[Claim 3.19 of \cite{impagliazzo2012new}]
\label{lem:IKW12:clm319}
\begin{align}
    \Pr_{\substack{
        (\rv{I}_1, \rv{I}_2, \rv{A}_1, \rv{A}_2, \rv{X}) \sim \Sample \\
        \rv{\rnd}_1, \rv{\rnd}_2 \in \Rnd
    }}\left[
        \begin{aligned}
            (\rv{I}_1, \rv{A}_1, \rv{\rnd}_1) \text{ is $\left(\alpha,\gamma,\tfrac{\epsilon}{2}\right)$-excellent} \\
            (\rv{I}_2, \rv{A}_2, \rv{\rnd}_2) \text{ is $\left(\alpha,\gamma,\tfrac{\epsilon}{2}\right)$-excellent}
        \end{aligned}
        \text{ and }
        \close{\rv{f}_1^k(\rv{X})}{\rv{f}_2^k(\rv{X})}{4\beta}
    \right]
    \geq \Omega(\epsilon^5).
\end{align}
\end{claim}
\begin{proof}
By \cref{lem:IKW12:lem38}, we have
\begin{align}
    \Pr_{\substack{
        (\rv{I}_1, \rv{I}_2, \rv{A}_1, \rv{A}_2, \rv{X}) \sim \Sample \\
        \rv{\rnd}_1, \rv{\rnd}_1' \in \Rnd
    }}\left[
        \far{F(\rv{X}; \rv{\rnd}_1')}{\rv{f}_1^k(\rv{X})}{2\beta}
        \;\middle|\;
        \begin{aligned}
            (\rv{I}_1, \rv{A}_1, \rv{\rnd}_1) \text{ is $\left(\alpha,\gamma,\tfrac{\epsilon}{2}\right)$-excellent} \\
            (\rv{X} \setminus \rv{I}_1, \rv{\rnd}_1') \in \Cons(\rv{I}_1, \rv{A}_1, \rv{\rnd}_1)
        \end{aligned}
    \right] & < \delta, \\
    \Pr_{\substack{
        (\rv{I}_1, \rv{I}_2, \rv{A}_1, \rv{A}_2, \rv{X}) \sim \Sample \\
        \rv{\rnd}_2, \rv{\rnd}_2' \in \Rnd
    }}\left[
        \far{F(\rv{X}; \rv{\rnd}_2')}{\rv{f}_2^k(\rv{X})}{2\beta}
        \;\middle|\;
        \begin{aligned}
            (\rv{I}_2, \rv{A}_2, \rv{\rnd}_2) \text{ is $\left(\alpha,\gamma,\tfrac{\epsilon}{2}\right)$-excellent} \\
            (\rv{X} \setminus \rv{I}_2, \rv{\rnd}_2') \in \Cons(\rv{I}_2, \rv{A}_2, \rv{\rnd}_1)
        \end{aligned}
    \right] & < \delta.
\end{align}
By \cref{lem:IKW12:clm318}, the conditions in the above formulas happen with probability $\Omega(\epsilon^5)$,
implying that
\begin{align}
\begin{aligned}
    & \Pr_{\substack{
        (\rv{I}_1, \rv{I}_2, \rv{A}_1, \rv{A}_2, \rv{X}) \sim \Sample \\
        \rv{\rnd}_1, \rv{\rnd}_2, \rv{\rnd}_1', \rv{\rnd}_2' \in \Rnd
    }}\left[
        \begin{aligned}
            (\rv{I}_1, \rv{A}_1, \rv{\rnd}_1) \text{ is $\left(\alpha,\gamma,\tfrac{\epsilon}{2}\right)$-excellent} \\
            (\rv{I}_2, \rv{A}_2, \rv{\rnd}_2) \text{ is $\left(\alpha,\gamma,\tfrac{\epsilon}{2}\right)$-excellent}
        \end{aligned}
        \text{ and }
        \begin{aligned}
            \close{F(\rv{X}; \rv{\rnd}_1')}{\rv{f}_1^k(\rv{X})}{2\beta} \\
            \close{F(\rv{X}; \rv{\rnd}_2')}{\rv{f}_2^k(\rv{X})}{2\beta}
        \end{aligned}
    \right] \\
    & \geq \Omega(\epsilon^5) (1 - 2\delta)
    \geq \Omega(\epsilon^5) \bigl(1 - o(\epsilon)\bigr)
    = \Omega(\epsilon^5).
\end{aligned}
\end{align}
Observe that
\begin{align}
    \Bigl(
    \close{F(\rv{X}; \rv{\rnd}_1')}{\rv{f}_1^k(\rv{X})}{2\beta}
    \text{ and }
    \close{F(\rv{X}; \rv{\rnd}_2')}{\rv{f}_2^k(\rv{X})}{2\beta}
    \Bigr)
    \implies 
    \close{\rv{f}_1^k(\rv{X})}{\rv{f}_2^k(\rv{X})}{4\beta},
\end{align}
as desired.
\end{proof}

\begin{claim}[Claim 3.20 of \cite{impagliazzo2012new}]
\label{lem:IKW12:clm320}
\begin{align}
    \Pr_{\substack{
        (\rv{I}_1, \rv{I}_2, \rv{A}_1, \rv{A}_2) \sim \Sample \\
        \rv{\rnd}_1, \rv{\rnd}_2 \in \Rnd
    }}\left[
        \begin{aligned}
            (\rv{I}_1, \rv{A}_1, \rv{\rnd}_1) \text{ is $\left(\alpha,\gamma,\tfrac{\epsilon}{2}\right)$-excellent} \\
            (\rv{I}_2, \rv{A}_2, \rv{\rnd}_2) \text{ is $\left(\alpha,\gamma,\tfrac{\epsilon}{2}\right)$-excellent}
        \end{aligned}
        \text{ and }
        \dist\bigl(\rv{f}_1, \rv{f}_2\bigr) \leq 8\beta
    \right] \geq \Omega(\epsilon^5).
\end{align}
\end{claim}
\begin{proof}
By \cref{lem:IKW12:clm319} and an averaging argument,
we have
\begin{align}
    \Pr_{\substack{
        (\rv{I}_1, \rv{I}_2, \rv{A}_1, \rv{A}_2) \sim \Sample \\
        \rv{\rnd}_1, \rv{\rnd}_2 \in \Rnd
    }}\left[
        \begin{aligned}
            (\rv{I}_1, \rv{A}_1, \rv{\rnd}_1) \text{ is $\left(\alpha,\gamma,\tfrac{\epsilon}{2}\right)$-excellent} \\
            (\rv{I}_2, \rv{A}_2, \rv{\rnd}_2) \text{ is $\left(\alpha,\gamma,\tfrac{\epsilon}{2}\right)$-excellent}
        \end{aligned}
        \text{ and }
        \Pr_{\rv{X} \in \binom{V}{k} : \rv{X} \supset \rv{I}_1 \cup \rv{I}_2}\left[
            \close{\rv{f}_1^k(\rv{X})}{\rv{f}_2^k(\rv{X})}{4\beta}
        \right] \geq \Omega(\epsilon^5)
    \right] \geq \Omega(\epsilon^5),
\end{align}
which implies \cref{lem:IKW12:clm320} immediately for the following reason.
Suppose for contradiction that for fixed $I_1, I_2 \in \binom{V}{\ell}$ with $I_1 \cap I_2 = \emptyset$,
\begin{align}
\begin{aligned}
    & \dist\bigl(\rv{f}_1, \rv{f}_2\bigr)
    > 8\beta, \\
    \implies
    & \dist\bigl(\rv{f}_1|_{V \setminus (I_1 \cup I_2)}, \rv{f}_2|_{V \setminus (I_1 \cup I_2)}\bigr)
    = \Pr_{\rv{v} \in V \setminus (I_1 \cup I_2)}\bigl[
        \rv{f}_1(\rv{v}) \neq \rv{f}_2(\rv{v})
    \bigr]
    > 8\beta - \tfrac{2\ell}{n}
    > 7\beta.
\end{aligned}
\end{align}
By the Chernoff bound, we have
\begin{align}
\begin{aligned}
    \Pr_{\rv{X} \in \binom{V}{k} : \rv{X} \supset I_1 \cup I_2}\Bigl[
        \close{\rv{f}_1^k(\rv{X})}{\rv{f}_2^k(\rv{X})}{4\beta}
    \Bigr]
    & \leq \Pr_{\rv{X} \in \binom{V \setminus(I_1\cup I_2)}{k-2\ell}}\Bigl[
        \close{\rv{f}_1^{k-2\ell}(\rv{X})}{\rv{f}_2^{k-2\ell}(\rv{X})}{5\beta}
    \Bigr] \\
    & \leq \exp\Bigl(-\Omega\bigl(\beta (k-2\ell)\bigr)\Bigr)
    \leq \exp\bigl(-\omega(\ell)\bigr)
    = o(\epsilon^5),
\end{aligned}
\end{align}
which is a contradiction.
\end{proof}

We are now ready to conclude the proof of \cref{thm:app:IKW12}.
\begin{proof}[Proof of \cref{thm:app:IKW12}]
By \cref{lem:IKW12:clm320} and an averaging argument,
there exists a function $f \colon V \to \Sigma$ such that
\begin{align}
    \Pr_{\substack{
        \rv{I}_0 \in \binom{V}{\ell} \\
        (\rv{A}_0, \rv{\rnd}_0) \in \binom{V\setminus\rv{I}_0}{k-\ell} \times \Rnd
    }}\Bigl[
        \underbrace{
            (\rv{I}_0, \rv{A}_0, \rv{\rnd}_0) \text{ is $\left(\alpha,\gamma,\tfrac{\epsilon}{2}\right)$-excellent}
            \text{ and }
            \dist\bigl(\rv{f}_0, f\bigr) \leq 8\beta
        }_{\evt(\rv{I}_0, \rv{A}_0, \rv{\rnd}_0) \defeq}
    \Bigr] \geq \Omega(\epsilon^5).
\end{align}
On the one hand, by \cref{lem:IKW12:lem38}, we have
\begin{align}
    \Pr_{\substack{
        \rv{I}_0 \in \binom{V}{\ell} \\
        (\rv{A}_0, \rv{\rnd}_0) \in \binom{V\setminus\rv{I}_0}{k-\ell} \times \Rnd \\
        (\rv{A}, \rv{\rnd}) \in \Cons(\rv{I}_0, \rv{A}_0) \times \Rnd
    }}\Bigl[
        \far{F(\rv{I}_0 \cup \rv{A}; \rv{\rnd})}{\rv{f}_0^k(\rv{I}_0 \cup \rv{A})}{2\beta}
        \Bigm|
        \evt(\rv{I}_0, \rv{A}_0, \rv{\rnd}_0)
    \Bigr]
    & < \delta
    = o(1).
\end{align}
On the other hand,
since $\dist\bigl(\rv{f}_0, f\bigr) \leq 8\beta$,
by applying the Chernoff bound, we have
\begin{align}
\begin{aligned}
    & \Pr_{\substack{
        \rv{I}_0 \in \binom{V}{\ell} \\
        (\rv{A}_0, \rv{\rnd}_0) \in \binom{V\setminus\rv{I}_0}{k} \times \Rnd \\
        (\rv{A}, \rv{\rnd}) \in \Cons(\rv{I}_0, \rv{A}_0) \times \Rnd
    }}\Bigl[
        \far{\rv{f}_0^k(\rv{I}_0 \cup \rv{A})}{f^k(\rv{I}_0 \cup \rv{A})}{16\beta}
        \Bigm|
        \evt(\rv{I}_0, \rv{A}_0, \rv{\rnd}_0)
    \Bigr] \\
    & = \frac{
        \displaystyle
        \Pr_{\substack{
            \rv{I}_0, \rv{A}_0, \rv{\rnd}_0 \\
            (\rv{A}, \rv{\rnd}) \in \binom{V\setminus\rv{I}_0}{k} \times \Rnd
        }}\Bigl[
            \far{\rv{f}_0^k(\rv{I}_0 \cup \rv{A})}{f^k(\rv{I}_0 \cup \rv{A})}{16\beta}
            \text{ and }
            \underbrace{\rv{A} \in \Cons(\rv{I}_0, \rv{A}_0, \rv{\rnd}_0)}_{\text{unnecessary}}
            \Bigm|
            \evt(\rv{I}_0, \rv{A}_0, \rv{\rnd}_0)
        \Bigr]
    }{
        \displaystyle
        \Pr_{\substack{
            \rv{I}_0, \rv{A}_0, \rv{\rnd}_0 \\
            (\rv{A}, \rv{\rnd}) \in \binom{V\setminus\rv{I}_0}{k} \times \Rnd
        }}\Bigl[
            (\rv{A}, \rv{\rnd}) \in \Cons(\rv{I}_0, \rv{A}_0, \rv{\rnd}_0) \times \Rnd
            \Bigm|
            \evt(\rv{I}_0, \rv{A}_0, \rv{\rnd}_0)
        \Bigr]
    } \\
    & < \frac{\exp\bigl(-\Omega(\beta k)\bigr)}{\frac{\epsilon}{2}}
    = o(1).
\end{aligned}
\end{align}
By applying the union bound, we have
\begin{align}
    \Pr_{\substack{
        \rv{I}_0, \rv{A}_0, \rv{\rnd}_0 \\
        (\rv{A}, \rv{\rnd}) \in \Cons(\rv{I}_0, \rv{A}_0) \times \Rnd
    }}\left[
        \begin{aligned}
            \close{F(\rv{I}_0 \cup \rv{A}; \rv{\rnd})}{\rv{f}_0^k(\rv{I}_0 \cup \rv{A})}{2\beta} \\
            \close{\rv{f}_0^k(\rv{I}_0 \cup \rv{A})}{f^k(\rv{I}_0 \cup \rv{A})}{16\beta}
        \end{aligned}
        \;\middle|\;
        \evt(\rv{I}_0, \rv{A}_0, \rv{\rnd}_0)
    \right]
    \geq 1-o(1)
    \geq \Omega(1).
\end{align}

Since $\frac{|\Cons(\rv{I}_0, \rv{A}_0, \rv{\rnd}_0)|}{\bigl|\binom{V}{k-\ell}_{\rv{I}_0} \times \Rnd\bigr|} \geq \frac{\epsilon}{2}$
whenever $\evt(\rv{I}_0, \rv{A}_0, \rv{\rnd}_0)$ holds, we have
\begin{align}
\begin{aligned}
    & \Pr_{\substack{
        \rv{I}_0, \rv{A}_0, \rv{\rnd}_0 \\
        (\rv{A}, \rv{\rnd}) \in \binom{V\setminus\rv{I}_0}{k-\ell} \times \Rnd
    }}\Bigl[
        \close{F(\rv{I}_0 \cup \rv{A}; \rv{\rnd})}{f^k(\rv{I}_0 \cup \rv{A})}{18\beta}
        \Bigm|
        \evt(\rv{I}_0, \rv{A}_0, \rv{\rnd}_0)
    \Bigr]
    \geq \Omega(1) \cdot \tfrac{\epsilon}{2}
    \geq \Omega(\epsilon), \\
    \implies & \Pr_{\substack{
        \rv{I}_0, \rv{A}_0, \rv{\rnd}_0 \\
        (\rv{A}, \rv{\rnd}) \in \binom{V\setminus\rv{I}_0}{k-\ell} \times \Rnd
    }}\Bigl[
        \close{F(\rv{I}_0 \cup \rv{A}; \rv{\rnd})}{f^k(\rv{I}_0 \cup \rv{A})}{18\beta}
    \Bigr]
    \geq \Omega(\epsilon^6), \\
    \implies & \Pr_{\rv{X} \in \binom{V}{k}}\Bigl[
        \close{F(\rv{X})}{f^k(\rv{X})}{18\beta}
    \Bigr] \geq \Omega(\epsilon^6),
\end{aligned}
\end{align}
as desired.
\end{proof}

\section{\texorpdfstring{%
Omitted Proofs in \cref{sec:PSPACE}
}{%
Omitted Proofs in Section \ref{sec:PSPACE}
}}
\label{app:PSPACE}

\begin{proof}[Proof of \cref{clm:PSPACE:soundness:hitting}]
Let $J \defeq J(n,k,\ell)$.
For a vertex $k$-tuple $\sq{x} = (x_1, \ldots, x_k) \in V^k$ and
an integer $k$-tuple $\sq{a} = (a_1, \ldots, a_k) \in [\Delta]^k$,
we use $E^k(\sq{x})[\sq{a}]$ to denote the edge $k$-tuple $(e_1, \ldots, e_k) \in E^k$ such that
each $e_i$ is the \nth{$a_i$} edge incident to $x_i$.

Define $\calM$ as the set of edge $k$-tuples in $\calS$ that form a size-$k$ matching; namely,
\begin{align}
    \calM \defeq \Bigl\{
        \sq{e} = (e_1, \ldots, e_k) \in \calS
        \Bigm|
        e_i \cap e_j = \emptyset \text{ for every } \{i,j\} \in \tbinom{[k]}{2}
    \Bigr\}.
\end{align}
Note that $|\calM| \leq |\calS|$.
We write 
\begin{align}
    \bigl(
        \sq{\rv{x}}_1, \ldots, \sq{\rv{x}}_q,
        \sq{\rv{e}}_1, \ldots, \sq{\rv{e}}_q,
        \rv{\pi}_1, \ldots, \rv{\pi}_q
    \bigr) \sim \W_q
\end{align}
for the random variables selected by $\W_q$.
Observe first that
\begin{align}
\label{eq:PSPACE:soundness:hitting:1}
\begin{aligned}
    & \Pr_{(\sq{\rv{e}}_1, \ldots, \sq{\rv{e}}_q, \rv{\pi}_1, \ldots, \rv{\pi}_q) \sim \W_q}\Biggl[
        \bigwedge_{i \in [q]} \sq{\rv{e}}_i \circ \rv{\pi}_i \in \calM
    \Biggr] \\
    & = \Pr_{(\sq{\rv{e}}_1, \ldots, \sq{\rv{e}}_q, \rv{\pi}_1, \ldots, \rv{\pi}_q) \sim \W_q}\Biggl[
        \bigwedge_{i \in [q]} \sq{\rv{e}}_i \circ \rv{\pi}_i \in \calS
        \text{ and }
        \bigwedge_{i \in [q]} \sq{\rv{e}}_i \circ \rv{\pi}_i \text{ is a matching}
    \Biggr] \\
    & \geq \Pr_{(\sq{\rv{e}}_1, \ldots, \sq{\rv{e}}_q, \rv{\pi}_1, \ldots, \rv{\pi}_q) \sim \W_q}\Biggl[
        \bigwedge_{i \in [q]} \sq{\rv{e}}_i \circ \rv{\pi}_i \in \calS
    \Biggr]
    - \Pr_{(\sq{\rv{e}}_1, \ldots, \sq{\rv{e}}_q, \rv{\pi}_1, \ldots, \rv{\pi}_q) \sim \W_q}\Biggl[
        \bigvee_{i \in [q]}
        \sq{\rv{e}}_i \circ \rv{\pi}_i \text{ is not a matching}
    \Biggr] \\
    & \geq \Pr_{(\sq{\rv{e}}_1, \ldots, \sq{\rv{e}}_q, \rv{\pi}_1, \ldots, \rv{\pi}_q) \sim \W_q}\Biggl[
        \bigwedge_{i \in [q]} \sq{\rv{e}}_i \circ \rv{\pi}_i \in \calS
    \Biggr]
    - \sum_{i \in [q]}
    \underbrace{\Pr_{(\sq{\rv{e}}_1, \ldots, \sq{\rv{e}}_q, \rv{\pi}_1, \ldots, \rv{\pi}_q) \sim \W_q}\bigl[
        \sq{\rv{e}}_i \circ \rv{\pi}_i \text{ is not a matching}
    \bigr]}_{\leq \frac{2k^2}{n}} \\
    & \geq \Pr_{(\sq{\rv{e}}_1, \ldots, \sq{\rv{e}}_q, \rv{\pi}_1, \ldots, \rv{\pi}_q) \sim \W_q}\Biggl[
        \bigwedge_{i \in [q]} \sq{\rv{e}}_i \circ \rv{\pi}_i \in \calS
    \Biggr] - \frac{2qk^2}{n},
\end{aligned}
\end{align}
where we used the inequality that
\begin{align}
    \Pr_{\sq{\rv{e}} = (\rv{e}_1, \ldots, \rv{e}_k) \in E^k}\bigl[
        \sq{\rv{e}} \text{ is not a matching}
    \bigr]
    \leq \sum_{\{i,j\} \in \binom{[k]}{2}}
    \Pr_{\sq{\rv{e}} = (\rv{e}_1, \ldots, \rv{e}_k) \in E^k}\bigl[
        \rv{e}_i \cap \rv{e}_j \neq \emptyset
    \bigr]
    \leq \binom{k}{2} \frac{4}{n}
    \leq \frac{2k^2}{n}.
\end{align}

We write $(\rv{X}_1, \ldots, \rv{X}_q) \sim \T_q(\ell,1)$ for $q$ vertex sets selected by $\T_q(\ell,1)$.
Observe next that
\begin{align}
\label{eq:PSPACE:soundness:hitting:2}
\begin{aligned}
    & \Pr_{\substack{
        (\rv{X}_1, \ldots, \rv{X}_q) \sim \T_q(\ell,1) \\
        \rv{\pi}_1, \ldots, \rv{\pi}_q \in \sym_k \\
        \sq{\rv{a}}_1, \ldots, \sq{\rv{a}}_q \in [\Delta]^k
    }}\Biggl[
        \bigwedge_{i \in [q]}
        E^k(\rv{X}_i \circ \rv{\pi}_i)[\sq{\rv{a}}_i] \in \calM
    \Biggr] \\
    & = \Pr_{(\sq{\rv{x}}_1, \ldots, \sq{\rv{x}}_q, \sq{\rv{e}}_1, \ldots, \sq{\rv{e}}_q, \rv{\pi}_1, \ldots, \rv{\pi}_q) \sim \W_q}\Biggl[
        \bigwedge_{i \in [q]}
        \sq{\rv{e}}_i \circ \rv{\pi}_i \in \calM
        \Biggm|
        \bigwedge_{i \in [q]}
        \sq{\rv{x}}_i \in V^{\ul{k}}
    \Biggr] \\
    & \geq \Pr_{(\sq{\rv{x}}_1, \ldots, \sq{\rv{x}}_q, \sq{\rv{e}}_1, \ldots, \sq{\rv{e}}_q, \rv{\pi}_1, \ldots, \rv{\pi}_q) \sim \W_q}\Biggl[
        \bigwedge_{i \in [q]}
        \sq{\rv{e}}_i \circ \rv{\pi}_i \in \calM
        \text{ and }
        \bigwedge_{i \in [q]}
        \sq{\rv{x}}_i \in V^{\ul{k}}
    \Biggr] \\
    & \geq \Pr_{(\sq{\rv{e}}_1, \ldots, \sq{\rv{e}}_q, \rv{\pi}_1, \ldots, \rv{\pi}_q) \sim \W_q}\Biggl[
        \bigwedge_{i \in [q]}
        \sq{\rv{e}}_i \circ \rv{\pi}_i \in \calM
    \Biggr]
    - \Pr_{(\sq{\rv{x}}_1, \ldots, \sq{\rv{x}}_q) \sim \W_q}\Biggl[
        \bigvee_{i \in [q]} \sq{\rv{x}}_i \notin V^{\ul{k}}
    \Bigg] \\
    & \geq \Pr_{(\sq{\rv{e}}_1, \ldots, \sq{\rv{e}}_q, \rv{\pi}_1, \ldots, \rv{\pi}_q) \sim \W_q}\Biggl[
        \bigwedge_{i \in [q]}
        \sq{\rv{e}}_i \circ \rv{\pi}_i \in \calM
    \Biggr] - \sum_{i \in [q]}
        \underbrace{
        \Pr_{(\sq{\rv{x}}_1, \ldots, \sq{\rv{x}}_q) \sim \W_q}\bigl[
            \sq{\rv{x}}_i \notin V^{\ul{k}}
        \bigr]}_{\leq \frac{k^2}{n}} \\
    & \geq \Pr_{(\sq{\rv{e}}_1, \ldots, \sq{\rv{e}}_q, \rv{\pi}_1, \ldots, \rv{\pi}_q) \sim \W_q}\Biggl[
        \bigwedge_{i \in [q]}
        \sq{\rv{e}}_i \circ \rv{\pi}_i \in \calM
    \Biggr] - \frac{qk^2}{n}.
\end{aligned}
\end{align}

Observe further that
\begin{align}
\label{eq:PSPACE:soundness:hitting:3}
\begin{aligned}
    & \Pr_{\substack{
        (\rv{X}_1, \ldots, \rv{X}_q) \sim \RW_q(J) \\
        \rv{\pi}_1, \ldots, \rv{\pi}_q \in \sym_k \\
        \sq{\rv{a}}_1, \ldots, \sq{\rv{a}}_q \in [\Delta]^k
    }}\Biggl[
        \bigwedge_{i \in [q]}
        E^k(\rv{X}_i \circ \rv{\pi}_i)[\sq{\rv{a}}_i] \in \calM
    \Biggr] \\
    & = \Pr_{\substack{
        (\rv{X}_1, \ldots, \rv{X}_q) \sim \T_q(\ell,1) \\
        \rv{\pi}_1, \ldots, \rv{\pi}_q \in \sym_k \\
        \sq{\rv{a}}_1, \ldots, \sq{\rv{a}}_q \in [\Delta]^k
    }}\Biggl[
        \bigwedge_{i \in [q]}
        E^k(\rv{X}_i \circ \rv{\pi}_i)[\sq{\rv{a}}_i] \in \calM
        \Biggm|
        \bigwedge_{i \in [q-1]}
        |\rv{X}_i \cap \rv{X}_{i+1}| = \ell
    \Biggr] \\
    & \geq \Pr_{\substack{
        (\rv{X}_1, \ldots, \rv{X}_q) \sim \T_q(\ell,1) \\
        \rv{\pi}_1, \ldots, \rv{\pi}_q \in \sym_k \\
        \sq{\rv{a}}_1, \ldots, \sq{\rv{a}}_q \in [\Delta]^k
    }}\Biggl[
        \bigwedge_{i \in [q]}
        E^k(\rv{X}_i \circ \rv{\pi}_i)[\sq{\rv{a}}_i] \in \calM
        \text{ and }
        \bigwedge_{i \in [q-1]}
        |\rv{X}_i \cap \rv{X}_{i+1}| = \ell
    \Biggr] \\
    & \geq \Pr_{\substack{
        (\rv{X}_1, \ldots, \rv{X}_q) \sim \T_q(\ell,1) \\
        \rv{\pi}_1, \ldots, \rv{\pi}_q \in \sym_k \\
        \sq{\rv{a}}_1, \ldots, \sq{\rv{a}}_q \in [\Delta]^k
    }}\Biggl[
        \bigwedge_{i \in [q]}
        E^k(\rv{X}_i \circ \rv{\pi}_i)[\sq{\rv{a}}_i] \in \calM
    \Biggr]
    - \Pr_{(\rv{X}_1, \ldots, \rv{X}_q) \sim \T_q(\ell,1)}\Biggl[
        \bigvee_{i \in [q-1]} |\rv{X}_i \cap \rv{X}_{i+1}| \neq \ell
    \Biggr]
    \\
    & \geq \Pr_{\substack{
        (\rv{X}_1, \ldots, \rv{X}_q) \sim \T_q(\ell,1) \\
        \rv{\pi}_1, \ldots, \rv{\pi}_q \in \sym_k \\
        \sq{\rv{a}}_1, \ldots, \sq{\rv{a}}_q \in [\Delta]^k
    }}\Biggl[
        \bigwedge_{i \in [q]}
        E^k(\rv{X}_i \circ \rv{\pi}_i)[\sq{\rv{a}}_i] \in \calM
    \Biggr]
    - \sum_{i \in [q-1]}
        \underbrace{
            \Pr_{(\rv{X}_1, \ldots, \rv{X}_q) \sim \T_q(\ell,1)}\bigl[|\rv{X}_i \cap \rv{X}_{i+1}| \neq \ell\bigr]
        }_{\leq \frac{k^2}{n}}
    \\
    & \geq \Pr_{\substack{
        (\rv{X}_1, \ldots, \rv{X}_q) \sim \T_q(\ell,1) \\
        \rv{\pi}_1, \ldots, \rv{\pi}_q \in \sym_k \\
        \sq{\rv{a}}_1, \ldots, \sq{\rv{a}}_q \in [\Delta]^k
    }}\Biggl[
        \bigwedge_{i \in [q]}
        E^k(\rv{X}_i \circ \rv{\pi}_i)[\sq{\rv{a}}_i] \in \calM
    \Biggr] - \frac{qk^2}{n}.
\end{aligned}
\end{align}

Define
\begin{align}
\begin{aligned}
    \tilde{\calM} & \defeq \Bigl\{
        (X, \pi, \sq{a}) \in V(J \times \sym_k \times [\Delta]^k)
        \Bigm|
        E^k(X \circ \pi)[\sq{a}] \in \calM
    \Bigr\}.
\end{aligned}
\end{align}

Then, we obtain
\begin{align}
\label{eq:PSPACE:soundness:hitting:4}
\begin{aligned}
    \Pr_{\substack{
        (\rv{X}_1, \ldots, \rv{X}_q) \sim \RW_q(J) \\
        \rv{\pi}_1, \ldots, \rv{\pi}_q \in \sym_k \\
        \sq{\rv{a}}_1, \ldots, \sq{\rv{a}}_q \in [\Delta]^k
    }}\Biggl[
        \bigwedge_{i \in [q]}
        E^k(\rv{X}_i \circ \rv{\pi}_i)[\sq{\rv{a}}_i] \in \calM
    \Biggr]
    = \Pr_{\left((\rv{X}_1, \rv{\pi}_1, \sq{\rv{a}}_1), \ldots (\rv{X}_q, \rv{\pi}_q, \sq{\rv{a}}_q)\right) \sim \RW_q(J \times \sym_k \times [\Delta]^k)}\Biggl[
       \bigwedge_{i \in [q]} 
       (\rv{X}_i, \rv{\pi}_i, \sq{\rv{a}}_i) \in \tilde{\calM}
    \Biggr].
\end{aligned}
\end{align}

Let $\sq{e} = (e_1, \ldots, e_k) \in E^k$ be an edge $k$-tuple that forms a size-$k$ matching.
Then, there are exactly $2^k$ triples $(X, \pi, \sq{a}) \in V(J \times \sym_k \times [\Delta]^k)$
such that $E^k(X \circ \pi)[\sq{a}] = \sq{e}$.
Since $\calM$ consists only of size-$k$ matchings, $|\tilde{\calM}| = 2^k |\calM|$.
Since
\begin{align}
\begin{aligned}
    |V(J \times \sym_k \times [\Delta]^k)|
        = \binom{n}{k} k! \Delta^k
    \text{ and }
    |E^k|
        = \left(\frac{n\Delta}{2}\right)^k,
\end{aligned}
\end{align}
we have
\begin{align}
\label{eq:PSPACE:soundness:hitting:5}
\begin{aligned}
    \frac{|\tilde{\calM}|}{|V(J \times \sym_k \times [\Delta]^k)|}
    & = \frac{|E^k|}{|V(J \times \sym_k \times [\Delta]^k)|} \frac{|\tilde{\calM}|}{|E^k|} \\
    & = \frac{\left(\frac{n\Delta}{2}\right)^k}{\binom{n}{k} k! \Delta^k}
        \frac{2^k|\calM|}{|E^k|}
    = \frac{n^k}{n^{\ul{k}}} \cdot \frac{|\calM|}{|E^k|}
    = \frac{|\calM|}{|E^k|} \left(1 + \bigO_k\left(\frac{1}{n}\right)\right)
    \leq \frac{|\calS|}{|E^k|} \left(1 + \bigO_k\left(\frac{1}{n}\right)\right),
\end{aligned}
\end{align}
where we used the inequality that
\begin{align}
    \frac{n^k}{n^{\ul{k}}}
    \leq \frac{n^k}{(n-k)^k}
    = \left(1+\frac{k}{n-k}\right)^k
    = 1 + \bigO\left(\frac{k^2}{n}\right)
    \qquad \text{(as } n = \omega(k^2) \text{)}.
\end{align}

By applying \cref{lem:Johnson:hitting} to $\tilde{\calM}$,
we have
\begin{align}
\label{eq:PSPACE:soundness:hitting:6}
\begin{aligned}
    & \Pr_{\left((\rv{X}_1, \rv{\pi}_1, \sq{\rv{a}}_1), \ldots (\rv{X}_q, \rv{\pi}_q, \sq{\rv{a}}_q)\right) \sim \RW_q(J \times \sym_k \times [\Delta]^k)}\Biggl[
       \bigwedge_{i \in [q]} 
       (\rv{X}_i, \rv{\pi}_i, \sq{\rv{a}}_i) \in \tilde{\calM}
    \Biggr] \\
    & \leq \left(\frac{|\tilde{\calM}|}{|V(J \times \sym_k \times [\Delta]^k)|}\right)^q
        + \bigO_q\left(\frac{\ell}{k}\right)
            \left(\frac{|\tilde{\calM}|}{|V(J \times \sym_k \times [\Delta]^k)|}\right) \\
    & \underbrace{\leq}_{\text{\cref{eq:PSPACE:soundness:hitting:5}}}
        \left(\frac{|\calS|}{|E^k|}\right)^q \left(1+\bigO_k\left(\frac{1}{n}\right)\right)^q
        + \bigO_q\left(\frac{\ell}{k}\right)
            \left(\frac{|\calS|}{|E^k|}\right) \left(1+\bigO_k\left(\frac{1}{n}\right)\right) \\
    & = \left(\frac{|\calS|}{|E^k|}\right)^q
        + \bigO_q\left(\frac{\ell}{k}\right) \left(\frac{|\calS|}{|E^k|}\right)
        + \bigO_{q,k}\left(\frac{1}{n}\right).
\end{aligned}
\end{align}
Combining 
\cref{eq:PSPACE:soundness:hitting:1,eq:PSPACE:soundness:hitting:2,eq:PSPACE:soundness:hitting:3,eq:PSPACE:soundness:hitting:4,eq:PSPACE:soundness:hitting:6}, we obtain
\begin{align}
\begin{aligned}
    & \Pr_{(\sq{\rv{e}}_1, \ldots, \sq{\rv{e}}_q, \rv{\pi}_1, \ldots, \rv{\pi}_q) \sim \W_q}\Biggl[
        \bigwedge_{i \in [q]} \sq{\rv{e}}_i \circ \rv{\pi}_i \in \calS
    \Biggr] \\
    & \leq \frac{2qk^2}{n} + \frac{qk^2}{n} + \frac{qk^2}{n}
        + \left(\frac{|\calS|}{|E^k|}\right)^q
        + \bigO_q\left(\frac{\ell}{k}\right) \left(\frac{|\calS|}{|E^k|}\right)
        + \bigO_{k,q}\left(\frac{1}{n}\right) \\
    & \leq \left(\frac{|\calS|}{|E^k|}\right)^q
        + \bigO_q\left(\frac{\ell}{k}\right) \left(\frac{|\calS|}{|E^k|}\right)
        + \bigO_{k,q}\left(\frac{1}{n}\right),
\end{aligned}
\end{align}
as desired.
\end{proof}

\begin{proof}[Proof of \cref{clm:PSPACE:W:Ttuple}]
For
$\rv{I}_i \in \binom{[k]}{\ell}$,
$\sq{\rv{x}}_i, \sq{\rv{x}}_{i+1} \in V^k$ with $\sq{\rv{x}}_i|_{\rv{I}_i} = \sq{\rv{x}}_{i+1}|_{\rv{I}_i}$,
$\sq{\rv{e}}_i \in E^k(\sq{\rv{x}}_i)$,
$\sq{\rv{e}}_{i+1} \in E^k(\sq{\rv{x}}_{i+1})$, and
$\rv{\pi}_i, \rv{\pi}_{i+1} \in \sym_k$,
define the event
$\evt(\rv{I}_i, \sq{\rv{x}}_i, \sq{\rv{e}}_i, \rv{\pi}_i, \sq{\rv{x}}_{i+1}, \sq{\rv{e}}_{i+1}, \rv{\pi}_{i+1})$
as
\begin{align}
    \close{
        \bigl(F(\sq{\rv{e}}_i \circ \rv{\pi}_i) \circ \rv{\pi}_i^{-1}\bigr)\bigl[\sq{\rv{x}}_i|_{\rv{I}_i}\bigr]}{
        \bigl(F(\sq{\rv{e}}_{i+1} \circ \rv{\pi}_{i+1}) \circ \rv{\pi}_{i+1}^{-1}\bigr)\bigl[\sq{\rv{x}}_{i+1}|_{\rv{I}_i}\bigr]}{
        1/\ell
    }.
\end{align}
By definition,
$\Ttuple(\ell,1)$ accepts $\Ftuple$ with the following probability:
\begin{align}
\begin{aligned}
    & \Pr\Bigl[\Ttuple^{\Ftuple}(\ell,1) = 1\Bigr] \\
    & = \Pr_{\substack{
        \rv{I}_1, \ldots, \rv{I}_{q-1} \in \binom{[k]}{\ell} \\
        \sq{\rv{x}}_1, \ldots, \sq{\rv{x}}_q \in V^k \\
        \sq{\rv{a}}_1, \ldots, \sq{\rv{a}}_q \in [\Delta]^k \\
        \rv{\pi}_1, \ldots, \rv{\pi}_q \in \sym_k
    }}\Biggl[
        \bigwedge_{i \in [q-1]}
        \close{
            \bigl(\Ftuple(\sq{\rv{x}}_i \circ \rv{\pi}_i, \sq{\rv{a}}_i) \circ \rv{\pi}_i^{-1}\bigr)|_{\rv{I}_i}}{
            \bigl(\Ftuple(\sq{\rv{x}}_{i+1} \circ \rv{\pi}_{i+1}, \sq{\rv{a}}_{i+1}) \circ \rv{\pi}_{i+1}^{-1}\bigr)|_{\rv{I}_i}}{
            1/\ell}
        \Biggm|
        \bigwedge_{i \in [q-1]}
        \sq{\rv{x}}_i|_{\rv{I}_i} = \sq{\rv{x}}_{i+1}|_{\rv{I}_i}
    \Biggr] \\
    & = \Pr_{\substack{
        \rv{I}_1, \ldots, \rv{I}_{q-1} \in \binom{[k]}{\ell} \\
        \sq{\rv{x}}_1, \ldots, \sq{\rv{x}}_q \in V^k \\
        \sq{\rv{e}}_i \in E^k(\sq{\rv{x}}_i) \; \forall i \in [q]\\
        \rv{\pi}_1, \ldots, \rv{\pi}_q \in \sym_k
    }}\Biggl[
        \bigwedge_{i \in [q-1]}
        \evt(\rv{I}_i, \sq{\rv{x}}_i, \sq{\rv{e}}_i, \rv{\pi}_i, \sq{\rv{x}}_{i+1}, \sq{\rv{e}}_{i+1}, \rv{\pi}_{i+1})
        \Biggm|
        \bigwedge_{i \in [q-1]}
        \sq{\rv{x}}_i|_{\rv{I}_i} = \sq{\rv{x}}_{i+1}|_{\rv{I}_i}
    \Biggr].
\end{aligned}
\end{align}
By definition, $\W_q$ accepts $F$ with the following probability:
\begin{align}
\begin{aligned}
    & \Pr\Bigl[\W_q^F = 1\Bigr] \\
    & = \Pr_{\substack{
        \rv{I}_1, \ldots, \rv{I}_{q-1} \in \binom{[k]}{\ell} \\
        \sq{\rv{x}}_1, \ldots, \sq{\rv{x}}_q \in V^k \\
        \sq{\rv{e}}_i \in E^k(\sq{\rv{x}}_i) \; \forall i \in [q]\\
        \rv{\pi}_1, \ldots, \rv{\pi}_q \in \sym_k
    }}\Biggl[
        \bigwedge_{i \in [q-1]}
        \begin{aligned}
            & \evt(\rv{I}_i, \sq{\rv{x}}_i, \sq{\rv{e}}_i, \rv{\pi}_i, \sq{\rv{x}}_{i+1}, \sq{\rv{e}}_{i+1}, \rv{\pi}_{i+1}) \text{ or} \\
            & |\{ \rv{x}_{i,1}, \ldots, \rv{x}_{i,k}, \rv{x}_{i+1,1}, \ldots, \rv{x}_{i+1,k} \}| < 2k-\ell
        \end{aligned}
        \Biggm|
        \bigwedge_{i \in [q-1]}
        \sq{\rv{x}}_i|_{\rv{I}_i} = \sq{\rv{x}}_{i+1}|_{\rv{I}_i}
   \Biggr] \\
    & \leq \Pr_{\substack{
        \rv{I}_1, \ldots, \rv{I}_{q-1} \in \binom{[k]}{\ell} \\
        \sq{\rv{x}}_1, \ldots, \sq{\rv{x}}_q \in V^k \\
        \sq{\rv{e}}_i \in E^k(\sq{\rv{x}}_i) \; \forall i \in [q]\\
        \rv{\pi}_1, \ldots, \rv{\pi}_q \in \sym_k
    }}\Biggl[
        \bigwedge_{i \in [q-1]}
        \evt(\rv{I}_i, \sq{\rv{x}}_i, \sq{\rv{e}}_i, \rv{\pi}_i, \sq{\rv{x}}_{i+1}, \sq{\rv{e}}_{i+1}, \rv{\pi}_{i+1})
        \Biggm|
        \bigwedge_{i \in [q-1]}
        \sq{\rv{x}}_i|_{\rv{I}_i} = \sq{\rv{x}}_{i+1}|_{\rv{I}_i}
    \Biggr] \\
    & \qquad + \Pr_{\substack{
        \rv{I}_1, \ldots, \rv{I}_{q-1} \in \binom{[k]}{\ell} \\
        \sq{\rv{x}}_1, \ldots, \sq{\rv{x}}_q \in V^k \\
        \sq{\rv{e}}_i \in E^k(\sq{\rv{x}}_i) \; \forall i \in [q]\\
        \rv{\pi}_1, \ldots, \rv{\pi}_q \in \sym_k
    }}\Biggl[
        \bigvee_{i \in [q-1]}
        |\{ \rv{x}_{i,1}, \ldots, \rv{x}_{i,k}, \rv{x}_{i+1,1}, \ldots, \rv{x}_{i+1,k} \}| < 2k-\ell
        \Biggm|
        \bigwedge_{i \in [q-1]}
        \sq{\rv{x}}_i|_{\rv{I}_i} = \sq{\rv{x}}_{i+1}|_{\rv{I}_i}
    \Biggr] \\
    & \leq \Pr\Bigl[\Ttuple^{\Ftuple}(\ell,1) = 1\Bigr] \\
    & \qquad + \sum_{i \in [q-1]}
    \underbrace{
        \Pr_{\substack{
            \rv{I}_1, \ldots, \rv{I}_{q-1} \in \binom{[k]}{\ell} \\
            \sq{\rv{x}}_1, \ldots, \sq{\rv{x}}_q \in V^k \\
            \sq{\rv{e}}_i \in E^k(\sq{\rv{x}}_i) \; \forall i \in [q]\\
            \rv{\pi}_1, \ldots, \rv{\pi}_q \in \sym_k
        }}\Bigl[
            |\{ \rv{x}_{i,1}, \ldots, \rv{x}_{i,k}, \rv{x}_{i+1,1}, \ldots, \rv{x}_{i+1,k} \}| < 2k-\ell
            \Bigm|
            \sq{\rv{x}}_i|_{\rv{I}_i} = \sq{\rv{x}}_{i+1}|_{\rv{I}_i}
        \Bigr]
    }_{\leq \frac{k^2}{n}} \\
    & \leq \Pr\Bigl[\Ttuple^{\Ftuple}(\ell,1) = 1\Bigr] + \frac{qk^2}{n},
\end{aligned}
\end{align}

Consequently, we have
\begin{align}
    \Pr\Bigl[\Ttuple^{\Ftuple}(\ell,1) = 1\Bigr]
    \geq \Pr\Bigl[\W_q^F = 1\Bigr] - \frac{qk^2}{n}
    \geq p - \frac{qk^2}{n},
\end{align}
as desired.
\end{proof}

\paragraph{Proof of \cref{clm:PSPACE:W:Tset}.}
Instead of directly proving \cref{clm:PSPACE:W:Tset}, we prove the following slight generalization.
\begin{lemma}
\label{lem:tuple-set:accept}
    Let $\Ftuple \colon V^k \times \Rnd \to \Sigma^k$ be a $k$-tuple function, and
    let $\Fset \colon \binom{V}{k} \times \Rnd \times \sym_k \to \Sigma^k$ be a $k$-set function such that
    $\Fset(X; \rnd, \pi) \defeq \Ftuple(X \circ \pi; \rnd) \circ \pi^{-1}$.
    If $\Ttuple(\ell,m)$ accepts $\Ftuple$ with probability at least $p$,
    then $\T_q(\ell,m)$ accepts $\Fset$ with probability at least $p-\frac{qk^2}{|V|}$.
\end{lemma}
\begin{proof} 
For $\rv{I}_i \in \binom{[k]}{\ell}$,
$\sq{\rv{x}}_i, \sq{\rv{x}}_{i+1} \in V^k$,
$\rv{\pi}_i, \rv{\pi}_{i+1} \in \sym_k$,
$\rv{\rnd}_i, \rv{\rnd}_{i+1} \in \Rnd$,
define the event
$\evt_{\text{tuple}}(\rv{I}_i, \sq{\rv{x}}_i, \rv{\pi}_i, \rv{\rnd}_i, \sq{\rv{x}}_{i+1}, \rv{\pi}_{i+1}, \rv{\rnd}_{i+1})$ as
\begin{align}
    \close{
    \bigl(\Ftuple(\sq{\rv{x}}_i \circ \rv{\pi}_i; \rv{\rnd}_i) \circ \rv{\pi}_i^{-1}\bigr)|_{\rv{I}_i}}{
    \bigl(\Ftuple(\sq{\rv{x}}_{i+1} \circ \rv{\pi}_{i+1}; \rv{\rnd}_{i+1}) \circ \rv{\pi}_{i+1}^{-1}\bigr)|_{\rv{I}_i}}{
    m/\ell}.
\end{align}
By assumption, we have
\begin{align}
\begin{aligned}
    & p
    \leq \Pr\Bigl[\Ttuple^{\Ftuple} = 1\Bigr] \\ 
    & = \Pr_{\substack{
        \rv{I}_1, \ldots, \rv{I}_{q-1} \in \binom{[k]}{\ell} \\
        \sq{\rv{x}}_1, \ldots, \sq{\rv{x}}_q \in V^k \\
        \rv{\pi}_1, \ldots, \rv{\pi}_q \in \sym_k \\
        \rv{\rnd}_1, \ldots, \rv{\rnd}_q \in \Rnd
    }}\Biggl[
        \bigwedge_{i \in [q-1]}
        \evt_{\text{tuple}}(\rv{I}_i, \sq{\rv{x}}_i, \rv{\pi}_i, \rv{\rnd}_i, \sq{\rv{x}}_{i+1}, \rv{\pi}_{i+1}, \rv{\rnd}_{i+1})
        \Biggm|
        \bigwedge_{i \in [q-1]}
        \sq{\rv{x}}_i|_{\rv{I}_i} = \sq{\rv{x}}_{i+1}|_{\rv{I}_i}
    \Biggr].
\end{aligned}
\end{align}
We write
\begin{align}
    \Bigl(
        \sq{\rv{I}} = (\rv{I}_i)_{i \in [q-1]}, \sq{\rv{X}} = (\rv{X}_i)_{i \in [q]},
        \sq{\rv{\rnd}} = (\rv{\rnd}_i)_{i \in [q]}, \sq{\rv{\pi}} = (\rv{\pi}_i)_{i \in [q]}
    \Bigr)
    \sim \T_q(\ell,m)
\end{align}
for the random variables selected by $\T_q(\ell,m)$.
By definition, $\T_q(\ell,m)$ accepts $F$ with the following probability:
\begin{align}
\label{eq:tuple-set:accept:1}
\begin{aligned}
    & \Pr\Bigl[\T_q^{\Fset}(\ell,m) = 1\Bigr] \\
    & = \Pr_{\substack{
        (\sq{\rv{I}}, \sq{\rv{X}}, \sq{\rv{\rnd}}, \sq{\rv{\pi}}) \sim \T_q(\ell,m)
    }}\Biggl[
        \bigwedge_{i \in [q-1]}
        \close{
            \Fset(\rv{X}_i; \rv{\rnd}_i, \rv{\pi}_i)|_{\rv{I}_i}}{
            \Fset(\rv{X}_{i+1}; \rv{\rnd}_{i+1}, \rv{\pi}_{i+1})|_{\rv{I}_i}}{
            m/\ell}
    \Biggr] \\
    & = \Pr_{\substack{
        (\sq{\rv{I}}, \sq{\rv{X}}, \sq{\rv{\rnd}}, \sq{\rv{\pi}}) \sim \T_q(\ell,m)
    }}\Biggl[
        \bigwedge_{i \in [q-1]}
        \close{
            \bigl(\Ftuple\bigl(\rv{X}_i \circ \rv{\pi}_i; \rv{\rnd}_i\bigr) \circ \rv{\pi}_i^{-1}\bigr)|_{\rv{I}_i}}{
            \bigl(\Ftuple\bigl(\rv{X}_{i+1} \circ \rv{\pi}_{i+1}; \rv{\rnd}_{i+1}\bigr) \circ \rv{\pi}_{i+1}^{-1}\bigr)|_{\rv{I}_i}}{
            m/\ell}
    \Biggr].
\end{aligned}
\end{align}
Observe that the following two distributions are equivalent:
\begin{itemize}
    \item
        $(\sq{\rv{x}}_1 \circ \rv{\pi}_1, \ldots, \sq{\rv{x}}_q \circ \rv{\pi}_q)$
        conditioned on the event that
        $\sq{\rv{x}}_i|_{\rv{I}_i} = \sq{\rv{x}}_{i+1}|_{\rv{I}_i}$
        for every $i \in [q-1]$ and
        $\sq{\rv{x}}_i \in V^{\ul{k}}$
        for every $i \in [q]$,
        where
        $\rv{I}_i \in \binom{[k]}{\ell}$ for every $i \in [q-1]$
        and
        $\sq{\rv{x}}_i \in V^k$ and $\rv{\pi}_i \in \sym_k$ for every $i \in [q]$.
    \item
        $(\rv{X}_1 \circ \rv{\pi}_1, \ldots, \rv{X}_q \circ \rv{\pi}_q)$,
        where
        $(\sq{\rv{I}}, \sq{\rv{X}}, \sq{\rv{\rnd}}, \sq{\rv{\pi}}) \sim \T_q(\ell,m)$.
\end{itemize}

\noindent
Therefore, we derive
\begin{align}
\label{eq:tuple-set:accept:2}
\begin{aligned}
    & \text{\cref{eq:tuple-set:accept:1}} \\
    & = \Pr_{\substack{
        \rv{I}_1, \ldots, \rv{I}_{q-1} \in \binom{[k]}{\ell} \\
        \sq{\rv{x}}_1, \ldots, \sq{\rv{x}}_q \in V^k \\
        \rv{\pi}_1, \ldots, \rv{\pi}_q \in \sym_k \\
        \rv{\rnd}_1, \ldots, \rv{\rnd}_q \in \Rnd
    }}\Biggl[
        \bigwedge_{i \in [q-1]}
        \evt_{\text{tuple}}(\rv{I}_i, \sq{\rv{x}}_i, \rv{\pi}_i, \rv{\rnd}_i, \sq{\rv{x}}_{i+1}, \rv{\pi}_{i+1}, \rv{\rnd}_{i+1})
        \Biggm|
        \bigwedge_{i \in [q]}
        \sq{\rv{x}}_i \in V^{\ul{k}}
        \text{ and }
        \bigwedge_{i \in [q-1]}
        \sq{\rv{x}}_i|_{\rv{I}_i} = \sq{\rv{x}}_{i+1}|_{\rv{I}_i}
    \Biggr] \\
    & \geq \Pr_{\substack{
        \rv{I}_1, \ldots, \rv{I}_{q-1} \in \binom{[k]}{\ell} \\
        \sq{\rv{x}}_1, \ldots, \sq{\rv{x}}_q \in V^k \\
        \rv{\pi}_1, \ldots, \rv{\pi}_q \in \sym_k \\
        \rv{\rnd}_1, \ldots, \rv{\rnd}_q \in \Rnd
    }}\Biggl[
        \bigwedge_{i \in [q-1]}
        \evt_{\text{tuple}}(\rv{I}_i, \sq{\rv{x}}_i, \rv{\pi}_i, \rv{\rnd}_i, \sq{\rv{x}}_{i+1}, \rv{\pi}_{i+1}, \rv{\rnd}_{i+1})
        \text{ and }
        \bigwedge_{i \in [q]}
        \sq{\rv{x}}_i \in V^{\ul{k}}
        \Biggm| 
        \bigwedge_{i \in [q-1]}
        \sq{\rv{x}}_i|_{\rv{I}_i} = \sq{\rv{x}}_{i+1}|_{\rv{I}_i}
    \Biggr] \\
    & \geq \underbrace{\Pr_{\substack{
        \rv{I}_1, \ldots, \rv{I}_{q-1} \in \binom{[k]}{\ell} \\
        \sq{\rv{x}}_1, \ldots, \sq{\rv{x}}_q \in V^k \\
        \rv{\pi}_1, \ldots, \rv{\pi}_q \in \sym_k \\
        \rv{\rnd}_1, \ldots, \rv{\rnd}_q \in \Rnd
    }}\Biggl[
        \bigwedge_{i \in [q-1]}
        \evt_{\text{tuple}}(\rv{I}_i, \sq{\rv{x}}_i, \rv{\pi}_i, \rv{\rnd}_i, \sq{\rv{x}}_{i+1}, \rv{\pi}_{i+1}, \rv{\rnd}_{i+1})
        \Biggm| 
        \bigwedge_{i \in [q-1]}
        \sq{\rv{x}}_i|_{\rv{I}_i} = \sq{\rv{x}}_{i+1}|_{\rv{I}_i}
    \Biggr]}_{\geq p} \\
    & \quad - \underbrace{\Pr_{\substack{
        \rv{I}_1, \ldots, \rv{I}_{q-1} \in \binom{[k]}{\ell} \\
        \sq{\rv{x}}_1, \ldots, \sq{\rv{x}}_q \in V^k \\
        \rv{\pi}_1, \ldots, \rv{\pi}_q \in \sym_k \\
        \rv{\rnd}_1, \ldots, \rv{\rnd}_q \in \Rnd
    }}\Biggl[
        \bigvee_{i \in [q]}
        \sq{\rv{x}}_i \notin V^{\ul{k}}
        \Biggm|
        \bigwedge_{i \in [q-1]}
        \sq{\rv{x}}_i|_{\rv{I}_i} = \sq{\rv{x}}_{i+1}|_{\rv{I}_i}
    \Biggr]}_{\leq \frac{qk^2}{|V|}} \\
    & \geq p - \frac{qk^2}{|V|},
\end{aligned}
\end{align}
where we used the inequality that
\begin{align}
\begin{aligned}
    \Pr_{\substack{
        \rv{I}_1, \ldots, \rv{I}_{q-1} \in \binom{[k]}{\ell} \\
        \sq{\rv{x}}_1, \ldots, \sq{\rv{x}}_q \in V^k
    }}\Biggl[
        \bigvee_{i \in [q]}
        \sq{\rv{x}}_i \notin V^{\ul{k}}
        \Biggm|
        \bigwedge_{i \in [q-1]}
        \sq{\rv{x}}_i|_{\rv{I}_i} = \sq{\rv{x}}_{i+1}|_{\rv{I}_i}
    \Biggr]
    & \leq
    \sum_{i \in [q]}
    \Pr_{\substack{
        \rv{I}_1, \ldots, \rv{I}_{q-1} \in \binom{[k]}{\ell} \\
        \sq{\rv{x}}_1, \ldots, \sq{\rv{x}}_q \in V^k
    }}\Biggl[
        \sq{\rv{x}}_i \notin V^{\ul{k}}
        \Biggm|
        \bigwedge_{i \in [q-1]}
        \sq{\rv{x}}_i|_{\rv{I}_i} = \sq{\rv{x}}_{i+1}|_{\rv{I}_i}
    \Biggr] \\
    & = q \Pr_{\sq{\rv{x}} \in V^k}\Bigl[
        \sq{\rv{x}} \notin V^{\ul{k}}
    \Bigr]
    \leq \frac{qk^2}{|V|},
\end{aligned}
\end{align}
as desired.
\end{proof}

\paragraph{Proof of \cref{clm:PSPACE:W:close}.}
Instead of directly proving \cref{clm:PSPACE:W:close}, we prove the following slight generalization.

\begin{lemma}
\label{lem:tuple-set:agree}
    Let $\Ftuple \colon V^k \times \Rnd \to \Sigma^k$ be a $k$-tuple function, and
    let $\Fset \colon \binom{V}{k} \times \Rnd \times \sym_k \to \Sigma^k$ be a $k$-set function such that
    $\Fset(X; \rnd, \pi) \defeq \Ftuple(X \circ \pi; \rnd) \circ \pi^{-1}$.
    Suppose that there exists
    a function $f \colon V \to \Sigma$ such that
    \begin{align}
        \Pr_{\substack{
            \rv{X} \in \binom{V}{k} \\
            (\rv{\rnd},\rv{\pi}) \in \Rnd \times \sym_k
        }}\left[
            \close{\Fset(\rv{X}; \rv{\rnd}, \rv{\pi})}{f^k(\rv{X})}{\delta}
        \right] \geq p
    \end{align}
    for some reals $p \in (0,1)$ and $\delta \in (0,1)$.
    Then,
    \begin{align}
        \Pr_{\substack{
            \sq{\rv{x}} \in V^k \\
            \rv{\rnd} \in \Rnd
        }}\left[
            \close{\Ftuple(\sq{\rv{x}}; \rv{\rnd})}{f^k(\sq{\rv{x}})}{\delta}
        \right]
        \geq p \left(1-\frac{k^2}{|V|}\right).
    \end{align}
\end{lemma}
\begin{proof} 
By assumption, we have
\begin{align}
\begin{aligned}
    p
    & \leq \Pr_{\substack{
        \rv{X} \in \binom{V}{k} \\
        (\rv{\rnd},\rv{\pi}) \in \Rnd \times \sym_k
    }}\left[
        \close{\Fset(\rv{X}; \rv{\rnd}, \rv{\pi})}{f^k(\rv{X})}{\delta}
    \right] \\
    & = \Pr_{\substack{
        \rv{X} \in \binom{V}{k} \\
        (\rv{\rnd},\rv{\pi}) \in \Rnd \times \sym_k
    }}\left[
        \close{\Ftuple(\rv{X} \circ \rv{\pi}; \rv{\rnd}) \circ \rv{\pi}^{-1}}{f^k(\rv{X})}{\delta}
    \right] \\
    & = \Pr_{\substack{
        \sq{\rv{x}} \in V^k \\
        \rv{\rnd} \in \Rnd
    }}\left[
        \close{\Ftuple(\sq{\rv{x}}; \rv{\rnd})}{f^k(\sq{\rv{x}})}{\delta}
        \Bigm|
        \sq{\rv{x}} \in V^{\ul{k}}
    \right],
\end{aligned}
\end{align}
where the last equality holds because
$\rv{X} \circ \rv{\pi}$ is uniformity distributed over $V^{\ul{k}}$.
Therefore, we derive
\begin{align}
\begin{aligned}
    \Pr_{\substack{
        \sq{\rv{x}} \in V^k \\
        \rv{\rnd} \in \Rnd
    }}\left[
        \close{\Ftuple(\sq{\rv{x}}; \rv{\rnd})}{f^k(\sq{\rv{x}})}{\delta}
    \right]
    & \geq \Pr_{\substack{
        \sq{\rv{x}} \in V^k \\
        \rv{\rnd} \in \Rnd
    }}\left[
        \close{\Ftuple(\sq{\rv{x}}; \rv{\rnd})}{f^k(\sq{\rv{x}})}{\delta}
        \text{ and }
        \sq{\rv{x}} \in V^{\ul{k}}
    \right] \\
    & = \underbrace{\Pr_{\substack{
        \sq{\rv{x}} \in V^k \\
        \rv{\rnd} \in \Rnd
    }}\left[
        \close{\Ftuple(\sq{\rv{x}}; \rv{\rnd})}{f^k(\sq{\rv{x}})}{\delta}
        \;\middle|\;
        \sq{\rv{x}} \in V^{\ul{k}}
    \right]}_{\geq p}
    \cdot \underbrace{\Pr_{\sq{\rv{x}} \in V^k}\Bigl[
        \sq{\rv{x}} \in V^{\ul{k}}
    \Bigr]}_{\geq 1 - \frac{k^2}{|V|}} \\
    & \geq p \left(1-\frac{k^2}{|V|}\right),
\end{aligned}
\end{align}
as desired.
\end{proof}

\end{document}